\documentclass[hyperref]{lmcs}
\pdfoutput=1

\usepackage{lastpage}
\renewcommand{\dOi}{14(3:16)2018}
\lmcsheading{}{1--\pageref{LastPage}}{}{}%
{Oct.~12,~2017}{Sep.~05,~2018}{}

\usepackage[utf8]{inputenc}
\usepackage{amsmath,amssymb}
\usepackage{mathtools}
\usepackage{mathrsfs}
\usepackage{url}
\usepackage{wrapfig}
\usepackage{stmaryrd}

\usepackage{multicol}
\usepackage{haskell}

\usepackage{tikz}
\usetikzlibrary{positioning,arrows,decorations.pathmorphing}
\tikzset{node distance=15mm, auto}

\let\temp\phi
\let\phi\varphi
\let\varphi\temp

\let\temp\epsilon
\let\epsilon\varepsilon
\let\varepsilon\temp

\let\temp\theta
\let\theta\vartheta
\let\vartheta\temp

\newif\ifdraft
\drafttrue

\ifdraft
  \usepackage{color}
  \usepackage{soul}

\else

\fi

\newcommand{\ie}{\emph{i.e.}}
\newcommand{\eg}{\emph{e.g.}}

\DeclareMathOperator{\Total}{Total}
\DeclareMathOperator{\Inv}{Inv}
\DeclareMathOperator{\Tr}{Tr}

\newcommand{\PInj}{\ensuremath{\mathbf{PInj}}}
\newcommand{\PHom}{\ensuremath{\mathbf{PHom}}}
\newcommand{\Pfn}{\ensuremath{\mathbf{Pfn}}}

\newcommand{\id}{\ensuremath{\text{id}}}
\newcommand{\zero}{\ensuremath{\text{0}}}

\newcommand{\C}{\ensuremath{\mathscr{C}}}
\newcommand{\D}{\ensuremath{\mathscr{D}}}
\newcommand{\I}{\ensuremath{\mathcal{I}}}

\newcommand{\comp}{}

\newcommand{\tot}{\xrightarrow}

\newcommand{\ridm}[1]{\ensuremath{\overline{#1}}}
\newcommand{\dec}[1]{\ensuremath{\langle #1 \rangle}}

\newcommand{\codiag}{\ensuremath{\nabla}}
\renewcommand{\dag}{\ensuremath{\dagger}}

\newcommand{\oxf}[1]{\left\llbracket #1 \right\rrbracket}

\newcommand{\bigstepc}[3]{\ensuremath{#1 \vdash #2 \downarrow #3}}
\newcommand{\bigstepp}[3]{\ensuremath{#1 \vdash #2 \rightsquigarrow #3}}

\newcommand{\stt}{\ensuremath{\mathit{tt}}}
\newcommand{\sff}{\ensuremath{\mathit{ff}}}
\newcommand{\snot}[1]{\ensuremath{\mathbf{not}~#1}}
\newcommand{\sand}[2]{\ensuremath{#1~\mathbf{and}~#2}}
\newcommand{\sskip}{\ensuremath{\mathbf{skip}}}
\newcommand{\sseq}[2]{\ensuremath{#1~\mathtt{;}~#2}}
\newcommand{\sif}[4]{\ensuremath{\mathbf{if}~#1~\mathbf{then}~#2~\mathbf{else}~#3~\mathbf{fi}~#4}}
\newcommand{\sfrom}[3]{\ensuremath{\mathbf{from}~#1~\mathbf{loop}~#2~\mathbf{until}~#3}}
\newcommand{\sloop}[3]{\ensuremath{\mathit{\underline{loop}}[#1,#2,#3]}}

\newlength{\fboxsepold}
\newcommand{\fboxtight}[1]{\setlength{\fboxsep}{3pt}%
\fbox{#1}%
\setlength{\fboxsep}{\fboxsepold}}

\def\eg{{\em e.g.}}

\begin{document}

\title[Categorical foundations for reversible flowchart languages]{A categorical foundation for structured reversible flowchart languages: Soundness and adequacy}

\author[Glück and Kaarsgaard]{Robert Glück}	
\address{DIKU, Department of Computer Science, University of Copenhagen, Denmark}	
\email{glueck@acm.org}  

\author[]{Robin Kaarsgaard}	
\address{\vskip-7 pt} 
\email{robin@di.ku.dk}  

\thanks{The authors acknowledge the support given by \emph{COST Action IC1405
Reversible computation: Extending horizons of computing.} We also thank the
anonymous reviewers of MFPS XXXIII for their thoughtful and detailed comments on a previous version of this paper.}


\titlecomment{This is the extended version of an article presented at MFPS XXXIII~\cite{Glueck2017}, extended with proofs that previously appeared in the appendix, as well as new sections on soundness, adequacy, and full abstraction.}


\begin{abstract}
  \noindent
  Structured reversible flowchart languages is a class of imperative reversible
  programming languages allowing for a simple diagrammatic representation of
  control flow built from a limited set of control flow structures. This class
  includes the reversible programming language Janus (without recursion), as
  well as more recently developed reversible programming languages such as
  \texttt{R-CORE} and \texttt{R-WHILE}.
  
  In the present paper, we develop a categorical foundation for this class of
  languages based on inverse categories with joins. We generalize the notion of
  extensivity of restriction categories to one that may be accommodated by
  inverse categories, and use the resulting \emph{decisions} to give a
  reversible representation of predicates and assertions. This leads to a
  categorical semantics for structured reversible flowcharts, which we show to
  be computationally sound and adequate, as well as equationally fully abstract 
  with respect to the operational semantics under certain conditions.
\end{abstract}

\maketitle


\section{Introduction} 
\label{sec:introduction} 
Reversible computing is an emerging paradigm that adopts a physical principle
of reality into a \emph{computation model without information erasure}.
Reversible computing extends the standard forward-only mode of computation with
the ability to execute in reverse as easily as forward. Reversible computing is
a necessity in the context of quantum computing and some bio-inspired
computation models. Regardless of the physical motivation, bidirectional
determinism is interesting in its own right. The potential benefits include the
design of innovative reversible architectures~(\eg,
\cite{ThGlAx10,ToAG12,YoAG12}), new programming models and techniques~(\eg,
\cite{YoGl07a,GlKa03,Moge11}), and the enhancement of software with
reversibility (\eg, \cite{CaPF99}).

Today the semantics of reversible programming languages   
are usually formalized using
traditional metalanguages, such as structural operational semantics or denotational semantics based on complete partial orders (\eg, \cite{GlYo16,YoGl07a}). However, these metalanguages 
were introduced for
the definition of conventional (irreversible) programming languages. The
fundamental properties of a reversible language, such as the required backward determinism and the invertibility of
object language programs, are not naturally captured by these metalanguages. As a result, these properties are to be shown for each semantic definition. This unsatisfying state of 
affairs is not surprising, however, as these properties have been rarely required for more than half a century of mainstream language development. The recent advances in the area of reversible computing have changed this situation.

This paper 
provides a new categorical foundation specifically for
formalizing reversible programming languages, in particular the semantics of
reversible structured flowchart languages~\cite{YoAG08b,YoAxGl:16:TCS}, which are the
reversible counterpart of the structured high-level programming languages used today. This
formalization is based on join inverse categories with a developed notion of
\emph{extensivity} for inverse categories, which gives rise to natural
representations of predicates and assertions, and consequently to models of
reversible structured flowcharts. 
It provides a framework for
modelling reversible languages, such that the reversible semantic properties of the
object language are naturally ensured by the metalanguage.

The semantic framework we are going to present in this paper covers the
reversible structured languages regardless of their concrete formation, such as
atomic operations, elementary predicates, and value domains. State of the art reversible
programming languages that are 
concretizations
of this computation model include the
well-known imperative language Janus~\cite{YoGl07a} without recursion, the while
languages \texttt{R-WHILE} and \texttt{R-CORE} with dynamic data
structures~\cite{GlYo16,GlYo17}, and the structured reversible language SRL with stacks and arrays~\cite{YoAxGl:16:TCS}. Structured control-flow is also a defining element of reversible object-oriented languages~\cite{HMG:16}.
Further, unstructured reversible flowchart
languages, such as reversible assembly languages with
jumps~\cite{Fran99,AxGY07a} and the unstructured reversible language RL~\cite{YoAxGl:16:TCS}, can be transformed into structured ones
thanks to the structured reversible program theorem~\cite{YoAG08b}.

The main contribution of this paper is to provide a metalanguage formalism based on join inverse categories that is geared to formalize the reversible flowchart languages. The languages formalized in this framework are ensured to have the defining reversible language properties, including backward determinism and local invertibility.
Another main property of the formalism is that every reversible structured language that is syntactically closed under inversion of its elementary operations is also closed under inversion of reversible control-flow operators. This is particularly useful, as it is sufficient to check this property for elementary constructs to ensure the correctness of the associated program inverter.
Key to our formalism are decisions, which provide a particularly advantageous reversible representation of predicates and their corresponding assertions. This insight may guide the design of reversible constructs that are often 
quite involved to model in a reversible setting, such as pattern matching.

The results in this paper are illustrated by introducing a family of small reversible flowchart languages $\mathtt{RINT}_k$ for reversible computing with integer data, a reversible counterpart of the family of classic counter languages used for theoretical investigations into irreversible languages. The family introduced here may well serve the same purpose in a reversible context.


\emph{Overview:}\enspace In Section~\ref{sec:reversible_structured_flowcharts},
we give an introduction to structured reversible flowchart languages, while
Section~\ref{sec:background} describes the restriction and inverse category
theory used as backdrop in later sections. In
Section~\ref{sec:extensivity_inverse_categories}, we warm up by developing a
notion of extensivity for inverse categories, based on extensive restriction
categories and its associated concept of \emph{decisions}. Then, in
Section~\ref{sec:modelling_reversible_flowcharts}, we put it all to use by
showing how decisions may be used to model predicates and ultimately also
reversible flowcharts, and we show that these are computationally sound and adequate with
respect to the operational semantics in
Section~\ref{sec:soundness_and_adequacy}. In
Section~\ref{sec:full_abstraction}, we extend the previous theorems by
giving a sufficient condition for equational full abstraction. In
Section~\ref{sec:applications}, we show how to verify program inversion using
the categorical semantics, develop a small language to exemplify our framework,
and discuss other applications in reversible programming.
Section~\ref{sec:conclusion} offers some concluding remarks.



\section{Reversible structured flowcharts} 
\label{sec:reversible_structured_flowcharts}
Structured reversible flowcharts~\cite{YoAG08b,YoAxGl:16:TCS} naturally model the control flow behavior of
reversible (imperative) programming languages in a simple diagrammatic
representation, 
as classical flowcharts do for conventional
languages. A crucial difference is that atomic steps are limited to
\emph{partial injective functions} and they require an additional
\emph{assertion}, an explicit orthogonalizing condition, at join points in the
control flow.

A structured reversible flowchart $F$ is built from four blocks
(Figure~\ref{fig:blocks}): An \emph{atomic step} that performs an elementary
operation on a domain $X$ specified by a partial injective function $a:
X\rightharpoonup X$; a \emph{while loop} over a block $B$ with entry assertion
$p_1 : X \rightarrow \mathit{Bool}$ and exit test $p_2 : X \rightarrow
\mathit{Bool}$; a \emph{selection} of block $B_1$ or $B_2$ with entry test $p_1
: X \rightarrow \mathit{Bool}$ and exit assertion $p_2 : X \rightarrow
\mathit{Bool}$; and a \emph{sequence} of blocks $B_1$ and $B_2$.

\begin{figure}[t]
 \setlength{\unitlength}{1.9mm}
 \small
 \begin{minipage}{.2\textwidth}
  \centering
  \input{Step.pic}
 \end{minipage}
 \begin{minipage}{.25\textwidth}
  \centering
  \input{WhileLoop.pic}
 \end{minipage}
 \begin{minipage}{.25\textwidth}
  \centering
   \input{Selection.pic}
 \end{minipage}
 \begin{minipage}{.25\textwidth}
  \centering
  \input{Sequence.pic}
 \end{minipage}
 \\[2.6ex]
 \begin{minipage}{.2\textwidth}
  \centering
   (a) Step
 \end{minipage}
 \begin{minipage}{.25\textwidth}
 \centering
   (b) While loop
 \end{minipage}
 \begin{minipage}{.25\textwidth}
  \centering
   (c) Selection
 \end{minipage}
 \begin{minipage}{.25\textwidth}
  \centering
   (d) Sequence
 \end{minipage}
 \caption{Structured reversible flowcharts.}
 \label{fig:blocks}
\end{figure}

A structured reversible flowchart $F$ consists of one main block. Blocks have
unique entry and exit points, and can be nested any number of times to form
more complex flowcharts. The interpretation of $F$ consists of a given domain
$X$ (typically of stores or states, which we shall denote by $\sigma$) and a
finite set of partial injective functions $a$ and predicates $p: X \rightarrow
\mathit{Bool}$. Computation starts at the entry point of $F$ in an initial
$x_0$ (the input), proceeds sequentially through the edges of $F$, and ends at
the exit point of $F$ in a final $x_n$ (the output), if $F$ is defined on the
given input. Though the specific set of predicates depend on the flowchart
language, they are often (as we will do here) assumed to be closed under
Boolean operators, in particular conjunction and negation. The operational
semantics for these, shown in Figure~\ref{fig:semantics_pred}, are the same as
in the irreversible case (see, \eg, \cite{Winskel1993}): We use the judgment
form $\bigstepp{\sigma}{p}{b}$ here to mean that the predicate $p$ evaluated on
the state $\sigma$ results in the Boolean value $b$.
\begin{figure}
\fboxtight{$\bigstepp{\sigma}{p}{b}$}
\begin{equation*}
  \begin{gathered}
    \frac{}{\bigstepp{\sigma}{\stt}{\stt}} \\[0.5\baselineskip]
    \frac{\bigstepp{\sigma}{p}{\stt}}{\bigstepp{\sigma}{\snot{p}}{\sff}} 
    \\[0.5\baselineskip]
    \frac{\bigstepp{\sigma}{p}{\stt} \qquad 
    \bigstepp{\sigma}{q}{\stt}}{\bigstepp{\sigma}{\sand{p}{q}}{\stt}} 
    \\[0.5\baselineskip]
    \frac{\bigstepp{\sigma}{p}{\sff} \qquad 
    \bigstepp{\sigma}{q}{\stt}}{\bigstepp{\sigma}{\sand{p}{q}}{\sff}} 
  \end{gathered}
  \qquad \enspace
  \begin{gathered}
    \frac{}{\bigstepp{\sigma}{\sff}{\sff}} \\[0.5\baselineskip]
    \frac{\bigstepp{\sigma}{p}{\sff}}{\bigstepp{\sigma}{\snot{p}}{\stt}} 
    \\[0.5\baselineskip]
    \frac{\bigstepp{\sigma}{p}{\stt} \qquad 
    \bigstepp{\sigma}{q}{\sff}}{\bigstepp{\sigma}{\sand{p}{q}}{\sff}}
    \\[0.5\baselineskip]
    \frac{\bigstepp{\sigma}{p}{\sff} \qquad 
    \bigstepp{\sigma}{q}{\sff}}{\bigstepp{\sigma}{\sand{p}{q}}{\sff}} 
  \end{gathered}
\end{equation*}
\caption{Operational semantics for Boolean predicates.}
\label{fig:semantics_pred}
\end{figure}

The assertion $p_1$ in a reversible while loop (marked by a circle, as introduced in \cite{YoGl07a}) is a new
flowchart operator: the predicate $p_1$ must be \emph{true} when the control
flow reaches the assertion along the $\mathbf{t}$-edge, and \emph{false} when
it reaches the assertion along the $\mathbf{f}$-edge; otherwise,
the loop is undefined. The test $p_2$ (marked by a diamond) has the usual
semantics. This means that $B$ in a loop is repeated as long as $p_1$ and
$p_2$ are \emph{false}.

The selection has an assertion $p_2$, which must be \emph{true} when the
control flow reaches the assertion from $B_1$, and \emph{false} when the
control flow reaches the assertion from $B_2$; otherwise, the selection is
undefined. As usual, the test $p_1$ selects $B_1$ or $B_2$. The assertion makes
the selection reversible.

Despite their simplicity, reversible structured flowcharts are \emph{reversibly
universal}~\cite{AxGl11b}, which means that they are computationally as
powerful as any reversible programming language can be. Given a suitable domain
$X$ for finite sets of atomic operations and predicates, there exists, for
every injective computable function $f : X \rightarrow Y$, a reversible
flowchart $F$ that computes $f$.

Reversible structured flowcharts (Figure~\ref{fig:blocks}) have a straightforward representation as program texts defined by the grammar 
$$
\begin{array}[t]{lclclclcl} 
B &::=& a 
& | & \sfrom{p}{B}{p} 
& | & \sif{p}{B}{B}{p}
& | & \sseq{B}{B} 
\end{array}
$$
It is often assumed, as we will do here, that the set of atomic steps contains
a step $\sskip$ that acts as the identity. The operational semantics for these
flowchart structures (or simply \emph{commands}) are shown in
Figure~\ref{fig:semantics_cmd}. Here, the judgment form
$\bigstepc{\sigma}{c}{\sigma'}$ is used and taken to mean that the command $c$
converges in the state $\sigma$ resulting in a new state $\sigma'$. Note the
use of the meta-command $\sloop{p}{c}{q}$: This does not represent any
reversible flowchart structure (and is thus not a syntactically valid command),
but is rather a piece of internal syntax used to give meaning to loops. This is
required since the role of the entry assertion in a reversible while loop
changes after the first iteration: When entering the loop from the outside
(\ie, before any iterations have occured), the entry assertion must be true,
but when entering from the inside (\ie, one or more iterations), the entry
assertion must be false.


\begin{figure}
\fboxtight{$\bigstepc{\sigma}{c}{\sigma'}$}
\begin{gather*}
  \begin{gathered}
    \frac{\phantom{\vdash\downarrow}}{\bigstepc{\sigma}{\sskip}{\sigma}} 
    \\[0.5\baselineskip]
    \frac{
      \bigstepp{\sigma}{p}{\stt} \quad\enspace
      \bigstepc{\sigma}{c_1}{\sigma'} \quad\enspace
      \bigstepp{\sigma'}{q}{\stt}
    }{
      \bigstepc{\sigma}{\sif{p}{c_1}{c_2}{q}}{\sigma'}
    }
    \\[0.5\baselineskip]
    \frac{
      \bigstepp{\sigma}{p}{\stt} \quad\enspace
      \bigstepp{\sigma}{q}{\stt}
    }{
      \bigstepc{\sigma}{\sfrom{p}{c}{q}}{\sigma}
    }
  \end{gathered}
  \qquad \enspace
  \begin{gathered}
    \frac{
      \bigstepc{\sigma}{c_1}{\sigma'} \qquad 
      \bigstepc{\sigma'}{c_2}{\sigma''}
    }{
      \bigstepc{\sigma}{\sseq{c_1}{c_2}}{\sigma''}
    } \\[0.5\baselineskip]
    \frac{
      \bigstepp{\sigma}{p}{\sff} \quad\enspace
      \bigstepc{\sigma}{c_2}{\sigma'} \quad\enspace
      \bigstepp{\sigma'}{q}{\sff}
    }{
      \bigstepc{\sigma}{\sif{p}{c_1}{c_2}{q}}{\sigma'}
    }
    \\[0.5\baselineskip]
    \frac{
      \bigstepp{\sigma}{p}{\sff} \quad\enspace
      \bigstepp{\sigma}{q}{\stt}
    }{
      \bigstepc{\sigma}{\sloop{p}{c}{q}}{\sigma}
    }
  \end{gathered}
  \\[0.5\baselineskip]
  \frac{
    \bigstepp{\sigma}{p}{\sff} \quad\enspace
    \bigstepp{\sigma}{q}{\sff} \quad\enspace
    \bigstepc{\sigma}{c}{\sigma'} \quad\enspace
    \bigstepc{\sigma'}{\sloop{p}{c}{q}}{\sigma''}
  }{
    \bigstepc{\sigma}{\sloop{p}{c}{q}}{\sigma''}
  }
  \\[0.5\baselineskip]
  \frac{
    \bigstepp{\sigma}{p}{\stt} \quad\enspace
    \bigstepp{\sigma}{q}{\sff} \quad\enspace
    \bigstepc{\sigma}{c}{\sigma'} \quad\enspace
    \bigstepc{\sigma'}{\sloop{p}{c}{q}}{\sigma''}
  }{
    \bigstepc{\sigma}{\sfrom{p}{c}{q}}{\sigma''}
  }
\end{gather*}
\caption{Operational semantics for the reversible flowchart structures.}
\label{fig:semantics_cmd}
\end{figure}

The reversible structured flowcharts defined above corresponds to the reversible language \texttt{R-WHILE}~\cite{GlYo16}, but their value domain, atomic functions and predicates are unspecified. As a minimum, a reversible flowchart needs blocks (a), (b), and (d) from Figure~\ref{fig:blocks}, because selection can be simulated by combining while loops that conditionally skip the body block or execute it once. \texttt{R-CORE}~\cite{GlYo17} is an example of such a minimal language.


\section{Restriction and inverse categories} 
\label{sec:background}
The following section contains the background on restriction and inverse
category theory necessary for our later developments. Unless otherwise
specified, the definitions and results presented in this section can be found
in introductory texts on the subject (\eg,
\cite{Giles2014,Guo2012,Cockett2002,Cockett2003,Cockett2007}).

Restriction categories~\cite{Cockett2002,Cockett2003,Cockett2007} axiomatize
categories of partial maps. This is done by assigning to each morphism $f$ a
\emph{restriction idempotent} $\ridm{f}$, which we think of as a partial
identity defined precisely where $f$ is. Formally, restriction categories are
defined as follows.
\begin{defi}\label{def:rest_cat}
  A \emph{restriction category} is a category \C{} equipped with a combinator
  mapping each morphism $A \tot{f} B$ to a morphism $A \tot{\ridm{f}} A$
  satisfying
  \begin{multicols}{2}
    \begin{enumerate}[(i),ref={\roman*}]
      \item $f \comp \ridm{f} = f$,
      \item $\ridm{g} \comp \ridm{f} = \ridm{f} \comp \ridm{g}$,
      \item $\ridm{f \comp \ridm{g}} = \ridm{f} \comp \ridm{g}$, and
      \item $\ridm{g} \comp f = f \comp \ridm{g \comp f}$
    \end{enumerate}
  \end{multicols}
  \noindent
  for all suitable $g$.
\end{defi}
As an example, the category \Pfn{} of sets and partial functions is a
restriction category, with $\ridm{f}(x) = x$ if $f$ is defined at $x$, and
undefined otherwise. Note that being a restriction category is a structure, not
a property; a category may be a restriction category in several different ways
(\eg, assigning $\ridm{f} = \id$ for each morphism $f$ gives a trivial
restriction structure to any category).

In restriction categories, we say that a morphism $A \tot{f} B$ is \emph{total}
if $\ridm{f} = \id_A$, and a \emph{partial isomorphism} if there exists a
(necessarily unique) \emph{partial inverse} $B \tot{f^\dag} A$ such that
$f^\dag \comp f = \ridm{f}$ and $f \comp f^\dag = \ridm{f^\dag}$. Isomorphisms
are then simply the total partial isomorphisms with total partial inverses. An
inverse category can then be defined as a special kind of restriction
category\footnote{This is a rather modern definition due to \cite{Cockett2002}.
Originally, inverse categories were defined as the categorical extensions of
inverse semigroups; see \cite{Kastl1979}.}.
\begin{defi}
  An \emph{inverse category} is a restriction category where each morphism is a 
  partial isomorphism.
\end{defi}
Every restriction category \C{} gives rise to an inverse category $\Inv(\C)$ (the \emph{cofree} inverse category of \C{}, see \cite{Kaarsgaard2017}),
which has as objects all objects of \C, and as morphisms all of the partial
isomorphisms of \C. As such, since partial isomorphisms in \Pfn{} are partial
injective functions, a canonical example of an inverse category is the category
$\Inv(\Pfn) \cong \PInj$ of sets and partial injective functions.

Since each morphism in an inverse category has a unique partial inverse, as
also suggested by our notation this makes inverse categories canonically
\emph{dagger categories}~\cite{Selinger2007}, in the sense that they come
equipped with a contravariant endofunctor $(-)^\dag$ satisfying
$f=f^{\dag\dag}$ and $\id_A^\dag = \id_A$ for each morphism $f$ and object $A$.

Given two restriction categories $\C$ and $\D$, the well-behaved functors
between them are \emph{restriction functors}, \ie, functors $F$ satisfying
$F(\ridm{f}) = \ridm{F(f)}$. Analogous to how regular semigroup homomorphisms
preserve partial inverses in inverse semigroups, when $\C$ and $\D$ are inverse
categories, all functors between them are restriction functors; specifically
they preserve the canonical dagger, \ie, $F(f^\dag) = F(f)^\dag$.

Before we move on, we show some basic facts about restriction idempotents in
restriction (and inverse) categories that will become useful later (see also, \eg, \cite{Cockett2002,Giles2014}).
\begin{lem}\label{lem:basic_ridm}
  In any restriction category, it is the case that
  \begin{enumerate}[(i), ref={\roman*}]
    \item $\ridm{\id} = \id$, \label{br1}
    \item $\ridm{\ridm{f}} = \ridm{f}$, \label{br2}
    \item $\ridm{g \comp f} \comp~ \ridm{f} = \ridm{g \comp f}$, \label{br3}
    \item $\ridm{\ridm{g} \comp f} = \ridm{g \comp f}$, \label{br4}
    \item if $g$ is total then $\ridm{g \comp f} = \ridm{f}$, and \label{br5}
    \item if $g \comp f$ is total then so is $f$. \label{br6}
  \end{enumerate}
  for all morphisms $X \tot{f} Y$ and $Y \tot{g} Z$.
\end{lem}
\begin{proof}
  For \eqref{br1} $\id = \id\comp~\ridm{\id} = \ridm{\id}$, and \eqref{br2} by
  $\ridm{\ridm{f}} = \ridm{\id \comp \ridm{f}} = \ridm{\id} \comp~ \ridm{f}
  = \id \comp \ridm{f} = \ridm{f}$ using \eqref{br1}. For \eqref{br3} we have
  $\ridm{g \comp f}\comp~\ridm{f} = \ridm{g \comp f\comp\ridm{f}} = \ridm{g
  \comp f}$, while \eqref{br4} follows by $\ridm{\ridm{g} \comp f} = 
  \ridm{f\comp \ridm{g \comp f}} = \ridm{f}\comp~\ridm{g \comp f} = \ridm{g 
  \comp f}\comp~\ridm{f} = \ridm{g \comp f}$ using \eqref{br3}. A special case
  of \eqref{br4} is \eqref{br5} since $g$ total means that $\ridm{g} =
  \id$, and so $\ridm{g \comp f} = \ridm{\ridm{g} \comp f} = \ridm{\id \comp f}
  = \ridm{f}$. Finally, \eqref{br6} follows by $g \comp f$ total means that 
  $\ridm{g \comp f} = \id$, so by this and \eqref{br3}, $\id = \ridm{g \comp
  f} = \ridm{g \comp f} ~\comp \ridm{f} = \id \comp \ridm{f} = \ridm{f}$, so
  $f$ is total as well.
\end{proof}

\subsection{Partial order enrichment and joins} 
\label{sub:partial_order_enrichment_and_joins}
A consequence of how restriction (and inverse) categories are defined is that
hom sets $\C(A,B)$ may be equipped with a partial order given by $f \le g$ iff
$g \comp \ridm{f} = f$ (this extends to an enrichment in the category of
partial orders and monotone functions). Intuitively, this states that $f$ is
below $g$ iff $g$ behaves exactly like $f$ when restricted to the points where
$f$ is defined. Notice that any morphism $e$ below an identity $\id_X$ is a
restriction idempotent under this definition, since $e \le \id_X$ iff
$\id_X\comp\ridm{e} = e$, \ie, iff $\ridm{e} = e$.

A sufficient condition for each $\C(A,B)$ to have a least
element is that $\C$ has a \emph{restriction zero}; a zero object $0$
in the usual sense which additionally satisfies $A \tot{0_{A,A}} A = A
\tot{\ridm{0_{A,A}}} A$ for each endo-zero map $0_{A,A}$.
One may now wonder when $\C(A,B)$ has joins as a partial order. Unfortunately,
$\C(A,B)$ has joins of all morphisms only in very degenerate cases. However, if
instead of considering arbitrary joins we consider joins of maps that are
somehow compatible, this becomes much more viable.
\begin{defi}
  In a restriction category, say that morphisms $X \tot{f} Y$ and $X \tot{g} Y$ 
  are
  \emph{disjoint} iff $f \comp \ridm{g} = 0$; and \emph{compatible} iff $f
  \comp \ridm{g} = g \comp \ridm{f}$.
\end{defi}
It can be shown that disjointness implies compatibility, as disjointness is
expectedly symmetric. Further, we may extend this to say that a set of parallel
morphisms is disjoint iff each pair of morphisms is disjoint, and likewise for
compatibility. This gives suitable notions of \emph{join restriction
categories}.
\begin{defi}\label{def:join_restriction_cat}
\nocite{Guo2012}
A restriction category $\mathscr{C}$ has compatible (disjoint) joins if it has a restriction zero, and satisfies 
that for each compatible (disjoint) subset $S$ of any hom
set $\C(A,B)$, there exists a morphism $\bigvee_{s \in S} s$ such that
\begin{enumerate}[(i), ref={\roman*}]
  \item $s \le \bigvee_{s \in S} s$ for all $s \in S$, and $s \le t$
  for all $s \in S$ implies $\bigvee_{s \in S} s \le t$; 
  \label{def:join_rest_sup}
  \item $\ridm{\bigvee_{s \in S} s} = \bigvee_{s \in S} \ridm{s}$;
  \label{def:join_rest_cont}
  \item $f \comp \left(\bigvee_{s \in S} s \right) = 
  \bigvee_{s \in S}(f \comp s)$ for all $f : B \to X$; and
  \item $\left(\bigvee_{s \in S} s \right) \comp g = \bigvee_{s \in 
  S}(s \comp g)$ for all $g : Y \to A$.
  \label{def:join_rest_compr}
\end{enumerate}
\end{defi}
For inverse categories, the situation is a bit more tricky, as the join of two compatible partial isomorphisms may not be a partial isomorphism. To ensure this, we need stronger relations:
\begin{defi}\label{def:join_inv_cat}
  In an inverse category, say that parallel maps $f$ and $g$ are
  \emph{disjoint} iff $f \comp \ridm{g} = 0$ and $f^\dag \comp \ridm{g^\dag} =
  0$; and \emph{compatible} iff $f \comp \ridm{g} = g \comp \ridm{f}$ and
  $f^\dag \comp \ridm{g^\dag} = g^\dag \comp \ridm{f^\dag}$.
\end{defi}
We may now extend this to notions of disjoint sets and compatible sets of
morphisms in inverse categories as before. This finally gives notions of \emph{join inverse categories}:
\begin{defi}
  An inverse category \C{} has compatible (disjoint) joins if it
  has a restriction zero and satisfies that for all compatible
  (disjoint) subsets $S$ of all hom sets $\C(A,B)$, there exists a morphism
  $\vee_{s \in S} s$ satisfying \eqref{def:join_rest_sup} --
  \eqref{def:join_rest_compr} of Definition~\ref{def:join_restriction_cat}.
\end{defi}
An example of a join inverse category is \PInj{} of sets and partial injective functions. Here, $\bigvee_{f \in S} f$ for a set of compatible partial injective functions $S$ is constructed as the function with graph the union of the graphs of all $f \in S$; or equivalently as the function
\begin{equation*}
  \left(\bigvee_{f \in S} f\right)(x) = \left\{ \begin{array}{ll}
    g(x) & \text{if there exists $g \in S$ such that $g(x)$ is defined} \\
    \text{undefined} \quad & \text{otherwise}
  \end{array}  \right.
\end{equation*}
Similarly, the category \PHom{} of topological spaces and partial
homeomorphisms (\ie, partial injective functions which are open and continuous
in their domain of definition) is an inverse category with joins constructed in
the same way. Further, for any join restriction category \C, the cofree inverse
category $\Inv(\C)$ of \C{} is a join inverse category as well~\cite[Lemma
3.1.27]{Guo2012}, and joins on both restriction and inverse categories can be
freely adjoined~\cite[Sec.~3.1]{Guo2012}.
A functor $F$ between restriction (or inverse) categories with joins is said to
be join-preserving when $F(\bigvee_{s \in S} s) = \bigvee_{s \in S} F(s)$.

\subsection{Restriction coproducts, extensivity, and related concepts} 
\label{sub:extensivity_of_restriction_categories}
While a restriction category may very well have coproducts, these are
ultimately only well-behaved when the coproduct injections $A \tot{\kappa_1}
A+B$ and $B \tot{\kappa_2} A+B$ are total; if this is the case, we say that the
restriction category has \emph{restriction coproducts}. If a restriction
category has all finite restriction coproducts, it also has a restriction zero
serving as unit. Note that in restriction categories with restriction
coproducts, the coproduct injections $\kappa_1$ (respectively $\kappa_2$) are
partial isomorphisms with partial inverse $\kappa_1^\dag = A+B \tot{[\id_A,
0_{B,A}]} A$ (respectively $\kappa_2^\dag = A+B \tot{[0_{A,B}, \id_B]} A$).

In \cite{Cockett2007}, it is shown that the existence of certain maps, called
\emph{decisions}, in a restriction category \C{} with restriction coproducts
leads to the subcategory $\Total(\C)$ of total maps being extensive (in the
sense of, \eg, \cite{Carboni1993}). This leads to the definition of an
\emph{extensive restriction category}\footnote{The name is admittedly mildly confusing, as an extensive restriction category is not extensive in the usual sense. Nevertheless, we stay with the established terminology.}.
\begin{defi}
  A restriction category is said to be \emph{extensive} (as a restriction
  category) if it has restriction coproducts and a restriction zero, and for
  each map $A \tot{f} B+C$ there is a unique \emph{decision} $A \tot{\dec{f}}
  A+A$ satisfying
  \begin{multicols}{2}
  \begin{description}
    \item[\enspace(D.1)] $\codiag \comp \dec{f} = \ridm{f}$ and
    \item[\enspace(D.2)] $(f + f) \comp \dec{f} = (\kappa_1 + \kappa_2) \comp 
    f$.
  \end{description}
  \end{multicols}
\end{defi}
In the above, $\codiag$ denotes the codiagonal $[\id,\id]$. A consequence of
these axioms is that each decision is a partial isomorphism; one can show that
$\dec{f}$ must be partial inverse to $[\ridm{\kappa_1^\dag \comp
f},\ridm{\kappa_2^\dag \comp f}]$ (see \cite{Cockett2007}). Further, when a
restriction category with restriction coproducts has finite joins, it is also
extensive with $\dec{f} = (\kappa_1 \comp \ridm{\kappa_1^\dag \comp f}) \vee
(\kappa_2 \comp \ridm{\kappa_2^\dag \comp f})$. As an example, \Pfn{} is extensive with $A \tot{\dec{f}} A+A$ for $A \tot{f} B+C$ given by
\begin{equation*}
  \dec{f}(x) = \left\{
  \begin{array}{ll}
    \kappa_1(x) & \text{if $f(x) = \kappa_1(y)$ for some $y \in B$} \\
    \kappa_2(x) & \text{if $f(x) = \kappa_2(z)$ for some $z \in C$} \\
    \text{undefined} & \text{if $f(x)$ is undefined}
  \end{array}\right. .
\end{equation*}

While inverse categories only have coproducts (much less restriction
coproducts) in very degenerate cases (see \cite{Giles2014}), they may very well
be equipped with a more general sum-like symmetric monoidal tensor, a
disjointness tensor.

\begin{defi}
  A \emph{disjointness tensor} on a restriction category is a symmetric 
  monoidal restriction functor $- \oplus -$ satisfying that its unit is the 
  restriction zero, and that the canonical maps
  \begin{equation*}
    \amalg_1 = A \tot{\rho^{-1}} A \oplus 0 \tot{\id \oplus \zero} A 
    \oplus B \qquad\qquad
    \amalg_2 = B \tot{\lambda^{-1}} 0 \oplus B \tot{\zero \oplus \id} A 
    \oplus B
  \end{equation*}
  are jointly epic, where $\rho$ respectively $\lambda$ is the left 
  respectively right unitor of the monoidal functor $- \oplus -$.
\end{defi}

It can be straightforwardly shown that any restriction coproduct gives rise to a disjointness tensor. A useful interaction between compatible joins and a join-preserving disjointness tensor in inverse categories was shown in \cite{Axelsen2016,Kaarsgaard2017}, namely that it leads to a \dag-trace (in the sense of \cite{Joyal1996,Selinger2011}):
\begin{prop}\label{prop:trace}
  Let \C{} be an inverse category with (at least countable) compatible joins 
  and a join-preserving disjointness tensor. Then \C{} has a trace operator
  given by
  $$\Tr_{A,B}^U(f) = f_{11} \vee \bigvee_{n \in \omega} f_{21} \comp f_{22}^n 
  \comp f_{12}$$
  satisfying $\Tr_{A,B}^U(f)^\dag = \Tr_{A,B}^U(f^\dag)$, where $f_{ij} = 
  \amalg_j^\dag \comp f \comp \amalg_i$.
\end{prop}


\section{Extensivity of inverse categories} 
\label{sec:extensivity_inverse_categories}
As discussed earlier, extensivity of restriction categories hinges on the
existence of certain partial isomorphisms -- decisions -- yet their
axiomatization relies on the presence of a map that is not a partial
isomorphism, the codiagonal.


In this section, we tweak the axiomatization of extensivity of restriction
categories to one that is equivalent, but additionally transports more easily
to inverse categories. We then give a definition of extensitivity for inverse
categories, from which it follows that $\Inv(\C)$ is an extensive inverse
category when $\C$ is an extensive restriction category.

Recall that decisions satisfy the following two axioms:
\begin{multicols}{2}
\begin{description}
  \item[\enspace(D.1)] $\codiag \comp \dec{f} = \ridm{f}$ and
  \item[\enspace(D.2)] $(f + f) \comp \dec{f} = (\kappa_1 + \kappa_2) \comp f$
\end{description}
\end{multicols}
As mentioned previously, an immediate problem with this is the reliance on the
codiagonal. However, intuitively, what \textbf{(D.1)} states is simply that the
decision $\dec{f}$ cannot do anything besides to tag its inputs appropriately. Using a disjoint join, we reformulate this axiom to the following:
\begin{description}
  \item[\enspace(D'.1)] $(\kappa_1^\dag \comp \dec{f}) \vee (\kappa_2^\dag \comp
  \dec{f}) = \ridm{f}$
\end{description}
Note that this axiom also subtly states that disjoint joins of the given form
always exist.

Say that a restriction category is \emph{pre-extensive} if it has restriction
coproducts, a restriction zero, and a combinator mapping each map $A \tot{f}
B+C$ to a \emph{pre-decision} $A \tot{\dec{f}} A + A$ (with no additional requirements). We can then show the following:

\begin{thm}\label{thm:pre_ext}
  Let \C{} be a pre-extensive restriction category. The following are 
  equivalent:
  \begin{enumerate}[(i), ref={\roman*}]
    \item \C{} is an extensive restriction category. \label{thm:p:ext}
    \item Every pre-decision of \C{} satisfies \textbf{(D.1)} and 
    \textbf{(D.2)}. \label{thm:p:dec_d1}
    \item Every pre-decision of \C{} satisfies \textbf{(D'.1)} and 
    \textbf{(D.2)}. \label{thm:p:dec_dp1}
  \end{enumerate}
\end{thm}

To show this theorem, we will need the following lemma:

\begin{lem}\label{lem:ext_join}
  In an extensive restriction category, joins of the form $(f + 0) \vee 
  (0 + g)$ exist for all maps $A \tot{f} B$ and $C \tot{g} D$ and are equal to 
  $f+g$.
\end{lem}
\begin{proof}
  By \cite{Cockett2007a}, for any map $A \tot{h} B+C$ in an
  extensive restriction category, $h = (\ridm{\kappa_1^\dag} \comp~ h) \vee
  (\ridm{\kappa_2^\dag} \comp~ h)$. But then $f+g = (\ridm{\kappa_1^\dag} \comp
  (f+g)) \vee (\ridm{\kappa_2^\dag} \comp (f+g)) = ((\id + \zero) \comp (f+g)) 
  \vee ((\zero + \id) \comp (f+g)) = (f + \zero) \vee (\zero + g)$.
\end{proof}

We can now continue with the proof.

\begin{proof}
  The equivalence between \eqref{thm:p:ext} and \eqref{thm:p:dec_d1} was given 
  in \cite{Cockett2007}. That \eqref{thm:p:dec_d1} and 
  \eqref{thm:p:dec_dp1} are equivalent follows by
  \begin{align*}
    (\kappa_1^\dag \comp \dec{f}) \vee (\kappa_2^\dag \comp \dec{f})
    & = ([\id, \zero] \comp \dec{f}) \vee ([\zero, \id] \comp \dec{f}) \\
    & = (\codiag \comp (\id + \zero) \comp \dec{f}) \vee (\codiag \comp (\zero
    + \id) \comp \dec{f}) \\
    & = \codiag \comp ((\id + \zero) \vee (\zero + \id)) \comp \dec{f} \\
    & = \codiag \comp (\id + \id) \comp \dec{f}
    = \codiag \comp \dec{f}
  \end{align*}
  where we note that the join $(\id + \zero) \vee (\zero + \id)$ exists and 
  equals $\id + \id$ when every pre-decision satisfies \textbf{(D.1)} and
  \textbf{(D.2)} by Lemma~\ref{lem:ext_join}. That the join also exists when
  every pre-decision satisfies \textbf{(D'.1)} and \textbf{(D.2)} follows as
  well, since the universal mapping property for coproducts guarantees that the
  only map $g$ satisfying $((\kappa_1 + \kappa_2) + (\kappa_1 + \kappa_2)) \comp
  g = (\kappa_1 + \kappa_2) \comp (\kappa_1 + \kappa_2)$ is $\kappa_1 +
  \kappa_2$ itself, so we must have $\dec{\kappa_1 + \kappa_2} = \kappa_1 +
  \kappa_2$, and
  \begin{align*}
    (\id + \zero) \vee (\zero + \id) 
     & = \ridm{\kappa_1^\dag} \vee \ridm{\kappa_2^\dag}
     = (\kappa_1 \comp \kappa_1^\dag) \vee (\kappa_2 \comp \kappa_2^\dag) \\
     & = (\kappa_1^\dag \comp (\kappa_1 + \kappa_2)) \vee
    (\kappa_2^\dag \comp (\kappa_1 + \kappa_2)) \\
     & = \ridm{\kappa_1 + \kappa_2} 
     = \id + \id
  \end{align*}
  which was what we wanted.
\end{proof}

Another subtle consequence of our amended first rule is that $\kappa_1^\dag
\comp \dec{f}$ is its own restriction idempotent (and likewise for
$\kappa_2^\dag \comp \dec{f}$) since $\kappa_1^\dag \comp \dec{f} \le
(\kappa_1^\dag \comp \dec{f}) \vee (\kappa_2^\dag \comp \dec{f}) = \ridm{f} \le
\id$, as the maps below identity are precisely the restriction idempotents.

Our next snag in transporting this definition to inverse categories has to do
with the restriction coproducts themselves, as it is observed in
\cite{Giles2014} that any inverse category with restriction coproducts is a
preorder. Intuitively, the problem is not that unicity of coproduct maps cannot
be guaranteed in non-preorder inverse categories, but rather that the coproduct
map $A+B \tot{[f,g]} C$ in a restriction category is not guaranteed to be a
partial isomorphism when $f$ and $g$ are.

For this reason, we will consider the more general disjointness tensor for
sum-like constructions rather than full-on restriction coproducts, as inverse
categories may very well have a disjointness tensor without it leading to
immediate degeneracy. Notably, \PInj{} has a disjointness tensor, constructed
on objects as the disjoint union of sets (precisely as the restriction
coproduct in \Pfn, but without the requirement of a universal mapping property). 
This leads us to the following definition:

\begin{defi}\label{def:inv_ext}
  An inverse category with a disjointness tensor is said to be \emph{extensive} 
  when each map $A \tot{f} B \oplus C$ has a unique decision 
  $A \tot{\dec{f}} A \oplus A$ satisfying
  \begin{description}
    \item[\enspace(D'.1)] $(\amalg_1^\dag \comp \dec{f}) \vee (\amalg_2^\dag
    \comp \dec{f}) = \ridm{f}$
    \item[\enspace(D'.2)] $(f \oplus f) \comp \dec{f} = (\amalg_1 \oplus
    \amalg_2) \comp f$.
  \end{description}
\end{defi}
As an example, \PInj{} is an extensive inverse category with the unique
decision $A \tot{\dec{f}} A \oplus A$ for a partial injection $A \tot{f} B
\oplus C$ given by
\begin{equation*}
  \dec{f}(x) = \left\{ \begin{array}{ll}
    \amalg_1(x) & \text{if}\ f(x) = \amalg_1(y)\ \text{for some $y \in B$} \\
    \amalg_2(x) & \text{if}\ f(x) = \amalg_2(z)\ \text{for some $z \in C$} \\
    \text{undefined} & \text{if $f(x)$ is undefined} 
  \end{array} \right. \enspace.
\end{equation*}
Aside from a shift from coproduct injections to the quasi-injections of the
disjointness tensor, a subtle change here is the notion of join. That is, for
restriction categories with disjoint joins, any pair of maps $f,g$ with $f
\comp \ridm{g} = 0$ has a join -- but for inverse categories, we additionally
require that their \emph{inverses} are disjoint as well, \ie, that $f^\dag
\comp \ridm{g^\dag} = 0$, for the join to exist. In this case, however, there
is no difference between the two. As previously discussed, a direct consequence
of this axiom is that each $\amalg_i^\dag \comp \dec{f}$ must be its own
restriction idempotent. Since restriction idempotents are self-adjoint (\ie,
satisfy $f = f^\dag$), they are disjoint iff their inverses are disjoint.

Since restriction coproducts give rise to a disjointness tensor, we may
straightforwardly show the following theorem.

\begin{thm}
  When \C{} is an extensive restriction category, $\Inv(\C)$ is an extensive 
  inverse category.
\end{thm}

Further, constructing the decision $\dec{f}$ as $(\amalg_1 \comp
\ridm{\amalg_1^\dag \comp f}) \vee (\amalg_2 \comp \ridm{\amalg_2^\dag \comp
f})$ (\ie, mirroring the construction of decisions in restriction categories
with disjoint joins), we may show the following.

\begin{thm}\label{thm:join_inv_ext}
  Let \C{} be an inverse category with a disjointness tensor, a restriction 
  zero, and finite disjoint joins which are further preserved by the 
  disjointness tensor. Then \C{} is extensive as an inverse category.
\end{thm}

Before we proceed, we show two small lemmas regarding decisions and the (join
preserving) disjointness tensor (for the latter, see also \cite{Giles2014}).
\begin{lem}\label{lem:basic_disj}
  In any inverse category with a disjointness tensor, it is the case that
  \begin{enumerate}[(i), ref={\roman*}]
    \item $\ridm{\amalg_1} = \ridm{\amalg_2} = \id$, \label{bd1}
    \item $\ridm{\amalg_1^\dagger} = \id \oplus 0$ and $\ridm{\amalg_2^\dagger}
    = 0 \oplus \id$, \label{bd2}
    \item $\amalg_j^\dagger \comp \amalg_i = \id$ if $i=j$ and 
    $\amalg_j^\dagger \comp \amalg_i = 0$ otherwise. \label{bd3}
  \end{enumerate}
\end{lem}
\begin{proof}
  For \eqref{bd1}, adding subscripts on identities and zero maps for clarity,
  \begin{align*}
    \ridm{\amalg_1} & = \ridm{(\id_X \oplus 0_{0,Y}) \comp 
    \rho^{-1}} = \ridm{\ridm{(\id_X \oplus 0_{0,Y})} \comp \rho^{-1}}
    = \ridm{(\ridm{\id_X} \oplus \ridm{0_{0,Y}}) \comp \rho^{-1}}
    = \ridm{(\id_X \oplus 0_{0,0}) \comp \rho^{-1}} \\
    & = \ridm{(\id_X \oplus \id_0) \comp \rho^{-1}}
    = \ridm{\id_{X \oplus 0} ~\comp \rho^{-1}} = \ridm{\rho^{-1}} = \id_X
  \end{align*}
  and analogously for $\ridm{\amalg_2}$. For \eqref{bd2},
  $$
  \ridm{\amalg_1^\dag} = \ridm{((\id \oplus 0) \comp \rho^{-1})^\dag} =
  \ridm{\rho \comp (\id^\dag \oplus 0^\dag)} = \ridm{\ridm{\rho} \comp (\id
  \oplus 0)} = \ridm{\id ~\comp (\id \oplus 0)} = \ridm{\id \oplus 0} = \id 
  \oplus 0
  $$
  and again, the proof for $\amalg_2$ is analogous. To show \eqref{bd3} when 
  $i=j$
  \begin{align*}
  \amalg_1^\dag \comp \amalg_1 
  & = \rho \comp (\id_X \oplus 0_{Y,0}) \comp (\id_X \oplus 0_{0,Y}) \comp
    \rho^{-1}
  = \rho \comp (\id_X \oplus 0_{0,0}) \comp \rho^{-1} 
  = \rho \comp (\id_X \oplus \id_0) \comp \rho^{-1} \\
  & = \rho \comp ~\id_{X \oplus 0} ~\comp \rho^{-1}
  = \rho \comp \rho^{-1}
  = \id_X
  \end{align*}
  and similarly for $\amalg_2^\dag \comp \amalg_2 = \id$ (noting that 
  $0_{Y,0} \comp 0_{0,Y} = 0_{0,0} = \id_0$ follows by the universal mapping 
  property of the zero object $0$). For $i \neq j$, we proceed with the case 
  where $i=1$ and $j=2$, yielding
  \begin{align*}
    \amalg_2^\dag \comp \amalg_1 = \rho \comp (0_{X,0} \oplus \id_Y) \comp
    (\id_X \oplus 0_{0,Y}) \comp \rho^{-1} = \rho \comp (0_{X,0} \oplus
    0_{0,Y}) \comp \rho^{-1} = \rho ~\comp 0_{X \oplus 0, 0 \oplus Y} ~\comp 
    \rho^{-1} = 0_{X,Y}
  \end{align*}
  where $0_{X,0} \oplus 0_{0,Y} = 0_{X \oplus 0, 0 \oplus Y}$ follows by the
  fact that $0_{X,0} \oplus 0_{0,Y}$ factors through $0 \oplus 0$ (by the
  universal mapping property of $0$) and the fact that $0 \oplus 0 \cong 0$
  (since $0$ serves as unit for $-\oplus-$), so $0_{X,0} \oplus 0_{0,Y}$
  factors through $0$ as well. The case where $i=2$ and $j=1$ is entirely
  analogous.
\end{proof}

\begin{lem}\label{lem:basic_ext}
  In any extensive inverse category, it is the case that
  \begin{enumerate}[(i), ref={\roman*}]
    \item $\ridm{\amalg_1^\dag \comp \dec{f}} = \amalg_1^\dag \comp \dec{f}$, 
    \label{be1}
    \item $\ridm{\amalg_1^\dagger} \vee \ridm{\amalg_2^\dagger} = \id$, 
    \label{be2}
    \item $\dec{f} = (\amalg_1 \comp \ridm{\amalg_1^\dag \comp \dec{f}}) \vee 
    (\amalg_2 \comp \ridm{\amalg_2^\dag \comp \dec{f}})$, and \label{be3}
    \item $\ridm{f} = \ridm{\dec{f}}$. \label{be4}
  \end{enumerate}
\end{lem}
\begin{proof}
  To show \eqref{be1}, we have by definition of the join that $\amalg_1^\dag
  \comp \dec{f} \le (\amalg_1^\dag \comp \dec{f}) \vee (\amalg_2^\dag \comp
  \dec{f}) = \ridm{f}$, where the equality follows precisely by the first axiom 
  of decisions. Since $\ridm{f} \le \id$ by $\id \comp \ridm{\ridm{f}}
  = \ridm{f}$ (since $\ridm{\ridm{f}} = \ridm{f}$ by
  Lemma~\ref{lem:basic_ridm}), it follows by transitivity of $- \le -$ (since
  it is a partial order) that $\amalg_1^\dag \comp \dec{f} \le \id$. But since
  $f \le g$ iff $g \comp \ridm{f} = f$, it follows from $\amalg_1^\dag \comp
  \dec{f} \le \id$ that $\ridm{\amalg_1^\dag \comp \dec{f}} = \id~ \comp
  \ridm{\amalg_1^\dag \comp \dec{f}} = \amalg_1^\dag \comp \dec{f}$.
  
  For \eqref{be2}, using Lemma~\ref{lem:basic_disj} we start by noting that
  $$\ridm{\amalg_1^\dag \comp (\amalg_1 \oplus \amalg_2)} = 
  \ridm{\ridm{\amalg_1^\dag} \comp (\amalg_1 \oplus \amalg_2)} = 
  \ridm{(\id \oplus 0) \comp (\amalg_1 \oplus \amalg_2)} =
  \ridm{\amalg_1 \oplus 0} =
  \ridm{\amalg_1} \oplus \ridm{0} =
  \id \oplus 0 =
  \ridm{\amalg_1^\dag}$$
  and similarly
  $$\ridm{\amalg_2^\dag \comp (\amalg_1 \oplus \amalg_2)} = 
  \ridm{\ridm{\amalg_2^\dag} \comp (\amalg_1 \oplus \amalg_2)} = 
  \ridm{(0 \oplus \id) \comp (\amalg_1 \oplus \amalg_2)} =
  \ridm{0 \oplus \amalg_2} =
  \ridm{0} \oplus \ridm{\amalg_2} =
  0 \oplus \id =
  \ridm{\amalg_2^\dag}$$
  and since $\dec{\amalg_1 \oplus \amalg_2} = \amalg_1 \oplus \amalg_2$, we have
  \begin{align*}
    \ridm{\amalg_1^\dag} \vee \ridm{\amalg_2^\dag} 
    & = \ridm{\amalg_1^\dag \comp
    (\amalg_1 \oplus \amalg_2)} \vee \ridm{\amalg_2^\dag \comp (\amalg_1 \oplus
    \amalg_2)} = \ridm{\amalg_1^\dag \comp \dec{\amalg_1 \oplus \amalg_2}} \vee
    \ridm{\amalg_2^\dag \comp \dec{\amalg_1 \oplus \amalg_2}} \\
    & = \ridm{\ridm{\amalg_1 \oplus \amalg_2}} = \ridm{\amalg_1 \oplus
    \amalg_2} = \ridm{\amalg_1} \oplus \ridm{\amalg_2} = \id \oplus \id = \id.
  \end{align*}
  Using this, we may prove \eqref{be3} (and later \eqref{be4}) as follows:
  \begin{align*}
    \dec{f} & = \id \comp \dec{f} = (\ridm{\amalg_1^\dag} \vee
    \ridm{\amalg_2^\dag}) \comp \dec{f} = (\ridm{\amalg_1^\dag} \comp \dec{f})
    \vee (\ridm{\amalg_2^\dag} \comp \dec{f})
    = (\amalg_1 \comp \amalg_1^\dag \comp \dec{f}) \vee (\amalg_2 \comp
    \amalg_2^\dag \comp \dec{f}) \\
    & = (\amalg_1 \comp \ridm{\amalg_1^\dag \comp \dec{f}}) \vee (\amalg_2 \comp
    \ridm{\amalg_2^\dag \comp \dec{f}})
  \end{align*}
  where $\amalg_1^\dag \comp \dec{f} = \ridm{\amalg_1^\dag \comp \dec{f}}$ 
  follows by \eqref{be1}. Using these again, we show \eqref{be4} by
  \begin{align*}
    \ridm{f} & = (\amalg_1^\dag \comp \dec{f}) \vee (\amalg_2^\dag \comp 
    \dec{f}) = \ridm{\amalg_1^\dag \comp \dec{f}} \vee \ridm{\amalg_2^\dag
    \comp \dec{f}}
    = \ridm{\ridm{\amalg_1^\dag} \comp \dec{f}} \vee \ridm{\ridm{\amalg_2^\dag}
    \comp \dec{f}}
    = \ridm{(\ridm{\amalg_1^\dag} \comp \dec{f}) \vee (\ridm{\amalg_2^\dag}
    \comp \dec{f})} \\
    & = \ridm{(\ridm{\amalg_1^\dag} \vee \ridm{\amalg_2^\dag})
    \comp \dec{f})} = \ridm{\id \comp \dec{f}} = \ridm{\dec{f}}
  \end{align*}
  which was what we wanted.
\end{proof}


\section{Modelling structured reversible flowcharts} 
\label{sec:modelling_reversible_flowcharts}
In the following, let \C{} be an inverse category with (at least countable)
compatible joins and a join-preserving disjointness tensor. As disjoint joins
are compatible, it follows that \C{} is an extensive inverse category with a
(uniform) \dag-trace operator.

In this section, we will show how this framework can be used model reversible
structured flowchart languages. First, we will show how decisions in extensive
inverse categories can be used to model predicates, and then how this
representation extends to give very natural semantics to reversible flowcharts
corresponding to conditionals and loops. 

\subsection{Predicates as decisions} 
\label{sub:predicates_as_decisions}
In suitably equipped categories, one naturally considers predicates on an
object $A$ as given by maps $A \to 1+1$. In inverse categories, however, the
mere idea of a predicate as a map of the form $A \to 1 \oplus 1$ is
problematic, as only very degenerate maps of this form are partial
isomorphisms. In the following, we show how decisions give rise to an
unconventional yet ultimately useful representation of predicates. To our
knowledge this representation is novel, motivated here by the necessity to
model predicates in a reversible fashion, as decisions are always partial
isomorphisms.

The simplest useful predicates are the predicates that are always true
(respectively always false). By convention, we represent these by the left
(respectively right) injection (which are both their own decisions),
\begin{align*}
  \oxf{\stt} & = \amalg_1 \\
  \oxf{\sff} & = \amalg_2.
\end{align*}
Semantically, we may think of decisions as a separation of an object $A$ into
\emph{witnesses} and \emph{counterexamples} of the predicate it represents.
In a certain sense, the axioms of decisions say that there is nothing more to a
decision than how it behaves when postcomposed with $\amalg_1^\dag$ or
$\amalg_2^\dag$. As such, given the convention above, we think of
$\amalg_1^\dag \comp \dec{p}$ as the witnesses of the predicate represented by
the decision $\dec{p}$, and $\amalg_2^\dag \comp \dec{p}$ as its
counterexamples.

With this in mind, we turn to boolean combinators. The negation of a
predicate-as-a-decision must simply swap witnesses for counterexamples (and
vice versa). In other words, we obtain the negation of a decision by
postcomposing with the commutator $\gamma$ of the disjointness tensor,
\begin{equation*}
  \oxf{\snot{p}} = \gamma \comp \oxf{p}.
\end{equation*}
With this, it is straightforward to verify that, \eg, $\oxf{\snot{\stt}} = \oxf{\sff}$, as
\begin{equation*}
  \oxf{\snot{\stt}} = \gamma \comp \amalg_1 = \gamma \comp (\id \oplus \zero) \comp \rho^{-1} = (\zero \oplus \id) \comp \gamma \comp \rho^{-1}
   = (\zero \oplus \id) \comp \lambda^{-1} = \amalg_2 = \oxf{\sff}.
\end{equation*}
For conjunction, we exploit that our category has (specifically) finite
disjoint joins, and define the conjunction of predicates-as-decisions $\oxf{p}$ and $\oxf{q}$ by
\begin{equation*}
  \oxf{\sand{p}{q}} = \left((\amalg_1 \comp \ridm{\amalg_1^\dag \comp \oxf{p}}
  ~\comp \ridm{\amalg_1^\dag \comp \oxf{q}}) \vee (\amalg_2 \comp
  (\ridm{\amalg_2^\dag \comp \oxf{p}} \vee \ridm{\amalg_2^\dag \comp 
  \oxf{q}}))\right) \comp \ridm{\oxf{p}} ~\comp \ridm{\oxf{q}}.
\end{equation*}
The intuition behind this definition is that the witnesses of a conjunction of
predicates is given by the meet of the witnesses of the each predicate, while
the counterexamples of a conjunction of predicates is the join of the
counterexamples of each predicate. Note that this is then precomposed with $\ridm{\oxf{p}} ~\comp \ridm{\oxf{q}}$ to ensure that the result is only defined where both $p$ and $q$ are; this gives 

 Noting that the meet of two restriction
idempotents is given by their composition, this is precisely what this
definition states. Similarly we define the disjunction of $\oxf{p}$ and $\oxf{q}$ by
\begin{equation*}
  \oxf{p~\mathbf{or}~q} = \left((\amalg_1 \comp
(\ridm{\amalg_1^\dag \comp \oxf{p}} \vee \ridm{\amalg_1^\dag \comp \oxf{q}}))
\vee (\amalg_2 \comp \ridm{\amalg_2^\dag \comp \oxf{p}} ~\comp \ridm{\amalg_2^\dag \comp \oxf{q}})\right) \comp \ridm{\oxf{p}} ~\comp \ridm{\oxf{q}},
\end{equation*}
as $\oxf{p~\mathbf{or}~q}$ then has as witnesses the join of the witnesses of
$\oxf{p}$ and $\oxf{q}$, and as counterexamples the meet of the counterexamples
of $\oxf{p}$ and $\oxf{q}$. With these definitions, it can be shown that, \eg,
the De Morgan laws are satisfied. However, since we can thus construct this from conjunctions and negations, we will leave disjunctions as syntactic sugar.

That all of these are indeed decisions can be shown straightforwardly, as summarized in the following closure theorem.

\begin{thm}\label{thm:dec_closed}
  Decisions in \C{} are closed under Boolean negation, conjunction, and
  disjunction.
\end{thm}

\subsection{Reversible structured flowcharts, categorically} 
\label{sub:reversible_structured_flowcharts_categorically}
To give a categorical account of structured reversible flowchart languages, we
assume the existence of a suitable distinguished object $\Sigma$ of stores,
which we think of as the \emph{domain of computation}, such that we may give
denotations to structured reversible flowcharts as morphisms $\Sigma \to
\Sigma$.

Since atomic steps (corresponding to elementary operations, \eg, store updates)
may vary from language to language, we assume that each such atomic step in our
language has a denotation as a morphism $\Sigma \to \Sigma$. In the realm of reversible flowcharts, these atomic steps are required to be partial injective functions; here, we abstract this to require that their denotation is a partial isomorphism (though this is a trivial requirement in inverse categories). 

Likewise, elementary predicates (\eg, comparison of values in a store) may vary
from language to language, so we assume that such elementary predicates have
denotations as well as decisions $\Sigma \to \Sigma \oplus \Sigma$. If
necessary (as is the case for Janus~\cite{YoGl07a}), we may then close these
elementary predicates under boolean combinations as discussed in the previous
section.

To start, we note how sequencing of flowcharts may be modelled trivially by
means of composition, \ie,
\begin{equation*}
  \oxf{\sseq{c_1}{c_2}} = \oxf{c_2} \comp \oxf{c_1}
\end{equation*}
or, using the diagrammatic notation of flowcharts and the string diagrams for monoidal categories in the style of \cite{Selinger2011} (read left-to-right and bottom-to-top),
\setlength{\unitlength}{1.6mm}
$$
\oxf{\input{Sequence_r.pic}} \enspace 
= \enspace \begin{tikzpicture}[baseline]
\node[draw, minimum size=6mm] (f) at (0,0.1) {\small $\oxf{c_1}$};
\node[draw, minimum size=6mm] (g) at (1.5,0.1) {\small $\oxf{c_2}$};

\draw[->] (-1.13,0.1) to node {} (-0.75,0.1);
\draw[-] (-0.76,0.1) to node {} (-0.38,0.1);
\draw[->] (0.39,0.1) to node {} (0.8,0.1);
\draw[-] (0.79,0.1) to node {} (1.11,0.1);
\draw[->] (1.88,0.1) to node {} (2.255,0.1);
\draw[-] (2.24,0.1) to node {} (2.61,0.1);
\end{tikzpicture} \enspace.
$$
Intuitively, a decision separates an object into
witnesses (in the first component) and counterexamples (in the second). As
such, the partial inverse to a decision must be defined only on witnesses in
the first component, and only on counterexamples in the second. But then, where decisions model \emph{predicates}, codecisions (\ie, partial inverses to
decisions) model \emph{assertions}.

With this in mind, we achieve a denotation of reversible conditionals as
\setlength{\unitlength}{1.9mm}
\begin{equation*}
  \oxf{\sif{p}{c_1}{c_2}{q}} =
  \oxf{q}^\dag \comp (\oxf{c_1} \oplus \oxf{c_2}) \comp \oxf{p}
\end{equation*}
or, as diagrams
$$
\oxf{\input{Selection_r.pic}} \enspace
= \enspace \begin{tikzpicture}[baseline]
\node[draw, minimum height=15mm] (p) at (0,0) {\small $\oxf{p}$};
\node[draw, minimum size=6mm] (g) at (1.5,0.475) {\small $\oxf{c_2}$};
\node[draw, minimum size=6mm] (f) at (1.5,-0.475) {\small $\oxf{c_1}$};
\node[draw, minimum height=15mm] (q) at (3,0) {\small $\oxf{q}^\dag$};

\draw[->] (-1,-0.475) to node {} (-0.65,-0.475);
\draw[-] (-0.66,-0.475) to node {} (-0.32,-0.475);

\draw[->] (0.325,0.475) to node {} (0.8,0.475);
\draw[-] (0.79,0.475) to node {} (g);
\draw[->] (0.325,-0.475) to node {} (0.8,-0.475);
\draw[-] (0.79,-0.475) to node {} (f);
\draw[->] (g) to node {} (2.25,0.475);
\draw[-] (2.24,0.475) to node {} (2.62,0.475);
\draw[->] (f) to node {} (2.25,-0.475);
\draw[-] (2.24,-0.475) to node {} (2.62,-0.475);

\draw[->] (3.38,-0.475) to node {} (3.69,-0.475);
\draw[-] (3.68,-0.475) to node {} (4,-0.475);
\end{tikzpicture} \enspace.
$$

To give the denotation of reversible loops, we use the \dag-trace operator. Defining a shorthand for the body of the loop as
\begin{equation*}
  \beta[p,c,q] = (\id_\Sigma \oplus \oxf{c}) \comp \oxf{q} \comp
  \oxf{p}^\dag
\end{equation*}
we obtain the denotation
\begin{equation*}
  \oxf{\sfrom{p}{c}{q}} = \Tr_{\Sigma,\Sigma}^\Sigma(\beta[p,c,q]) = 
  \Tr_{\Sigma,\Sigma}^\Sigma((\id_\Sigma \oplus \oxf{c}) \comp \oxf{q} \comp
  \oxf{p}^\dag)
\end{equation*}
or diagrammatically
$$
\oxf{\input{WhileLoop_r.pic}} \enspace
= \enspace \begin{tikzpicture}[baseline]
\node[draw, minimum height=15mm] (q) at (0,0) {\small $\oxf{p}^\dag$};
\node[draw, minimum height=15mm] (p) at (1.5,0) {\small $\oxf{q}$};
\node[draw, minimum size=6mm] (f) at (3,0.475) {\small $\oxf{c}$};

\draw[-] (-0.7,0.475) to node {} (-0.38,0.475);

\draw[->] (-1,-0.475) to node {} (-0.66,-0.475);
\draw[-] (-0.67,-0.475) to node {} (-0.38,-0.475);
\draw[->] (0.38,-0.475) to node {} (0.78,-0.475);
\draw[-] (0.77,-0.475) to node {} (1.18,-0.475);
\draw[->] (1.82,0.475) to node {} (2.26,0.475);
\draw[-] (2.25,0.475) to node {} (2.7,0.475);
\draw[->] (1.82,-0.475) to node {} (2.86,-0.475);
\draw[-] (2.85,-0.475) to node {} (4,-0.475);

\draw[-] (3.3,0.475) to node {} (3.627,0.475);
\draw[-] (3.62,0.475) to node {} (3.62,1.257);
\draw[->] (3.62,1.25) to node {} (1.5,1.25);
\draw[-] (1.6,1.25) to node {} (-0.707,1.25);
\draw[-] (-0.7,1.25) to node {} (-0.7,0.468);
\end{tikzpicture} \enspace.
$$
That this has the desired operational behavior follows from the fact that the
\dag-trace operator is canonically constructed in join inverse categories as
\begin{equation*}
  \Tr_{X,Y}^U(f) = f_{11} \vee \bigvee_{n \in \omega} f_{21} \comp f_{22}^n
  \comp f_{12} \enspace.
\end{equation*}
Recall that $f_{ij} = \amalg_j^\dag \comp f \comp \amalg_i$. As such, for our
loop construct defined above, the $f_{11}$-cases correpond to cases where a
given state bypasses the loop entirely; $f_{21} \comp f_{12}$ (that is, for
$n=0$) to cases where exactly one iteration is performed by a given state
before exiting the loop; $f_{21} \comp f_{22} \comp f_{12}$ to cases where two
iterations are performed before exiting; and so on. In this way, the given
trace semantics contain all successive loop unrollings, as desired. We will
make this more formal in the following section, where we show computational
soundness and adequacy for these with respect to the operational semantics.

In order to be able to provide a correspondence between categorical and operational semantics, we also need an interpretation of the meta-command $\underline{\mathit{loop}}$. While it may not be so clear at the present, it turns out that the appropriate one is
\begin{equation*}
  \oxf{\sloop{p}{c}{q}} = \bigvee_{n \in \omega} \beta[p,c,q]_{21} \comp 
  \beta[p,c,q]_{22}^n \enspace.
\end{equation*}

While it may seem like a small point, the mere existence of a categorical
semantics in inverse categories for a reversible programming language has some
immediate benefits. In particular, that a programming language is reversible
can be rather complicated to show by means of operational semantics (see, \eg,
\cite[Sec.~2.3]{YoGl07a}), yet it follows directly in our categorical
semantics, as it is compositional and all morphisms in inverse categories have a unique partial
inverse. 



\section{Computational soundness and adequacy} 
\label{sec:soundness_and_adequacy}
Computational soundness and adequacy (see, \eg, \cite{Fiore1994})
are the two fundamental properties of operational semantics with respect to
their denotational counterparts, as soundness and completeness are for proof
systems with respect to their semantics. In brief, computational soundness and
adequacy state that the respective notions of \emph{convergence} of the
operational and denotational semantics are in agreement.

In the operational semantics, the notion of convergence seems straightforward:
a program $c$ converges in a state $\sigma$ if there exists another state
$\sigma'$ such that $\bigstepc{\sigma}{c}{\sigma'}$. On the denotational side,
it seems less obvious what a proper notion of convergence is.

An idea (used by, \eg, Fiore~\cite{Fiore1994}) is to let values (in this case,
states) be interpreted as \emph{total} morphisms from some sufficiently simple
object $I$ into an appropriate object $V$ (here, we will use our object
$\Sigma$ of states). In this context, the notion of convergence for a program
$p$ in a state $\sigma$ is then that the resulting morphism $\oxf{p} \comp
\oxf{\sigma}$ is, again, (the denotation of) a state -- \ie, it is total.
Naturally, this approach requires machinery to separate total maps from partial
ones. As luck would have it inverse categories fit the bill perfectly, as they
can be regarded as special instances of restriction categories.

To make this idea more clear in the current context, and to allow us to use the
established formulations of computational soundness and adequacy,
we define a model of a structured reversible flowchart language to be the
following:

\begin{defi}
  A model of a structured reversible flowchart language $\mathcal{L}$ consists 
  of a join inverse category $\C$ with a disjointness tensor, further equipped 
  with distinguished objects $I$ and $\Sigma$ satisfying
  \begin{enumerate}[(i), ref={\roman*}]
    \item the identity and zero maps on $I$ are distinct, \ie, $\id_I \neq 
    0_{I,I}$,
    \item if $I \tot{e} I$ is a restriction idempotent then $e = \id_I$ or $e = 
    0_{I,I}$, and
    \item each $\mathcal{L}$-state $\sigma$ is interpreted as a 
    \emph{total} morphism $I \tot{\oxf{\sigma}} \Sigma$.
  \end{enumerate}
\end{defi}

Here, we think of $I$ as the \emph{indexing object}, and $\Sigma$ as the
\emph{object of states}. In irreversible programming languages, the first two
conditions in the definition above are often left out, as the indexing object
is typically chosen to be the terminal object $1$. However, terminal
objects are degenerate in inverse categories, as they always coincide with the
initial object when they exist -- that is, they are zero objects. For this
reason, we require instead the existence of a sufficiently simple indexing
object, as described by these two properties. For example, in \PInj, any
one-element set will satisfy these conditions. 

Even further, the third condition is typically proven rather than assumed. We
include it here as an assumption since structured reversible flowchart
languages may take many different forms, and we have no way of knowing how the
concrete states are formed. As such, rather than limiting ourselves to
languages where states take a certain form in order to show totality of
interpretation, we instead assume it to be able to show properties about more programming languages.

This also leads us to another important point: We are only able to show
computational soundness and adequacy for the operational
semantics as they are stated, \ie, we are not able to take into account the
specific atomic steps (besides $\sskip$) or elementary predicates of the
language.

As such, computational soundness and adequacy (and what may
follow from that) should be understood \emph{conditionally}: If a structured
reversible flowchart language has a model of the form above \emph{and} it is
computationally sound and adequate with respect to its atomic steps and
elementary predicates, then the entire interpretation is sound and adequate as
well.

We begin by recalling the definition of the denotation of predicates and commands in a model of a structured reversible flowchart language from Section~\ref{sec:modelling_reversible_flowcharts}.

\begin{defi}\label{def:predicates}
  Recall the interpretation of predicates in $\mathcal{L}$ as decisions in \C:
  \begin{enumerate}[(i), ref={\roman*}]
    \item $\oxf{\stt} = \amalg_1$,
    \item $\oxf{\sff} = \amalg_2$,
    \item $\oxf{\snot{p}} = \gamma \comp \oxf{p}$,
    \item $\oxf{\sand{p}{q}} = \left((\amalg_1 \comp \ridm{\amalg_1^\dag
    \comp \oxf{p}} ~\comp \ridm{\amalg_1^\dag \comp \oxf{q}}) \vee
    (\amalg_2 \comp (\ridm{\amalg_2^\dag \comp \oxf{p}} \vee
    \ridm{\amalg_2^\dag \comp \oxf{q}}))\right) \comp
    \ridm{\oxf{p}} ~\comp \ridm{\oxf{q}}$.
  \end{enumerate}
\end{defi}

\begin{defi}\label{def:commands}
  Recall the interpretation of commands in $\mathcal{L}$ (and the meta-command 
  $\underline{\mathit{loop}}$) as morphisms $\Sigma 
  \to \Sigma$ in \C:
  \begin{enumerate}[(i), ref={\roman*}]
    \item $\oxf{\sskip} = \id_\Sigma$, 
    \item $\oxf{\sseq{c_1}{c_2}} = \oxf{c_2} \comp \oxf{c_1}$, 
    \item $\oxf{\sif{p}{c_1}{c_2}{q}} = \oxf{q}^\dag \comp (\oxf{c_1} \oplus 
    \oxf{c_2}) \comp \oxf{p}$,
    \item $\oxf{\sfrom{p}{c}{q}} = \Tr_{\Sigma,\Sigma}^\Sigma(\beta[p,c,q])$, 
    and
    \item $\oxf{\sloop{p}{c}{q}} = \bigvee_{n \in \omega} 
    \beta[p,c,q]_{21} \comp \beta[p,c,q]_{22}^n$
  \end{enumerate}
  where $\beta[p,c,q] = (\id_\Sigma \oplus \oxf{c}) \comp \oxf{q} \comp
  \oxf{p}^\dag : \Sigma \oplus \Sigma \to \Sigma \oplus \Sigma$. Note also
  that $f_{ij} = \amalg_j^\dag \comp f \comp \amalg_i$.
\end{defi}

The overall strategy we will use to show computational soundness and
computational adequacy for programs is to start by showing it for predicates.
To begin to tackle this, we first need a lemma regarding the totality of
predicates.

\begin{lem}\label{lem:total_pred}
  Let $p$ and $q$ be $\mathcal{L}$-predicates. It is the case that
  \begin{enumerate}[(i), ref={\roman*}]
    \item\label{prtot:const} $I \tot{\oxf{\sigma}} \Sigma \tot{\oxf{\stt}} 
    \Sigma \oplus \Sigma$ and $I \tot{\oxf{\sigma}} \Sigma \tot{\oxf{\sff}} 
    \Sigma \oplus \Sigma$ are total,
    \item\label{prtot:not} $I \tot{\oxf{\sigma}} \Sigma \tot{\oxf{\snot{p}}} 
    \Sigma \oplus
    \Sigma$ is total iff $I \tot{\oxf{\sigma}} \Sigma \tot{\oxf{p}} 
    \Sigma \oplus \Sigma$ is, and
    \item\label{prtot:and} $I \tot{\oxf{\sigma}} \Sigma
    \tot{\oxf{\sand{p}{q}}} \Sigma \oplus \Sigma$ is total iff $I
    \tot{\oxf{\sigma}} \Sigma \tot{\oxf{p}} \Sigma \oplus \Sigma$ and $I
    \tot{\oxf{\sigma}} \Sigma \tot{\oxf{q}} \Sigma \oplus \Sigma$ both are.
  \end{enumerate}
\end{lem}
\begin{proof}
  For \eqref{prtot:const}, it follows that $$\ridm{\oxf{\stt}\comp \oxf{\sigma}}
  = \ridm{\amalg_1 \comp \oxf{\sigma}} = \ridm{\ridm{\amalg_1} \comp
  \oxf{\sigma}} = \ridm{\id_\Sigma \comp \oxf{\sigma}} = \ridm{\oxf{\sigma}} =
  \id_I,$$ where the final equality follows by the definition of a model. The 
  case for $\sff$ is entirely analogous.
  
  For \eqref{prtot:not}, we have that $$\ridm{\oxf{\snot{p}} \comp
  \oxf{\sigma}} = \ridm{\gamma \comp \oxf{p} \comp \oxf{\sigma}} =
  \ridm{\ridm{\gamma} \comp \oxf{p} \comp \oxf{\sigma}} = \ridm{\id_{\Sigma
  \oplus \Sigma} \comp \oxf{p} \comp \oxf{\sigma}} = \ridm{\oxf{p} \comp
  \oxf{\sigma}}$$ which implies directly that $I \tot{\oxf{\sigma}} \Sigma
  \tot{\oxf{\snot{p}}} \Sigma \oplus \Sigma$ is total iff $I \tot{\oxf{\sigma}}
  \Sigma \tot{\oxf{p}} \Sigma \oplus \Sigma$ is.
  
  For \eqref{prtot:and}, it suffices to show that $\ridm{\oxf{\sand{p}{q}} 
  \comp \oxf{\sigma}} = \ridm{\oxf{p} \comp \oxf{\sigma}} ~\comp \ridm{\oxf{q}
  \comp \oxf{\sigma}}$, since $\ridm{\oxf{p} \comp \oxf{\sigma}} = 
  \ridm{\oxf{q} \comp \oxf{\sigma}} = \id_I$ then yields
  $\ridm{\oxf{\sand{p}{q}} \comp \oxf{\sigma}} = \id_I \comp \id_I = \id_I$ 
  directly; the other direction follows by the fact that if $\ridm{g} \comp 
  \ridm{f} = \id$ then $\id = \ridm{g} \comp \ridm{f} = \ridm{g} \comp \ridm{f} 
  \comp~ \ridm{f} = \id \comp \ridm{f} = \ridm{f}$ (and analogously for $g$).
  
  We start by observing that
  \begin{align}
    \ridm{\oxf{p} \comp \oxf{\sigma}} & =
    \ridm{((\amalg_1 \comp \ridm{\amalg_1^\dag \comp \oxf{p}}) \vee (\amalg_2 
    \comp \ridm{\amalg_2^\dag \comp \oxf{p}})) \comp \oxf{\sigma}} \label{ps1}\\
    & = \ridm{(\amalg_1 \comp \ridm{\amalg_1^\dag \comp \oxf{p}} \comp 
    \oxf{\sigma}) \vee (\amalg_2 
    \comp \ridm{\amalg_2^\dag \comp \oxf{p}} \comp \oxf{\sigma})} \label{ps2}\\
    & = \ridm{(\amalg_1 \comp \oxf{\sigma} \comp \ridm{\amalg_1^\dag \comp
    \oxf{p} \comp \oxf{\sigma}}) \vee (\amalg_2 \comp \oxf{\sigma} \comp
    \ridm{\amalg_2^\dag \comp \oxf{p} \comp \oxf{\sigma}})} \label{ps3}\\
    & = \ridm{(\amalg_1 \comp \oxf{\sigma} \comp \ridm{\amalg_1^\dag \comp
    \oxf{p} \comp \oxf{\sigma}})} \vee \ridm{(\amalg_2 \comp \oxf{\sigma} \comp
    \ridm{\amalg_2^\dag \comp \oxf{p} \comp \oxf{\sigma}})} \label{ps4}\\
    & = \ridm{(\ridm{\amalg_1 \comp \oxf{\sigma}} \comp \ridm{\amalg_1^\dag
    \comp \oxf{p} \comp \oxf{\sigma}})} \vee \ridm{(\ridm{\amalg_2 \comp
    \oxf{\sigma}} \comp \ridm{\amalg_2^\dag \comp \oxf{p} \comp
    \oxf{\sigma}})} \label{ps5}\\
    & = \ridm{\ridm{\amalg_1^\dag \comp \oxf{p} \comp \oxf{\sigma}}} \vee
    \ridm{\ridm{\amalg_2^\dag \comp \oxf{p} \comp \oxf{\sigma}}} \label{ps6}\\
    & = \ridm{\amalg_1^\dag \comp \oxf{p} \comp \oxf{\sigma}} \vee
    \ridm{\amalg_2^\dag \comp \oxf{p} \comp \oxf{\sigma}} \label{ps7}
  \end{align}
  where \eqref{ps1} follows by Lemma~\ref{lem:basic_ext}, \eqref{ps2} by
  distributivity of composition of joins, \eqref{ps3} by the fourth axiom of
  restriction categories (see Definition~\ref{def:rest_cat}), \eqref{ps4} by
  distributivity of restriction over joins (see
  Definition~\ref{def:join_inv_cat}), \eqref{ps5} by Lemma~\ref{lem:basic_ridm},
  \eqref{ps6} since $\ridm{\amalg_1 \comp \oxf{\sigma}} = \ridm{\amalg_2 \comp
  \oxf{\sigma}} = \id_I$ follows by \eqref{prtot:const}, and \eqref{ps7} by
  Lemma~\ref{lem:basic_ridm}.
  We may establish by analogous argument that $$\ridm{\oxf{q} \comp
  \oxf{\sigma}} = \ridm{\amalg_1^\dag \comp \oxf{q} \comp \oxf{\sigma}} \vee
  \ridm{\amalg_2^\dag \comp \oxf{q} \comp \oxf{\sigma}}$$ as well. In the
  following, let $\sigma_p = \oxf{p} \comp \oxf{\sigma}$ and $\sigma_q =
  \oxf{q} \comp \oxf{\sigma}$. We start by computing
  \begin{align}
    \oxf{\sand{p}{q}} \comp \oxf{\sigma} & = ((\amalg_1 \comp
    \ridm{\amalg_1^\dag \comp \oxf{p}} ~\comp \ridm{\amalg_1^\dag \comp
    \oxf{q}}) \vee (\amalg_2 \comp (\ridm{\amalg_2^\dag \comp
    \oxf{p}} \vee \ridm{\amalg_2^\dag \comp \oxf{q}}))) \comp
    \ridm{\oxf{p}} ~\comp \ridm{\oxf{q}} \comp \oxf{\sigma} \label{paqs1}\\
    & = ((\amalg_1 \comp \ridm{\amalg_1^\dag \comp \oxf{p}} ~\comp
    \ridm{\amalg_1^\dag \comp \oxf{q}}) \vee (\amalg_2 \comp
    (\ridm{\amalg_2^\dag \comp \oxf{p}} \vee \ridm{\amalg_2^\dag \comp
    \oxf{q}}))) \comp \oxf{\sigma} \comp \ridm{\oxf{p} \comp
    \oxf{\sigma}} ~\comp \ridm{\oxf{q} \comp \oxf{\sigma}} \label{paqs2}\\
    & = ((\amalg_1 \comp \ridm{\amalg_1^\dag \comp \oxf{p}} ~\comp
    \ridm{\amalg_1^\dag \comp \oxf{q}} \comp \oxf{\sigma}) \vee
    (\amalg_2 \comp (\ridm{\amalg_2^\dag \comp \oxf{p}} \vee
    \ridm{\amalg_2^\dag \comp \oxf{q}}) \comp \oxf{\sigma})) \comp
    \ridm{\sigma_p} ~\comp \ridm{\sigma_q} \label{paqs3}\\
    & = ((\amalg_1 \comp \oxf{\sigma} \comp \ridm{\amalg_1^\dag \comp
    \oxf{p} \comp \oxf{\sigma}} ~\comp \ridm{\amalg_1^\dag \comp \oxf{q} \comp
    \oxf{\sigma}}) \vee \nonumber \\ 
    & \qquad (\amalg_2 \comp \oxf{\sigma} \comp
    (\ridm{\amalg_2^\dag \comp \oxf{p} \comp \oxf{\sigma}} \vee
    \ridm{\amalg_2^\dag \comp \oxf{q} \comp \oxf{\sigma}}))) \comp
    \ridm{\sigma_p} ~\comp \ridm{\sigma_q} \label{paqs4}\\
    & = ((\amalg_1 \comp \oxf{\sigma} \comp \ridm{\amalg_1^\dag \comp
    \sigma_p} ~\comp \ridm{\amalg_1^\dag \comp \sigma_q}) \vee
    (\amalg_2 \comp \oxf{\sigma} \comp (\ridm{\amalg_2^\dag \comp
    \sigma_p} \vee \ridm{\amalg_2^\dag \comp \sigma_q}))) \comp
    \ridm{\sigma_p} ~\comp \ridm{\sigma_q} \label{paqs5}
  \end{align}
  where \eqref{paqs1} follows by the definition of $\oxf{\sand{p}{q}}$ (see
  Definition~\ref{def:predicates}), \eqref{paqs2} by two applications of the
  fourth axiom of restriction categories (see Definition~\ref{def:rest_cat}),
  \eqref{paqs3} by distributivity of composition over joins (see
  Definition~\ref{def:join_inv_cat}) and the definitions of $\sigma_p$ and
  $\sigma_q$, \eqref{paqs4} by repeated applications of the fourth axiom of
  restriction categories, and \eqref{paqs5} by definition of $\sigma_p$ and
  $\sigma_q$.
  But then we have
  \normalsize
  \begin{align}
    \ridm{\oxf{\sand{p}{q}} \comp \oxf{\sigma}} 
    & = \ridm{((\amalg_1 \comp \oxf{\sigma} \comp \ridm{\amalg_1^\dag \comp
    \sigma_p} ~\comp \ridm{\amalg_1^\dag \comp \sigma_q}) \vee
    (\amalg_2 \comp \oxf{\sigma} \comp (\ridm{\amalg_2^\dag \comp
    \sigma_p} \vee \ridm{\amalg_2^\dag \comp \sigma_q}))) \comp
    \ridm{\sigma_p} ~\comp \ridm{\sigma_q}} \label{rpaqs1}\\
    & = \ridm{\ridm{((\amalg_1 \comp \oxf{\sigma} \comp \ridm{\amalg_1^\dag
    \comp \sigma_p} ~\comp \ridm{\amalg_1^\dag \comp \sigma_q}) \vee
    (\amalg_2 \comp \oxf{\sigma} \comp (\ridm{\amalg_2^\dag \comp
    \sigma_p} \vee \ridm{\amalg_2^\dag \comp \sigma_q})))} \comp
    \ridm{\sigma_p} ~\comp \ridm{\sigma_q}} \label{rpaqs2}\\
    & = \ridm{\ridm{((\amalg_1 \comp \oxf{\sigma} \comp \ridm{\amalg_1^\dag
    \comp \sigma_p} ~\comp \ridm{\amalg_1^\dag \comp \sigma_q})} \vee
    \ridm{(\amalg_2 \comp \oxf{\sigma} \comp (\ridm{\amalg_2^\dag
    \comp \sigma_p} \vee \ridm{\amalg_2^\dag \comp \sigma_q})))}
    \comp \ridm{\sigma_p} ~\comp \ridm{\sigma_q}} \label{rpaqs3}\\
    & = \ridm{((\ridm{\amalg_1 \comp \oxf{\sigma}} \comp
    \ridm{\amalg_1^\dag \comp \sigma_p} ~\comp \ridm{\amalg_1^\dag \comp
    \sigma_q}) \vee (\ridm{\amalg_2 \comp \oxf{\sigma}} \comp
    (\ridm{\amalg_2^\dag \comp \sigma_p} \vee \ridm{\amalg_2^\dag \comp
    \sigma_q}))) \comp \ridm{\sigma_p} ~\comp \ridm{\sigma_q}} 
    \label{rpaqs4}\\
    & = ((\ridm{\amalg_1^\dag \comp \sigma_p} ~\comp
    \ridm{\amalg_1^\dag \comp \sigma_q}) \vee (\ridm{\amalg_2^\dag
    \comp \sigma_p} \vee \ridm{\amalg_2^\dag \comp \sigma_q} ))
    \comp \ridm{\sigma_p} ~\comp \ridm{\sigma_q} \label{rpaqs5}\\
    & = (\ridm{\amalg_1^\dag \comp \sigma_p} ~\comp \ridm{\amalg_1^\dag
    \comp \sigma_q} ~\comp \ridm{\sigma_p} ~\comp \ridm{\sigma_q}) \vee
    (\ridm{\amalg_2^\dag \comp \sigma_p} ~\comp \ridm{\sigma_p} ~\comp
    \ridm{\sigma_q} ) \vee (\ridm{\amalg_2^\dag \comp \sigma_q}
    ~\comp \ridm{\sigma_p} ~\comp \ridm{\sigma_q} ) \label{rpaqs7}\\
    & = (\ridm{\amalg_1^\dag \comp \sigma_p} ~\comp \ridm{\sigma_p} ~\comp
    \ridm{\amalg_1^\dag \comp \sigma_q} ~\comp \ridm{\sigma_q}) \vee
    (\ridm{\amalg_2^\dag \comp \sigma_p} ~\comp \ridm{\sigma_p} ~\comp
    \ridm{\sigma_q} ) \vee (\ridm{\amalg_2^\dag \comp \sigma_q} 
    ~\comp \ridm{\sigma_q} ~\comp \ridm{\sigma_p} ) \label{rpaqs8}\\
    & = (\ridm{\amalg_1^\dag \comp \sigma_p} ~\comp
    \ridm{\amalg_1^\dag \comp \sigma_q} ) \vee
    (\ridm{\amalg_2^\dag \comp \sigma_p} ~\comp
    \ridm{\sigma_q} ) \vee (\ridm{\amalg_2^\dag \comp \sigma_q} 
    ~\comp \ridm{\sigma_p} ) \label{rpaqs9}\\
    & = (\ridm{\amalg_1^\dag \comp \sigma_p} ~\comp \ridm{\amalg_1^\dag
    \comp \sigma_q} ) \vee (\ridm{\amalg_2^\dag \comp \sigma_p}
    \comp (\ridm{\amalg_1^\dag \comp \sigma_q} \vee \ridm{\amalg_2^\dag
    \comp \sigma_q}) ) \vee (\ridm{\amalg_2^\dag \comp
    \sigma_q} \comp (\ridm{\amalg_1^\dag \comp \sigma_p} \vee
    \ridm{\amalg_2^\dag \comp \sigma_p}) ) \label{rpaqs10}\\
    & = (\ridm{\amalg_1^\dag \comp \sigma_p} ~\comp \ridm{\amalg_1^\dag
    \comp \sigma_q} ) \vee (\ridm{\amalg_2^\dag \comp
    \sigma_p} ~\comp\ridm{\amalg_1^\dag \comp \sigma_q}) \vee (
    \ridm{\amalg_2^\dag \comp \sigma_p} ~\comp\ridm{\amalg_2^\dag \comp
    \sigma_q}) \vee (\ridm{\amalg_2^\dag \comp
    \sigma_q} ~\comp\ridm{\amalg_1^\dag \comp \sigma_p}) \vee \nonumber 
    \\ 
    & \qquad ( \ridm{\amalg_2^\dag \comp \sigma_q}
    ~\comp\ridm{\amalg_2^\dag \comp \sigma_p}) \label{rpaqs11}\\
    & = (\ridm{\amalg_1^\dag \comp \sigma_p} ~\comp \ridm{\amalg_1^\dag
    \comp \sigma_q} ) \vee (\ridm{\amalg_1^\dag \comp
    \sigma_p} ~\comp\ridm{\amalg_2^\dag \comp \sigma_q}) \vee (
    \ridm{\amalg_2^\dag \comp \sigma_p} ~\comp\ridm{\amalg_1^\dag \comp
    \sigma_q}) \vee (\ridm{\amalg_2^\dag \comp
    \sigma_p} ~\comp\ridm{\amalg_2^\dag \comp \sigma_q}) \label{rpaqs12}\\
    & = (\ridm{\amalg_1^\dag \comp \sigma_p} \vee \ridm{\amalg_2^\dag
    \comp \sigma_p}) \comp (\ridm{\amalg_1^\dag \comp \sigma_q} \vee
    \ridm{\amalg_2^\dag \comp \sigma_q})
    = \ridm{\sigma_p} ~\comp \ridm{\sigma_q} \label{rpaqs13}
  \end{align}
  \normalsize
  Here, \eqref{rpaqs1} follows by \eqref{paqs1}--\eqref{paqs5}, \eqref{rpaqs2} 
  by
  Lemma~\ref{lem:basic_ridm}, \eqref{rpaqs3} by distributivity of restriction
  over joins, \eqref{rpaqs4} by the third axiom of restriction categories,
  \eqref{rpaqs5} follows by \eqref{prtot:const}, \eqref{rpaqs7} by
  distributivity of composition over joins, \eqref{rpaqs8} by commutativity of
  restriction idempotents (see Definition~\ref{def:rest_cat}), \eqref{rpaqs9}
  by Lemma~\ref{lem:basic_ridm}, \eqref{rpaqs10} by Lemma~\ref{lem:basic_ext},
  \eqref{rpaqs11} by distributivity of composition over joins, \eqref{rpaqs12}
  by idempotence of joins, and \eqref{rpaqs13} by distributivity of composition
  over joins and definition of $\sigma_p$ and $\sigma_q$.
\end{proof}

A common way to show computational soundness (see, \eg, \cite{Fiore1994}) is to
show a kind of preservation property; that interpretations are, in a sense,
preserved across evaluation in the operational semantics. This is shown for
predicates in the following lemma:

\begin{lem}\label{lem:preservation_pred}
  If $\bigstepp{\sigma}{p}{b}$ then $\oxf{p} \comp \oxf{\sigma} = \oxf{b} \comp 
  \oxf{\sigma}$.
\end{lem}
\begin{proof}
  By induction on the structure of the derivation $\mathcal{D}$ of 
  $\bigstepp{\sigma}{p}{b}$.
  \begin{itemize}
    \item Case $\mathcal{D} = \dfrac{}{\bigstepp{\sigma}{\stt}{\stt}}$. 
    We trivially have $\oxf{\stt} \comp \oxf{\sigma} = \oxf{\stt} \comp 
    \oxf{\sigma}$.
    \item Case $\mathcal{D} = \dfrac{}{\bigstepp{\sigma}{\sff}{\sff}}$. 
    Again, we trivially have $\oxf{\sff} \comp \oxf{\sigma} = \oxf{\sff} \comp 
    \oxf{\sigma}$.
    \item Case $\mathcal{D} = \dfrac{\bigstepp{\sigma}{p}{\stt}}
    {\bigstepp{\sigma}{\snot{p}}{\sff}}$. \\[0.5\baselineskip]
    By induction we have that $\oxf{p}
    \comp \oxf{\sigma} = \oxf{\stt} \comp \oxf{\sigma} = \amalg_1 \comp
    \oxf{\sigma}$. But then
    $\oxf{\snot{p}} \comp \oxf{\sigma} = \gamma \comp \oxf{p} \comp
    \oxf{\sigma} = \gamma \comp \amalg_1 \comp \oxf{\sigma} = \amalg_2 \comp
    \oxf{\sigma} = \oxf{\sff} \comp \oxf{\sigma}$.
    \item Case $\mathcal{D} = \dfrac{\bigstepp{\sigma}{p}{\sff}}
    {\bigstepp{\sigma}{\snot{p}}{\stt}}$.\\[0.5\baselineskip] 
    By induction we have that $\oxf{p}
    \comp \oxf{\sigma} = \oxf{\sff} \comp \oxf{\sigma} = \amalg_2 \comp
    \oxf{\sigma}$. Thus
    $\oxf{\snot{p}} \comp \oxf{\sigma} = \gamma \comp \oxf{p} \comp
    \oxf{\sigma} = \gamma \comp \amalg_2 \comp \oxf{\sigma} = \amalg_1 \comp
    \oxf{\sigma} = \oxf{\stt} \comp \oxf{\sigma}$.
    \item Case $\mathcal{D} = \dfrac{\bigstepp{\sigma}{p}{\stt} \qquad 
    \bigstepp{\sigma}{q}{\stt}}{\bigstepp{\sigma}{\sand{p}{q}}{\stt}}$. 
    \\[0.5\baselineskip] By induction $\oxf{p} \comp \oxf{\sigma} = \oxf{\stt} 
    \comp \oxf{\sigma} = \amalg_1 \comp \oxf{\sigma}$ and $ \oxf{q} \comp 
    \oxf{\sigma} = \oxf{\stt} \comp \oxf{\sigma} = \amalg_1 \comp \oxf{\sigma}$.
    We compute
    \begin{align}
      \oxf{\sand{p}{q}} \comp \oxf{\sigma} & = ((\amalg_1 \comp
      \ridm{\amalg_1^\dag \comp \oxf{p}} ~\comp \ridm{\amalg_1^\dag \comp
      \oxf{q}}) \vee (\amalg_2 \comp (\ridm{\amalg_2^\dag \comp
      \oxf{p}} \vee \ridm{\amalg_2^\dag \comp \oxf{q}}))) \comp
      \ridm{\oxf{p}} ~\comp \ridm{\oxf{q}} \comp \oxf{\sigma} \label{sand1_1}\\
      & = ((\amalg_1 \comp \oxf{\sigma} \comp \ridm{\amalg_1^\dag \comp
      \oxf{p} \comp \oxf{\sigma}} ~\comp \ridm{\amalg_1^\dag \comp \oxf{q}
      \comp \oxf{\sigma}}) \vee \nonumber \\ & \qquad (\amalg_2 
      \comp \oxf{\sigma} \comp
      (\ridm{\amalg_2^\dag \comp \oxf{p} \comp \oxf{\sigma}} \vee
      \ridm{\amalg_2^\dag \comp \oxf{q} \comp \oxf{\sigma}}))) \comp 
      \ridm{\oxf{p} \comp \oxf{\sigma}} ~\comp
      \ridm{\oxf{q} \comp \oxf{\sigma}} \label{sand1_2}\\
      & = ((\amalg_1 \comp \oxf{\sigma} \comp \ridm{\amalg_1^\dag \comp
      \amalg_1 \comp \oxf{\sigma}} ~\comp \ridm{\amalg_1^\dag \comp \amalg_1
      \comp \oxf{\sigma}}) \vee \nonumber \\ & \qquad (\amalg_2
      \comp \oxf{\sigma} \comp (\ridm{\amalg_2^\dag \comp \amalg_1 \comp
      \oxf{\sigma}} \vee \ridm{\amalg_2^\dag \comp \amalg_1 \comp
      \oxf{\sigma}}))) \comp \ridm{\amalg_1 \comp \oxf{\sigma}}
      ~\comp \ridm{\amalg_1 \comp \oxf{\sigma}} \label{sand1_3}\\
      & = ((\amalg_1 \comp \oxf{\sigma} \comp \ridm{\oxf{\sigma}} ~\comp 
      \ridm{\oxf{\sigma}}) \vee (\amalg_2 \comp \oxf{\sigma} \comp
      (\ridm{0_{\Sigma,\Sigma} \comp \oxf{\sigma}} \vee
      \ridm{0_{\Sigma,\Sigma} \comp \oxf{\sigma}}))) \comp
      \ridm{\oxf{\sigma}} ~\comp \ridm{\oxf{\sigma}} \label{sand1_4}\\
      & = ((\amalg_1 \comp \oxf{\sigma} \comp \comp 
      \ridm{\oxf{\sigma}}) \vee (\amalg_2 \comp \oxf{\sigma} \comp
      \ridm{0_{\Sigma,\Sigma} \comp \oxf{\sigma}})) \comp
      \ridm{\oxf{\sigma}} \label{sand1_4.5}\\
      & = ((\amalg_1 \comp \oxf{\sigma}) \vee (\amalg_2
      \comp \oxf{\sigma} \comp \ridm{0_{I,\Sigma}})) \comp
      \ridm{\oxf{\sigma}} \label{sand1_5}\\
      & = ((\amalg_1 \comp \oxf{\sigma}) \vee (\amalg_2
      \comp \oxf{\sigma} \comp 0_{I,I})) \comp \ridm{\oxf{\sigma}}
      \label{sand1_6}\\
      & = ((\amalg_1 \comp \oxf{\sigma}) \vee
      0_{I,\Sigma\oplus\Sigma}) \comp
      \ridm{\oxf{\sigma}} \label{sand1_7}\\
      & = \amalg_1 \comp \oxf{\sigma} \comp \ridm{\oxf{\sigma}} = \amalg_1 
      \comp \oxf{\sigma} = \oxf{\stt} \comp \oxf{\sigma}. \label{sand1_8}
    \end{align}
    where \eqref{sand1_1} is the definition of $\oxf{\sand{p}{q}}$, 
    \eqref{sand1_2} follows by \eqref{paqs1} -- \eqref{paqs4}, \eqref{sand1_3} 
    by $\oxf{p} \comp \oxf{\sigma} = \oxf{q} \comp \oxf{\sigma} = \amalg_1
    \comp \oxf{\sigma}$, \eqref{sand1_4} by Lemma~\ref{lem:basic_disj}, 
    \eqref{sand1_4.5} by Lemma~\ref{lem:basic_ridm} and idempotence of joins, 
    \eqref{sand1_5} by the first axiom of restriction categories and the
    universal mapping property of the zero object, \eqref{sand1_6} by the zero
    object a restriction zero, \eqref{sand1_7} once again by the universal
    mapping property of the zero object, and \eqref{sand1_8} by the first axiom
    of restriction categories and the definition of $\oxf{\stt}$.

    \item Case $\mathcal{D} = \dfrac{\bigstepp{\sigma}{p}{\sff} \qquad 
    \bigstepp{\sigma}{q}{\stt}}{\bigstepp{\sigma}{\sand{p}{q}}{\sff}}$. 
    \\[0.5\baselineskip] By induction $\oxf{p} \comp \oxf{\sigma} = \oxf{\sff} 
    \comp \oxf{\sigma} = \amalg_2 \comp \oxf{\sigma}$ and $ \oxf{q} \comp 
    \oxf{\sigma} = \oxf{\stt} \comp \oxf{\sigma} = \amalg_1 \comp \oxf{\sigma}$.
    
    \begin{align}
      \oxf{\sand{p}{q}} \comp \oxf{\sigma} & = ((\amalg_1 \comp
      \ridm{\amalg_1^\dag \comp \oxf{p}} ~\comp \ridm{\amalg_1^\dag \comp
      \oxf{q}}) \vee (\amalg_2 \comp (\ridm{\amalg_2^\dag \comp
      \oxf{p}} \vee \ridm{\amalg_2^\dag \comp \oxf{q}}))) \comp
      \ridm{\oxf{p}} ~\comp \ridm{\oxf{q}} \comp \oxf{\sigma} \label{sand3_1}\\
      & = ((\amalg_1 \comp \oxf{\sigma} \comp \ridm{\amalg_1^\dag \comp
      \oxf{p} \comp \oxf{\sigma}} ~\comp \ridm{\amalg_1^\dag \comp \oxf{q}
      \comp \oxf{\sigma}}) \vee \nonumber \\ & \qquad (\amalg_2 \comp
      \oxf{\sigma} \comp (\ridm{\amalg_2^\dag \comp \oxf{p} \comp
      \oxf{\sigma}} \vee \ridm{\amalg_2^\dag \comp \oxf{q} \comp
      \oxf{\sigma}}))) \comp \ridm{\oxf{p} \comp \oxf{\sigma}}
      ~\comp \ridm{\oxf{q} \comp \oxf{\sigma}} \label{sand3_2}\\
      & = ((\amalg_1 \comp \oxf{\sigma} \comp \ridm{\amalg_1^\dag \comp
      \amalg_2 \comp \oxf{\sigma}} ~\comp \ridm{\amalg_1^\dag \comp \amalg_1
      \comp \oxf{\sigma}}) \vee \nonumber \\ & \qquad (\amalg_2 \comp
      \oxf{\sigma} \comp (\ridm{\amalg_2^\dag \comp \amalg_2 \comp
      \oxf{\sigma}} \vee \ridm{\amalg_2^\dag \comp \amalg_1 \comp
      \oxf{\sigma}}))) \comp \ridm{\amalg_2 \comp \oxf{\sigma}}
      ~\comp \ridm{\amalg_1 \comp \oxf{\sigma}} \label{sand3_3}\\
      & = ((\amalg_1 \comp \oxf{\sigma} \comp \ridm{0_{\Sigma,\Sigma} \comp
      \oxf{\sigma}} ~\comp \ridm{\oxf{\sigma}}) \vee (\amalg_2 \comp
      \oxf{\sigma} \comp (\ridm{\oxf{\sigma}} \vee \ridm{0_{\Sigma,\Sigma}
      \comp \oxf{\sigma}}))) \comp \ridm{\oxf{\sigma}} ~\comp
      \ridm{\oxf{\sigma}} \label{sand3_4}\\
      & = ((\amalg_1 \comp \oxf{\sigma} \comp 0_{I,I} ~\comp
      \ridm{\oxf{\sigma}}) \vee (\amalg_2 \comp \oxf{\sigma} \comp
      (\ridm{\oxf{\sigma}} \vee 0_{I,I}))) \comp
      \ridm{\oxf{\sigma}} \label{sand3_5}\\
      & = (0_{I,\Sigma \oplus \Sigma} \vee (\amalg_2 \comp
      \oxf{\sigma} \comp \ridm{\oxf{\sigma}})) \comp
      \ridm{\oxf{\sigma}} \label{sand3_6}\\
      & = \amalg_2 \comp \oxf{\sigma} \comp \ridm{\oxf{\sigma}} ~\comp 
      \ridm{\oxf{\sigma}} = \amalg_2 \comp \oxf{\sigma} = \oxf{\sff} \comp 
      \oxf{\sigma} \label{sand3_7}
    \end{align}
    where \eqref{sand3_1} is the definition of $\oxf{\sand{p}{q}}$,
    \eqref{sand3_2} follows by \eqref{paqs1} -- \eqref{paqs4}, \eqref{sand3_3}
    by $\oxf{p} \comp \oxf{\sigma} = \amalg_2 \comp \oxf{\sigma}$ and $\oxf{q}
    \comp \oxf{\sigma} = \amalg_1 \comp \oxf{\sigma}$, \eqref{sand3_4} by
    Lemma~\ref{lem:basic_disj}, \eqref{sand3_5} by the first axiom of
    restriction categories as well as the universal mapping property of the
    zero object and the fact that it is a restriction zero, \eqref{sand3_6} by
    the universal mapping property of the zero object and the fact that zero
    maps are unit for joins, \eqref{sand3_7} by zero maps units for joins, the
    first axiom of restriction categories, and the definition of $\oxf{\sff}$.

    \item Case $\mathcal{D} = \dfrac{\bigstepp{\sigma}{p}{\stt} \qquad
    \bigstepp{\sigma}{q}{\sff}}{\bigstepp{\sigma}{\sand{p}{q}}{\sff}}$, similar 
    to the previous case.
    \item Case $\mathcal{D} = \dfrac{\bigstepp{\sigma}{p}{\sff} \qquad
    \bigstepp{\sigma}{q}{\sff}}{\bigstepp{\sigma}{\sand{p}{q}}{\sff}}$.
    \\[0.5\baselineskip] By induction $\oxf{p} \comp \oxf{\sigma} = \oxf{\sff} 
    \comp \oxf{\sigma} = \amalg_2 \comp \oxf{\sigma}$ and $ \oxf{q} \comp 
    \oxf{\sigma} = \oxf{\sff} \comp \oxf{\sigma} = \amalg_2 \comp \oxf{\sigma}$.
    \begin{align}
      \oxf{\sand{p}{q}} \comp \oxf{\sigma} & = ((\amalg_1 \comp
      \ridm{\amalg_1^\dag \comp \oxf{p}} ~\comp \ridm{\amalg_1^\dag \comp
      \oxf{q}}) \vee (\amalg_2 \comp (\ridm{\amalg_2^\dag \comp
      \oxf{p}} \vee \ridm{\amalg_2^\dag \comp \oxf{q}}))) \comp
      \ridm{\oxf{p}} ~\comp \ridm{\oxf{q}} \comp \oxf{\sigma} \label{sand4_1}\\
      & = ((\amalg_1 \comp \oxf{\sigma} \comp \ridm{\amalg_1^\dag \comp
      \oxf{p} \comp \oxf{\sigma}} ~\comp \ridm{\amalg_1^\dag \comp \oxf{q}
      \comp \oxf{\sigma}}) \vee \nonumber \\ & \qquad (\amalg_2 
      \comp \oxf{\sigma} \comp (\ridm{\amalg_2^\dag \comp \oxf{p} \comp
      \oxf{\sigma}} \vee \ridm{\amalg_2^\dag \comp \oxf{q} \comp
      \oxf{\sigma}}))) \comp \ridm{\oxf{p} \comp \oxf{\sigma}}
      ~\comp \ridm{\oxf{q} \comp \oxf{\sigma}} \label{sand4_2}\\
      & = ((\amalg_1 \comp \oxf{\sigma} \comp \ridm{\amalg_1^\dag \comp
      \amalg_2 \comp \oxf{\sigma}} ~\comp \ridm{\amalg_1^\dag \comp \amalg_2
      \comp \oxf{\sigma}}) \vee \nonumber \\ & \qquad (\amalg_2
      \comp \oxf{\sigma} \comp (\ridm{\amalg_2^\dag \comp \amalg_2 \comp
      \oxf{\sigma}} \vee \ridm{\amalg_2^\dag \comp \amalg_2 \comp
      \oxf{\sigma}}))) \comp \ridm{\amalg_2 \comp \oxf{\sigma}}
      ~\comp \ridm{\amalg_2 \comp \oxf{\sigma}} \label{sand4_3}\\
      & = ((\amalg_1 \comp
      \oxf{\sigma} \comp \ridm{0_{\Sigma,\Sigma} \comp \oxf{\sigma}} ~\comp
      \ridm{0_{\Sigma,\Sigma} \comp \oxf{\sigma}}) \vee (\amalg_2
      \comp \oxf{\sigma} \comp (\ridm{\oxf{\sigma}} \vee
      \ridm{\oxf{\sigma}}))) \comp \ridm{\oxf{\sigma}} ~\comp
      \ridm{\oxf{\sigma}} \label{sand4_4}\\
      & = ((\amalg_1 \comp \oxf{\sigma} \comp 0_{I,I}) \vee
      (\amalg_2 \comp \oxf{\sigma} \comp \ridm{\oxf{\sigma}})) \comp
      \ridm{\oxf{\sigma}} \label{sand4_5}\\
      & = (0_{I,\Sigma \oplus \Sigma} \vee (\amalg_2
      \comp \oxf{\sigma} \comp \ridm{\oxf{\sigma}})) \comp
      \ridm{\oxf{\sigma}} \label{sand4_6}\\
      & = \amalg_2 \comp \oxf{\sigma} \comp \ridm{\oxf{\sigma}} ~\comp
      \ridm{\oxf{\sigma}} = \amalg_2 \comp \oxf{\sigma} = \oxf{\sff} \comp
      \oxf{\sigma} \label{sand4_7}
    \end{align}
    where \eqref{sand4_1} is the definition of $\oxf{\sand{p}{q}} \comp 
    \oxf{\sigma}$, \eqref{sand4_2} follows by \eqref{paqs1} -- \eqref{paqs4}, 
    \eqref{sand4_3} by $\oxf{p} \comp \oxf{\sigma} = \oxf{q} \comp 
    \oxf{\sigma} = \amalg_2 \comp \oxf{\sigma}$, \eqref{sand4_4} by 
    Lemma~\ref{lem:basic_disj}, \eqref{sand4_5} by idempotence of restriction 
    idempotents and joins, \eqref{sand4_6} by the universal mapping property of 
    the zero object, and \eqref{sand4_7} by the first axiom of restriction 
    categories and the definition of $\oxf{\sff}$. \qedhere
  \end{itemize}
\end{proof}

With this done, the computational soundness lemma for predicates follows readily.

\begin{lem}
  If there exists $b$ such that $\bigstepp{\sigma}{p}{b}$ then $\oxf{p} \comp 
  \oxf{\sigma}$ is total.
\end{lem}
\begin{proof}
  Suppose there exists $b$ such that $\bigstepp{\sigma}{p}{b}$ by some
  derivation. It follows by the operational semantics that $b$
  must be either $\stt$ or $\sff$, and in either case it follows by
  Lemma~\ref{lem:total_pred}~\eqref{prtot:const} that $\oxf{b} \comp
  \oxf{\sigma}$ is total, \ie, $\ridm{\oxf{b} \comp \oxf{\sigma}} = \id_I$. 
  Applying the derivation of $\bigstepp{\sigma}{p}{b}$ to
  Lemma~\ref{lem:preservation_pred} yields that $\oxf{p} \comp \oxf{\sigma} =
  \oxf{b} \comp \oxf{\sigma}$, so specifically $\ridm{\oxf{p} \comp
  \oxf{\sigma}} = \ridm{\oxf{b} \comp \oxf{\sigma}} = \id_I$, as desired.
\end{proof}

Adequacy for predicates can then be shown by induction on the structure of the
predicate, and by letting Lemma~\ref{lem:total_pred} (regarding the totality of predicates) do much of the heavy lifting.

\begin{lem}\label{lem:adequacy_pred}
  If $\oxf{p} \comp \oxf{\sigma}$ is total then there exists $b$ such that 
  $\bigstepp{\sigma}{p}{b}$.
\end{lem}
\begin{proof}
  By induction on the structure of $p$.
  \begin{itemize}
    \item Case $p = \stt$. Then $\bigstepp{\sigma}{\stt}{\stt}$ by 
    $\dfrac{}{\bigstepp{\sigma}{\stt}{\stt}}$.
    \item Case $p = \sff$. Then $\bigstepp{\sigma}{\sff}{\sff}$ by 
    $\dfrac{}{\bigstepp{\sigma}{\sff}{\sff}}$.
    \item Case $p = \snot{p'}$. Since $\oxf{\snot{p'}} \comp \oxf{\sigma}$ is 
    total, it follows by Lemma~\ref{lem:total_pred} that $\oxf{p'} \comp 
    \oxf{\sigma}$ is total as well, so by induction, there exists $b$ such that 
    $\bigstepp{\sigma}{p'}{b}$ by some derivation $\mathcal{D}$. We have two 
    cases to consider: If $b = \stt$, $\mathcal{D}$ is a derivation of 
    $\bigstepp{\sigma}{p'}{\stt}$, and so we may derive 
    $\bigstepp{\sigma}{\snot{p'}}{\sff}$ by
    \begin{equation*}
      \dfrac{\overset{\mathcal{D}}{\bigstepp{\sigma}{p'}{\stt}}}
      {\bigstepp{\sigma}{\snot{p'}}{\sff}}
    \end{equation*}
    If on the other hand $b = \sff$, $\mathcal{D}$ is a derivation of 
    $\bigstepp{\sigma}{p'}{\sff}$, and we may use the other $\mathbf{not}$-rule 
    with $\mathcal{D}$ to derive
    \begin{equation*}
      \dfrac{\overset{\mathcal{D}}{\bigstepp{\sigma}{p'}{\sff}}}
      {\bigstepp{\sigma}{\snot{p'}}{\stt}} \enspace.
    \end{equation*}
    \item Case $p = \sand{q}{r}$. Since we have that $\oxf{\sand{q}{r}} \comp 
    \oxf{\sigma}$ is total, by Lemma~\ref{lem:total_pred}, so are $\oxf{q} 
    \comp \oxf{\sigma}$ and $\oxf{r} \comp \oxf{\sigma}$. Thus, it follows by 
    induction that there exist $b_1$ and $b_2$ such that 
    $\bigstepp{\sigma}{q}{b_1}$ respectively $\bigstepp{\sigma}{r}{b_2}$ by
    derivations $\mathcal{D}_1$ respectively $\mathcal{D}_2$. This gives us
    four cases depending on what $b_1$ and $b_2$ are. Luckily, these four cases
    match precisely the four different rules we have for $\mathbf{and}$: For 
    example, if $b_1 = \stt$ and $b_2 = \sff$, we may derive 
    $\bigstepp{\sigma}{\sand{q}{r}}{\sff}$ by
    \begin{equation*}
      \dfrac{\overset{\mathcal{D}_1}{\bigstepp{\sigma}{q}{\stt}} \qquad
      \overset{\mathcal{D}_2}{\bigstepp{\sigma}{r}{\sff}}}
      {\bigstepp{\sigma}{\sand{q}{r}}{\sff}},
    \end{equation*}
    and so on. \qedhere
  \end{itemize}
\end{proof}

With computational soundness and adequacy done for the
predicates, we turn our attention to commands. Before we can show computational
soundness, we will need a technical lemma regarding the denotational behaviour
of loop bodies in states $\sigma$ when the relevant predicates are either true
or false (see Definition~\ref{def:commands} for the definition of the loop body
$\beta[p,c,q]$).

\begin{lem}\label{lem:beta_preds}
  Let $\sigma$ be a state, and $p$ and $q$ be predicates. Then
  \begin{enumerate}
    \item\label{lem:case:tt_tt} If $\bigstepp{\sigma}{p}{\stt}$ and
    $\bigstepp{\sigma}{q}{\stt}$ then $\beta[p,c,q]_{11} \comp \oxf{\sigma} =
    \oxf{\sigma}$,
    \item\label{lem:case:ff_tt} If $\bigstepp{\sigma}{p}{\sff}$ and
    $\bigstepp{\sigma}{q}{\stt}$ then $\beta[p,c,q]_{21} \comp \oxf{\sigma} =
    \oxf{\sigma}$,
    \item\label{lem:case:tt_ff} If $\bigstepp{\sigma}{p}{\stt}$ and
    $\bigstepp{\sigma}{q}{\sff}$ then $\beta[p,c,q]_{12} \comp \oxf{\sigma} =
    \oxf{c} \comp \oxf{\sigma}$, and
    \item\label{lem:case:ff_ff} If $\bigstepp{\sigma}{p}{\sff}$ and
    $\bigstepp{\sigma}{q}{\sff}$ then $\beta[p,c,q]_{22} \comp \oxf{\sigma} =
    \oxf{c} \comp \oxf{\sigma}$.
  \end{enumerate}
  Further, in each case, for all other choices of $i$ and $j$, 
  $\beta[p,c,q]_{ij} \comp \oxf{\sigma} = 0_{I,\Sigma}$.
\end{lem}
\begin{proof}
  For \eqref{lem:case:tt_tt}, suppose $\bigstepp{\sigma}{p}{\stt}$ and
  $\bigstepp{\sigma}{q}{\stt}$, so by Lemma~\ref{lem:preservation_pred},
  $\oxf{p} \comp \oxf{\sigma} = \oxf{\stt} \comp \oxf{\sigma} = \amalg_1 \comp
  \oxf{\sigma}$ and $\oxf{q} \comp \oxf{\sigma} = \oxf{\stt} \comp \oxf{\sigma}
  = \amalg_1 \comp \oxf{\sigma}$. We have
  \begin{align}
    \beta[p,c,q]_{11} \comp \oxf{\sigma} & = \amalg_1^\dag \comp (\id_\Sigma
    \oplus \oxf{c}) \comp \oxf{q} \comp \oxf{p}^\dag \comp \amalg_1 \comp
    \oxf{\sigma} = \amalg_1^\dag \comp (\id_\Sigma \oplus \oxf{c}) \comp
    \oxf{q} \comp \oxf{p}^\dag \comp \oxf{p} \comp \oxf{\sigma}\label{b11_1} \\ 
    & =
    \amalg_1^\dag \comp (\id_\Sigma \oplus \oxf{c}) \comp \oxf{q} \comp
    \ridm{\oxf{p}} \comp \oxf{\sigma} = \amalg_1^\dag \comp (\id_\Sigma \oplus
    \oxf{c}) \comp \oxf{q} \comp \oxf{\sigma} \comp \ridm{\oxf{p} \comp
    \oxf{\sigma}} \label{b11_2} \\ 
    & = \amalg_1^\dag \comp (\id_\Sigma \oplus \oxf{c}) \comp
    \amalg_1 \comp \oxf{\sigma} \comp \ridm{\amalg_1 \comp \oxf{\sigma}} =
    \amalg_1^\dag \comp (\id_\Sigma \oplus \oxf{c}) \comp \amalg_1 \comp
    \oxf{\sigma} \label{b11_3} \\
    & = \amalg_1^\dag \comp \amalg_1 \comp \id_\Sigma \comp \oxf{\sigma}
    = \ridm{\amalg_1} \comp \oxf{\sigma} = \id_\Sigma \comp \oxf{\sigma} = 
    \oxf{\sigma} \label{b11_4}
  \end{align}
  where \eqref{b11_1} follows by definition of $\beta[p,c,q]_{11}$ and $\oxf{p} 
  \comp \oxf{\sigma} = \amalg_1 \comp \oxf{\sigma}$, \eqref{b11_2} by $\oxf{p}$ 
  a partial isomorphism and the fourth axiom of restriction categories, 
  \eqref{b11_3} by $\oxf{p} \comp \oxf{\sigma} = \oxf{q} \comp \oxf{\sigma} =
  \amalg_1 \comp \oxf{\sigma}$ and the first axiom of restriction categories,
  and \eqref{b11_4} by naturality and totality of $\amalg_1$.

  The proof of \eqref{lem:case:ff_tt} is analogous to that of 
  \eqref{lem:case:tt_tt}.
  
  For \eqref{lem:case:tt_ff}, suppose $\bigstepp{\sigma}{p}{\stt}$ and
  $\bigstepp{\sigma}{q}{\sff}$, so by Lemma~\ref{lem:preservation_pred},
  $\oxf{p} \comp \oxf{\sigma} = \oxf{\stt} \comp \oxf{\sigma} = \amalg_1 \comp
  \oxf{\sigma}$ and $\oxf{q} \comp \oxf{\sigma} = \oxf{\sff} \comp \oxf{\sigma}
  = \amalg_2 \comp \oxf{\sigma}$. We compute
  \begin{align}
    \beta[p,c,q]_{12} \comp \oxf{\sigma} & = \amalg_2^\dag \comp (\id_\Sigma
    \oplus \oxf{c}) \comp \oxf{q} \comp \oxf{p}^\dag \comp \amalg_1 \comp
    \oxf{\sigma} = \amalg_2^\dag \comp (\id_\Sigma
    \oplus \oxf{c}) \comp \oxf{q} \comp \oxf{p}^\dag \comp \oxf{p} \comp
    \oxf{\sigma} \label{b12_1}\\
    & = \amalg_2^\dag \comp (\id_\Sigma \oplus \oxf{c}) \comp \oxf{q} \comp
    \ridm{\oxf{p}} \comp \oxf{\sigma} = \amalg_2^\dag \comp (\id_\Sigma \oplus
    \oxf{c}) \comp \oxf{q} \comp \oxf{\sigma} \comp \ridm{\oxf{p} \comp
    \oxf{\sigma}} \label{b12_2}\\
    & = \amalg_2^\dag \comp (\id_\Sigma \oplus \oxf{c}) \comp \amalg_2 \comp
    \oxf{\sigma} \comp \ridm{\amalg_1 \comp \oxf{\sigma}} = \amalg_2^\dag \comp
    \amalg_2 \comp \oxf{c} \comp \oxf{\sigma} \comp \ridm{\ridm{\amalg_1} \comp
    \oxf{\sigma}} \label{b12_3}\\
    & = \ridm{\amalg_2} \comp \oxf{c} \comp \oxf{\sigma} \comp \ridm{\id_\Sigma
    \comp \oxf{\sigma}} = \id_\Sigma \comp \oxf{c} \comp \oxf{\sigma} \comp 
    \ridm{\oxf{\sigma}} = \oxf{c} \comp \oxf{\sigma} \label{b12_4}
  \end{align}
  where \eqref{b12_1} follows by definition of $\beta[p,c,q]_{12}$ and $\oxf{p} 
  \comp \oxf{\sigma} = \amalg_1 \comp \oxf{\sigma}$, \eqref{b12_2} by $\oxf{p}$ 
  a partial isomorphism and the fourth axiom of restriction categories, 
  \eqref{b12_3} by $\oxf{q} \comp \oxf{\sigma} = \amalg_2 \comp \oxf{\sigma}$,
  naturality of $\amalg_2$ and Lemma~\ref{lem:basic_ridm}, and \eqref{b12_4} by
  totality of $\amalg_2$ and the first axiom of restriction categories.
  
  The proof of \eqref{lem:case:ff_ff} is analogous.
  
  To see that in each case, for all other choices of $i,j$, $\beta[p,c,q]_{ij}
  \comp \oxf{\sigma} = 0_{I,\Sigma}$, we show
  a few of the cases where $\bigstepp{\sigma}{p}{\stt}$ and 
  $\bigstepp{\sigma}{q}{\stt}$. The rest follow by the same line of reasoning.
  Recall that when $\bigstepp{\sigma}{p}{\stt}$ and $\bigstepp{\sigma}{q}{\stt}$
  we have $\oxf{p} \comp \oxf{\sigma} = \oxf{\stt} \comp \oxf{\sigma} =
  \amalg_1 \comp \oxf{\sigma}$ and $\oxf{q} \comp \oxf{\sigma} = \oxf{\stt}
  \comp \oxf{\sigma} = \amalg_1 \comp \oxf{\sigma}$.
  \begin{align}
    \beta[p,c,q]_{12} \comp \oxf{\sigma} & = \amalg_2^\dag \comp (\id_\Sigma
    \oplus \oxf{c}) \comp \oxf{q} \comp \oxf{p}^\dag \comp \amalg_1 \comp
    \oxf{\sigma} = \amalg_2^\dag \comp (\id_\Sigma
    \oplus \oxf{c}) \comp \oxf{q} \comp \oxf{p}^\dag \comp \oxf{p} \comp
    \oxf{\sigma} \label{nb12_1}\\
    & = \amalg_2^\dag \comp (\id_\Sigma \oplus \oxf{c}) \comp \oxf{q} \comp
    \ridm{\oxf{p}} \comp \oxf{\sigma} = \amalg_2^\dag \comp (\id_\Sigma \oplus
    \oxf{c}) \comp \oxf{q} \comp \oxf{\sigma} \comp \ridm{\oxf{p} \comp 
    \oxf{\sigma}} \label{nb12_2}\\
    & = \amalg_2^\dag \comp (\id_\Sigma \oplus
    \oxf{c}) \comp \amalg_1 \comp \oxf{\sigma} \comp \ridm{\oxf{p} \comp 
    \oxf{\sigma}} = \amalg_2^\dag \comp \amalg_1 \comp \id_\Sigma \comp
    \oxf{\sigma} \comp \ridm{\oxf{p} \comp \oxf{\sigma}} \label{nb12_3}\\
    & = 0_{\Sigma, \Sigma} \comp \id_\Sigma \comp
    \oxf{\sigma} \comp \ridm{\oxf{p} \comp \oxf{\sigma}} = 
    0_{I,\Sigma}\label{nb12_4}
  \end{align}
  where \eqref{nb12_1} and \eqref{nb12_2} follow as in \eqref{b12_1} and
  \eqref{b12_2}, \eqref{nb12_3} by $\oxf{q} \comp \oxf{\sigma} = \amalg_1 \comp
  \oxf{\sigma}$ and naturality of $\amalg_1$, and \eqref{nb12_4} by
  Lemma~\ref{lem:basic_disj} and the universal mapping property of the zero
  object.
  
  Finally, for $\beta[p,c,q]_{21}$ we have
  \begin{align}
    \beta[p,c,q]_{21} \comp \oxf{\sigma} & = \amalg_1^\dag \comp (\id_\Sigma
    \oplus \oxf{c}) \comp \oxf{q} \comp \oxf{p}^\dag \comp \amalg_2 \comp
    \oxf{\sigma} = \amalg_1^\dag \comp (\id_\Sigma
    \oplus \oxf{c}) \comp \oxf{q} \comp \oxf{p}^\dag \comp \gamma \comp 
    \amalg_1 \comp \oxf{\sigma} \label{nb21_1} \\
    & = \amalg_1^\dag \comp (\id_\Sigma
    \oplus \oxf{c}) \comp \oxf{q} \comp \oxf{p}^\dag \comp \gamma \comp 
    \oxf{p} \comp \oxf{\sigma} = \amalg_1^\dag \comp (\id_\Sigma
    \oplus \oxf{c}) \comp \oxf{q} \comp \oxf{p}^\dag \comp 
    \oxf{\snot{p}} \comp \oxf{\sigma} \label{nb21_2} \\
    & = \amalg_1^\dag \comp (\id_\Sigma \oplus \oxf{c}) \comp \oxf{q} \comp
    0_{\Sigma,\Sigma} \comp \oxf{\sigma} = 0_{I,\Sigma} \label{nb21_3}
  \end{align}
  where \eqref{nb21_1} follows by definition of $\beta[p,c,q]_{21}$ and $\gamma 
  \comp \amalg_1 = \amalg_2$, \eqref{nb21_2} by $\oxf{p} \comp \oxf{\sigma} =
  \amalg_1 \comp \oxf{\sigma}$ and definition of $\oxf{\snot{p}}$, and 
  \eqref{nb21_3} by $\oxf{p}^\dag \comp \oxf{\snot{p}} = 0_{\Sigma,\Sigma}$ and
  the universal mapping property of the zero object.
\end{proof}

With this lemma done, we turn our attention to the preservation lemma for commands in order to show computational soundness.

\begin{lem}\label{lem:preservation_cmd}
  If $\bigstepc{\sigma}{c}{\sigma'}$ then $\oxf{c} \comp \oxf{\sigma} = 
  \oxf{\sigma'}$.
\end{lem}
\begin{proof}
  By induction on the structure of the derivation $\mathcal{D}$ of 
  $\bigstepc{\sigma}{c}{\sigma'}$.
  
  \begin{itemize}
    \item Case $\mathcal{D} = \dfrac{}{\bigstepc{\sigma}{\sskip}{\sigma}}$. We 
    have $\oxf{\sskip} \comp \oxf{\sigma} = \id_\Sigma \comp \oxf{\sigma} =
    \oxf{\sigma}$.
    \item Case $\mathcal{D} = \dfrac{\bigstepc{\sigma}{c_1}{\sigma'} \qquad
    \bigstepc{\sigma'}{c_2}{\sigma''}}
    {\bigstepc{\sigma}{\sseq{c_1}{c_2}}{\sigma''}}$.
    \\[0.5\baselineskip] By induction, $\oxf{c_1} \comp \oxf{\sigma} =
    \oxf{\sigma'}$ and $\oxf{c_2} \comp \oxf{\sigma'} = \oxf{\sigma''}$. But 
    then $$\oxf{\sseq{c_1}{c_2}} \comp \oxf{\sigma} = \oxf{c_2} \comp \oxf{c_1} 
    \comp \oxf{\sigma} = \oxf{c_2} \comp \oxf{\sigma'} = \oxf{\sigma''}$$
    as desired.
    \item Case $\mathcal{D} = \dfrac{
      \bigstepp{\sigma}{p}{\stt} \quad\enspace
      \bigstepc{\sigma}{c_1}{\sigma'} \quad\enspace
      \bigstepp{\sigma'}{q}{\stt}
    }{
      \bigstepc{\sigma}{\sif{p}{c_1}{c_2}{q}}{\sigma'}
    }$.\\[0.5\baselineskip]
    By induction, $\oxf{c_1} \comp \oxf{\sigma} = \oxf{\sigma'}$, and by 
    Lemma~\ref{lem:preservation_pred}, $\oxf{p} \comp \oxf{\sigma} = 
    \oxf{\stt} \comp \oxf{\sigma} = \amalg_1 \comp \oxf{\sigma}$ and $\oxf{q}
    \comp \oxf{\sigma'} = \oxf{\stt} \comp \oxf{\sigma'} = \amalg_1 \comp
    \oxf{\sigma'}$. We compute:
    \begin{align}
      \oxf{\sif{p}{c_1}{c_2}{q}} \comp \oxf{\sigma} & = \oxf{q}^\dag \comp
      (\oxf{c_1} \oplus \oxf{c_2}) \comp \oxf{p} \comp \oxf{\sigma} 
      \label{pre_if1_1}\\
      & = \oxf{q}^\dag \comp (\oxf{c_1} \oplus \oxf{c_2}) \comp \amalg_1 \comp
      \oxf{\sigma} \label{pre_if1_2}\\
      & = \oxf{q}^\dag \comp \amalg_1 \comp \oxf{c_1} \comp \oxf{\sigma}
      = \oxf{q}^\dag \comp \amalg_1 \comp \oxf{\sigma'} 
      = \oxf{q}^\dag \comp \oxf{q} \comp \oxf{\sigma'} \label{pre_if1_3}\\
      & = \ridm{\oxf{q}} \comp \oxf{\sigma'} = \oxf{\sigma'} \comp
      \ridm{\oxf{q} \comp \oxf{\sigma'}} = \oxf{\sigma'} \comp
      \ridm{\amalg_1 \comp \oxf{\sigma'}} \label{pre_if1_4}\\
      & = \oxf{\sigma'} \comp \ridm{\ridm{\amalg_1} \comp \oxf{\sigma'}} = 
      \oxf{\sigma'} \comp \ridm{\id_\Sigma \comp \oxf{\sigma'}} = \oxf{\sigma'} 
      \comp \ridm{\oxf{\sigma'}} = \oxf{\sigma'}\label{pre_if1_5}
    \end{align}
    where \eqref{pre_if1_1} follows by definition of
    $\oxf{\sif{p}{c_1}{c_2}{q}}$; \eqref{pre_if1_2} by $\oxf{p} \comp
    \oxf{\sigma} = \amalg_1 \comp \oxf{\sigma}$; \eqref{pre_if1_3} by
    naturality of $\amalg_1$, $\oxf{c_1} \comp \oxf{\sigma} = \oxf{\sigma'}$,
    and $\oxf{q} \comp \oxf{\sigma'} = \amalg_1 \comp \oxf{\sigma'}$;
    \eqref{pre_if1_4} by $\oxf{q}$ a partial isomorphism, the fourth axiom of
    restriction categories, and $\oxf{q} \comp \oxf{\sigma'} = \amalg_1 \comp
    \oxf{\sigma'}$; and \eqref{pre_if1_5} by Lemmas \ref{lem:basic_ridm} and
    \ref{lem:basic_disj} and the first axiom of restriction categories.

    \item Case $\mathcal{D} = \dfrac{
      \bigstepp{\sigma}{p}{\sff} \quad\enspace
      \bigstepc{\sigma}{c_2}{\sigma'} \quad\enspace
      \bigstepp{\sigma'}{q}{\sff}
    }{
      \bigstepc{\sigma}{\sif{p}{c_1}{c_2}{q}}{\sigma'}
    }$.\\[0.5\baselineskip]
    By induction, $\oxf{c_2} \comp \oxf{\sigma} = \oxf{\sigma'}$, and by 
    Lemma~\ref{lem:preservation_pred}, $\oxf{p} \comp \oxf{\sigma} = 
    \oxf{\sff} \comp \oxf{\sigma} = \amalg_2 \comp \oxf{\sigma}$ and $\oxf{q}
    \comp \oxf{\sigma'} = \oxf{\sff} \comp \oxf{\sigma'} = \amalg_2 \comp
    \oxf{\sigma'}$. We have
    \begin{align}
      \oxf{\sif{p}{c_1}{c_2}{q}} \comp \oxf{\sigma} & = \oxf{q}^\dag \comp
      (\oxf{c_1} \oplus \oxf{c_2}) \comp \oxf{p} \comp \oxf{\sigma} \\
      & = \oxf{q}^\dag \comp (\oxf{c_1} \oplus \oxf{c_2}) \comp \amalg_2 \comp
      \oxf{\sigma} \\
      & = \oxf{q}^\dag \comp \amalg_2 \comp \oxf{c_2} \comp \oxf{\sigma}
      = \oxf{q}^\dag \comp \amalg_2 \comp \oxf{\sigma'} 
      = \oxf{q}^\dag \comp \oxf{q} \comp \oxf{\sigma'} \\
      & = \ridm{\oxf{q}} \comp \oxf{\sigma'} = \oxf{\sigma'} \comp
      \ridm{\oxf{q} \comp \oxf{\sigma'}} = \oxf{\sigma'} \comp
      \ridm{\amalg_2 \comp \oxf{\sigma'}} \\
      & = \oxf{\sigma'} \comp \ridm{\ridm{\amalg_2} \comp \oxf{\sigma'}} = 
      \oxf{\sigma'} \comp \ridm{\id_\Sigma \comp \oxf{\sigma'}} = \oxf{\sigma'} 
      \comp \ridm{\oxf{\sigma'}} = \oxf{\sigma'}.
    \end{align}
    which follows by similar arguments as in \eqref{pre_if1_1} -- 
    \eqref{pre_if1_5}.
    
    \item Case $\mathcal{D} = \dfrac{
      \bigstepp{\sigma}{p}{\stt} \quad\enspace
      \bigstepp{\sigma}{q}{\stt}
    }{
      \bigstepc{\sigma}{\sfrom{p}{c}{q}}{\sigma}
    }$. \\[0.5\baselineskip]
    Since $\bigstepp{\sigma}{p}{\stt}$ and $\bigstepp{\sigma}{q}{\stt}$, by 
    Lemma~\ref{lem:beta_preds} we get $\beta[p,c,q]_{11} \comp \oxf{\sigma} = 
    \oxf{\sigma}$ and $\beta[p,c,q]_{12} \comp \oxf{\sigma} = 0_{I,\Sigma}$,
    and so
    \begin{align}
      & \oxf{\sfrom{p}{c}{q}} \comp \oxf{\sigma} = 
      \Tr_{\Sigma,\Sigma}^\Sigma(\beta[p,c,q]) \label{sfrom1_1} \\
      & \qquad = \left(\beta[p,c,q]_{11} \vee \bigvee_{n \in \omega}
      \beta[p,c,q]_{21} ~\comp \beta[p,c,q]_{22}^n ~\comp
      \beta[p,c,q]_{12}\right) \comp \oxf{\sigma} \label{sfrom1_2} \\
      & \qquad = \left(\left(\beta[p,c,q]_{11} \comp \oxf{\sigma}\right) \vee
      \bigvee_{n \in \omega} \beta[p,c,q]_{21} ~\comp \beta[p,c,q]_{22}^n
      ~\comp \beta[p,c,q]_{12} \comp \oxf{\sigma}\right) \label{sfrom1_3} \\
      & \qquad = \oxf{\sigma} \vee
      \bigvee_{n \in \omega} \beta[p,c,q]_{21} ~\comp \beta[p,c,q]_{22}^n
      ~\comp 0_{I,\Sigma}
      = \oxf{\sigma} \vee 0_{I,\Sigma} = \oxf{\sigma} \label{sfrom1_4}
    \end{align}
    where \eqref{sfrom1_1} follows by definition of $\oxf{\sfrom{p}{c}{q}}$,
    \eqref{sfrom1_2} by definition of the canonical trace in join inverse 
    categories (see Proposition~\ref{prop:trace}), \eqref{sfrom1_3} by 
    distributivity of composition over joins, and \eqref{sfrom1_4} by 
    $\beta[p,c,q]_{11} \comp \oxf{\sigma} = \oxf{\sigma}$, 
    $\beta[p,c,q]_{12} \comp \oxf{\sigma} = 0_{I,\Sigma}$, the universal
    mapping property of the zero object, and the fact that zero maps are units
    for joins.
    
    \item Case $\mathcal{D} = \dfrac{
      \bigstepp{\sigma}{p}{\stt} \quad\enspace
      \bigstepp{\sigma}{q}{\sff} \quad\enspace
      \bigstepc{\sigma}{c}{\sigma'} \quad\enspace
      \bigstepc{\sigma'}{\sloop{p}{c}{q}}{\sigma''}
    }{
      \bigstepc{\sigma}{\sfrom{p}{c}{q}}{\sigma''}
    }$. \\[0.5\baselineskip]
    By induction, $\oxf{c} \comp \oxf{\sigma} = \oxf{\sigma'}$ and 
    $\oxf{\sloop{p}{c}{q}} \comp \oxf{\sigma'} = \oxf{\sigma''}$, and since
    $\bigstepp{\sigma}{p}{\stt}$ and $\bigstepp{\sigma}{q}{\sff}$, by 
    Lemma~\ref{lem:beta_preds} we get $\beta[p,c,q]_{11} \comp \oxf{\sigma} = 
    0_{I,\Sigma}$ and $\beta[p,c,q]_{12} \comp \oxf{\sigma} = 
    \oxf{c} \comp \oxf{\sigma}$. Thus
    \begin{align}
      & \oxf{\sfrom{p}{c}{q}} \comp \oxf{\sigma} =
      \Tr_{\Sigma,\Sigma}^\Sigma(\beta[p,c,q]) \comp \oxf{\sigma}
      \label{sfrom2_1}\\
      & \qquad = \left(\beta[p,c,q]_{11} \vee \bigvee_{n \in \omega}
      \beta[p,c,q]_{21} ~\comp \beta[p,c,q]_{22}^n ~\comp
      \beta[p,c,q]_{12}\right) \comp \oxf{\sigma} \label{sfrom2_2}\\
      & \qquad = \left(\left(\beta[p,c,q]_{11} \comp \oxf{\sigma}\right) \vee
      \bigvee_{n \in \omega} \beta[p,c,q]_{21} ~\comp \beta[p,c,q]_{22}^n
      ~\comp \beta[p,c,q]_{12} \comp \oxf{\sigma}\right) \label{sfrom2_3}\\
      & \qquad = \left(0_{I,\Sigma} \vee
      \bigvee_{n \in \omega} \beta[p,c,q]_{21} ~\comp \beta[p,c,q]_{22}^n
      ~\comp \oxf{c} \comp \oxf{\sigma}\right) \label{sfrom2_4}\\
      & \qquad = \bigvee_{n \in \omega} \beta[p,c,q]_{21} ~\comp
      \beta[p,c,q]_{22}^n ~\comp \oxf{\sigma'} \label{sfrom2_5}\\
      & \qquad = \left(\bigvee_{n \in \omega} \beta[p,c,q]_{21} ~\comp
      \beta[p,c,q]_{22}^n\right) \comp \oxf{\sigma'} \label{sfrom2_6}\\
      & \qquad = \oxf{\sloop{p}{c}{q}} \comp \oxf{\sigma'} = \oxf{\sigma''}
      \label{sfrom2_7}
    \end{align}
    where \eqref{sfrom2_1} follows by the definition of $\oxf{\sfrom{p}{c}{q}}$,
    \eqref{sfrom2_2} by the definition of the canonical trace 
    (Proposition~\ref{prop:trace}), \eqref{sfrom2_3} by distributivity of 
    composition over joins, \eqref{sfrom2_4} by $\beta[p,c,q]_{11} \comp
    \oxf{\sigma} = 0_{I,\Sigma}$ and $\beta[p,c,q]_{12} \comp \oxf{\sigma} =
    \oxf{c} \comp \oxf{\sigma}$, \eqref{sfrom2_5} by $\oxf{\sigma'} = \oxf{c} 
    \comp \oxf{\sigma}$ and the fact that zero maps are units for joins, 
    \eqref{sfrom2_6} by distributivity of composition over joins, and 
    \eqref{sfrom2_7} by definition of $\oxf{\sloop{p}{c}{q}}$ and
    $\oxf{\sloop{p}{c}{q}} \comp \oxf{\sigma'} = \oxf{\sigma''}$.
    
    \item Case $\mathcal{D} = \dfrac{
      \bigstepp{\sigma}{p}{\sff} \quad\enspace
      \bigstepp{\sigma}{q}{\stt}
    }{
      \bigstepc{\sigma}{\sloop{p}{c}{q}}{\sigma}
    }$. \\[0.5\baselineskip]
    Since $\bigstepp{\sigma}{p}{\sff}$ and $\bigstepp{\sigma}{q}{\stt}$, by 
    Lemma~\ref{lem:beta_preds} we get $\beta[p,c,q]_{22} \comp \oxf{\sigma} = 
    0_{I,\Sigma}$ and $\beta[p,c,q]_{21} \comp \oxf{\sigma} = \oxf{\sigma}$. 
    This gives us
    \begin{align}
      \oxf{\sloop{p}{c}{q}} \comp \oxf{\sigma} & = 
      \left(\bigvee_{n \in \omega} \beta[p,c,q]_{21} \comp \beta[p,c,q]_{22}^n 
      \right) \comp \oxf{\sigma} \label{sloop1_1}\\
      & = \left(\beta[p,c,q]_{21} \vee \bigvee_{n \in \omega} \beta[p,c,q]_{21} 
      \comp \beta[p,c,q]_{22}^{n+1} \right) \comp \oxf{\sigma} 
      \label{sloop1_2}\\
      & = \left(\left(\beta[p,c,q]_{21} \comp \oxf{\sigma}\right) \vee 
      \bigvee_{n \in \omega} \beta[p,c,q]_{21} 
      \comp \beta[p,c,q]_{22}^{n+1} \comp \oxf{\sigma} \right) 
      \label{sloop1_3}\\
      & = \left(\oxf{\sigma} \vee 
      \bigvee_{n \in \omega} \beta[p,c,q]_{21} 
      \comp \beta[p,c,q]_{22}^{n} \comp \beta[p,c,q]_{22} \comp \oxf{\sigma} 
      \right) \label{sloop1_4}\\
      & = \left(\oxf{\sigma} \vee 
      \bigvee_{n \in \omega} \beta[p,c,q]_{21} 
      \comp \beta[p,c,q]_{22}^{n} \comp 0_{I,\Sigma} \right) \label{sloop1_5}\\
      & = \oxf{\sigma} \vee 0_{I, \Sigma} = \oxf{\sigma} \label{sloop1_6}
    \end{align}
    where \eqref{sloop1_1} follows by definition of $\oxf{\sloop{p}{c}{q}}$, 
    \eqref{sloop1_2} by unrolling the case for $n=0$, \eqref{sloop1_3} by 
    distributivity of composition over joins, \eqref{sloop1_4} by 
    $\beta[p,c,q]_{21} \comp \oxf{\sigma} = \oxf{\sigma}$ and unrolling of 
    the $n+1$-ary composition, \eqref{sloop1_5} by $\beta[p,c,q]_{22} \comp
    \oxf{\sigma} = 0_{I,\Sigma}$, and \eqref{sloop1_6} by the universal mapping
    property of the zero object and the fact that zero maps are units for joins.
    
    \item Case $\mathcal{D} = \dfrac{
      \bigstepp{\sigma}{p}{\sff} \quad\enspace
      \bigstepp{\sigma}{q}{\sff} \quad\enspace
      \bigstepc{\sigma}{c}{\sigma'} \quad\enspace
      \bigstepc{\sigma'}{\sloop{p}{c}{q}}{\sigma''}
    }{
      \bigstepc{\sigma}{\sloop{p}{c}{q}}{\sigma''}
    }$. \\[0.5\baselineskip]
    By induction, $\oxf{c} \comp \oxf{\sigma} = \oxf{\sigma'}$ and 
    $\oxf{\sloop{p}{c}{q}} \comp \oxf{\sigma'} = \oxf{\sigma''}$, and since 
    ${\bigstepp{\sigma}{p}{\sff}}$ and $\bigstepp{\sigma}{q}{\sff}$, it follows 
    by Lemma~\ref{lem:beta_preds} that $\beta[p,c,q]_{22} \comp \oxf{\sigma} =
    \oxf{c} \comp \oxf{\sigma}$ and $\beta[p,c,q]_{21} \comp \oxf{\sigma} =
    0_{I,\Sigma}$. We then have
    \begin{align}
      \oxf{\sloop{p}{c}{q}} \comp \oxf{\sigma} & = \left(\bigvee_{n \in \omega} 
      \beta[p,c,q]_{21} \comp \beta[p,c,q]_{22}^n\right) \comp \oxf{\sigma} 
      \label{sloop2_1}\\
      & = \left(\beta[p,c,q]_{21} \vee \bigvee_{n \in \omega} \beta[p,c,q]_{21} 
      \comp \beta[p,c,q]_{22}^{n+1}\right) \comp \oxf{\sigma} \label{sloop2_2}\\
      & = \left( \left(\beta[p,c,q]_{21} \comp \oxf{\sigma} \right) \vee 
      \bigvee_{n \in \omega} \beta[p,c,q]_{21} \comp
      \beta[p,c,q]_{22}^{n+1} \comp \oxf{\sigma}\right) \label{sloop2_3}\\
      & = 0_{I,\Sigma} \vee \bigvee_{n \in \omega} \beta[p,c,q]_{21} \comp
      \beta[p,c,q]_{22}^{n+1} \comp \oxf{\sigma} \label{sloop2_4}\\
      & = \bigvee_{n \in \omega} \beta[p,c,q]_{21} \comp
      \beta[p,c,q]_{22}^{n} \comp \beta[p,c,q]_{22} \comp \oxf{\sigma} 
      \label{sloop2_5}\\
      & = \bigvee_{n \in \omega} \beta[p,c,q]_{21} \comp
      \beta[p,c,q]_{22}^{n} \comp \oxf{c} \comp \oxf{\sigma} \label{sloop2_6}\\
      & = \bigvee_{n \in \omega} \beta[p,c,q]_{21} \comp
      \beta[p,c,q]_{22}^{n} \comp \oxf{\sigma'} \label{sloop2_7}\\
      & = \left(\bigvee_{n \in \omega} \beta[p,c,q]_{21} \comp
      \beta[p,c,q]_{22}^{n}\right) \comp \oxf{\sigma'} \label{sloop2_8}\\
      & = \oxf{\sloop{p}{c}{q}} \comp \oxf{\sigma'} = \oxf{\sigma''} 
      \label{sloop2_9}
    \end{align}
    where \eqref{sloop2_1} follows by definition of $\oxf{\sloop{p}{c}{q}}$, 
    \eqref{sloop2_2} by unrolling the case for $n=0$, \eqref{sloop2_3} by 
    distributivity of composition over joins, \eqref{sloop2_4} by
    $\beta[p,c,q]_{21} \comp \oxf{\sigma} = 0_{I,\Sigma}$, \eqref{sloop2_5} by 
    zero maps units for joins and unrolling of the $n+1$-ary composition,
    \eqref{sloop2_6} by $\beta[p,c,q]_{22} \comp \oxf{\sigma} =
    \oxf{c} \comp \oxf{\sigma}$, \eqref{sloop2_7} by $\oxf{c} \comp 
    \oxf{\sigma} = \oxf{\sigma'}$, \eqref{sloop2_8} by distributivity of 
    composition over joins, and finally \eqref{sloop2_9} by definition of 
    $\oxf{\sloop{p}{c}{q}}$ and $\oxf{\sloop{p}{c}{q}} \comp \oxf{\sigma'} = 
    \oxf{\sigma''}$.
  \end{itemize}
\end{proof}

This finally allows us to show the computational soundness theorem for commands -- and so, for programs -- in a straightforward manner.

\begin{thm}[Computational soundness]\label{thm:soundness}
  If there exists $\sigma'$ such that $\bigstepc{\sigma}{c}{\sigma'}$ then 
  $\oxf{c} \comp \oxf{\sigma}$ is total.
\end{thm}
\begin{proof}
  Suppose there exists $\sigma'$ such that $\bigstepc{\sigma}{c}{\sigma'}$.
  By Lemma~\ref{lem:preservation_cmd}, $\oxf{c} \comp \oxf{\sigma} = 
  \oxf{\sigma'}$, and since the interpretation of \emph{any} state is assumed 
  to be total, it follows that $\ridm{\oxf{c} \comp \oxf{\sigma}} = 
  \ridm{\oxf{\sigma'}} = \id_I$, which was what we wanted.
\end{proof}

With computational soundness done, we only have computational adequacy left to
prove. Adequacy is much simpler than usual in our case, as we have no higher
order data to deal with, and as such, it can be shown by plain structural
induction rather than by the assistance of logical relations. Nevertheless, we
require two technical lemmas in order to succeed.

\begin{lem}\label{lem:trace_disjoint}
  Let $f : A \oplus U \to B \oplus U$ and $s : I \to A$ be morphisms. If 
  $\Tr^U_{A,B}(f) \comp s$ is total, either 
  $\Tr^U_{A,B}(f) \comp s = f_{11} \comp s$, or there exists $n \in \omega$ 
  such that $\Tr^U_{A,B}(f) \comp s = f_{21} \comp f_{22}^n \comp f_{12} 
  \comp s$.
\end{lem}
\begin{proof}
  Since the trace is canonically constructed, it takes the form
  $$\Tr^U_{A,B}(f) = f_{11} \vee \bigvee_{n \in \omega} f_{21} \comp f_{22}^n 
  \comp f_{12}.$$
  Further, in the proof of Theorem~20 in \cite{Kaarsgaard2017}, it is shown 
  that this join not only exists but is a \emph{disjoint} join, \ie, for any 
  choice of $n \in \omega$,
  \begin{equation*}
    f_{11} \comp \ridm{f_{21} \comp f_{22}^n \comp f_{12}} = \left(f_{21} \comp
    f_{22}^n \comp f_{12}\right) \comp \ridm{f_{11}} = 0_{A,B}
  \end{equation*}
  and for all $n,m \in \omega$ with $n \neq m$,
  \begin{equation*}
    \ridm{f_{21} \comp f_{22}^n \comp f_{12}} \comp f_{21} \comp f_{22}^m
    \comp f_{12} = \ridm{f_{21} \comp f_{22}^m \comp f_{12}} \comp
    f_{21} \comp f_{22}^n \comp f_{12} = 0_{A,B}.
  \end{equation*}
  But then, 
  \begin{align*}
    \ridm{\Tr^U_{A,B}(f) \comp s} & = \ridm{\left(f_{11} \vee \bigvee_{n \in
    \omega} f_{21} \comp f_{22}^n \comp f_{12}\right) \comp s} \\
    & = \ridm{\left(f_{11} \comp s\right) \vee \bigvee_{n \in \omega}
    \left(f_{21} \comp f_{22}^n \comp f_{12} \comp s\right)} \\
    & = \ridm{f_{11} \comp s} \vee \bigvee_{n \in \omega}
    \ridm{f_{21} \comp f_{22}^n \comp f_{12} \comp s}.
  \end{align*}
  Since all of the morphisms $\ridm{f_{11} \comp s}$ and $\ridm{f_{21} \comp 
  f_{22}^n \comp f_{12} \comp s}$ for any $n \in \omega$ are restriction 
  idempotents $I \to I$, it follows for each of them that they are either equal 
  to $\id_I$ or to $0_{I,I}$. Suppose that none of these are equal to the 
  identity $\id_I$. Then they must all be $0_{I,I}$, and so
  $\ridm{\Tr^U_{A,B}(f) \comp s} = 0_{I,I} \neq \id_I$, contradicting totality. 
  On the other hand, suppose that there exists an identity among these. Then, 
  it follows by the disjointness property above that the rest must be $0_{I,I}$.
\end{proof}

With these done, we are finally ready to tackle the adequacy theorem.

\begin{thm}[Computational adequacy]\label{thm:adequacy}
  If $\oxf{c} \comp \oxf{\sigma}$ is total then there exists $\sigma'$ such 
  that ${\bigstepc{\sigma}{c}{\sigma'}}$.
\end{thm}
\begin{proof}
  By induction on the structure of $c$.
  \begin{itemize}
    \item Case $c = \sskip$. Then $\bigstepc{\sigma}{\sskip}{\sigma}$ by 
    $\dfrac{}{\bigstepc{\sigma}{\sskip}{\sigma}}$.
    \item Case $c = \sseq{c_1}{c_2}$. \\[0.5\baselineskip] In this case, 
    $\oxf{c} \comp \oxf{\sigma} = \oxf{\sseq{c_1}{c_2}} \comp \oxf{\sigma} = 
    \oxf{c_2} \comp \oxf{c_1} \comp \oxf{\sigma}$. Since this is total,
    so is $\oxf{c_1} \comp \oxf{\sigma}$ by Lemma~\ref{lem:basic_ridm}. But 
    then, by induction,
    there exists $\sigma'$ such that $\bigstepc{\sigma}{c_1}{\sigma'}$ by some
    derivation $\mathcal{D}_1$, and by Lemma~\ref{lem:preservation_cmd},
    $\oxf{c_1} \comp \oxf{\sigma} = \oxf{\sigma'}$. But then $\oxf{c_2} \comp
    \oxf{c_1} \comp \oxf{\sigma} = \oxf{c_2} \comp \oxf{\sigma'}$, so by
    induction there exists $\sigma''$ such that
    $\bigstepc{\sigma'}{c_2}{\sigma''}$ by some derivation $\mathcal{D}_2$. But 
    then $\bigstepc{\sigma}{\sseq{c_1}{c_2}}{\sigma''}$ by
    \begin{equation*}
      \dfrac{\overset{\mathcal{D}_1}{\bigstepc{\sigma}{c_1}{\sigma'}}
      \qquad\enspace
      \overset{\mathcal{D}_2}{\bigstepc{\sigma'}{c_2}{\sigma''}}}{\bigstepc{\sigma
      }{\sseq{c_1}{c_2}}{\sigma''}} \enspace.
    \end{equation*}
    \item Case $c = \sif{p}{c_1}{c_2}{q}$. \\[0.5\baselineskip] Thus, $\oxf{c}
    \comp \oxf{\sigma} = \oxf{\sif{p}{c_1}{c_2}{q}} \comp \oxf{\sigma} =
    \oxf{q}^\dag \comp (\oxf{c_1} \oplus \oxf{c_2}) \comp \oxf{p} \comp
    \oxf{\sigma}$, and since this is total, $\oxf{p} \comp \oxf{\sigma}$ is
    total as well by analogous argument to the previous case. It then follows
    by Lemma~\ref{lem:adequacy_pred} that there exists $b$ such that
    $\bigstepp{\sigma}{p}{b}$ by some derivation $\mathcal{D}_1$, and by 
    Lemma~\ref{lem:preservation_pred}, $\oxf{p} \comp \oxf{\sigma} = \oxf{b}
    \comp \oxf{\sigma}$. We have two cases depending on what $b$ is.
    
    When $b = \stt$ we have
    \begin{align*}
      \oxf{c} \comp \oxf{\sigma} & = \oxf{q}^\dag \comp (\oxf{c_1} \oplus
      \oxf{c_2}) \comp \oxf{p} \comp \oxf{\sigma} = \oxf{q}^\dag \comp
      (\oxf{c_1} \oplus \oxf{c_2}) \comp \oxf{\stt} \comp \oxf{\sigma} \\
      & = \oxf{q}^\dag \comp (\oxf{c_1} \oplus \oxf{c_2}) \comp \amalg_1 \comp
      \oxf{\sigma} = \oxf{q}^\dag \comp \amalg_1 \comp \oxf{c_1} \comp 
      \oxf{\sigma}
    \end{align*}
    by $\oxf{p} \comp \oxf{\sigma} = \oxf{\stt} \comp \oxf{\sigma} = \amalg_1
    \comp \oxf{\sigma}$ and naturality of $\amalg_1$. Since this is total,
    $\oxf{c_1} \comp \oxf{\sigma}$ must be total as well by
    Lemma~\ref{lem:basic_ridm}. But then, by induction, there exists $\sigma'$
    such that $\bigstepc{\sigma}{c_1}{\sigma'}$ by some derivation
    $\mathcal{D}_2$, and by Lemma~\ref{lem:preservation_cmd}, $\oxf{c_1} \comp
    \oxf{\sigma} = \oxf{\sigma'}$. Continuing the computation, we get
    \begin{align*}
      \oxf{c} \comp \oxf{\sigma} & = \oxf{q}^\dag \comp \amalg_1 \comp 
      \oxf{c_1} \comp \oxf{\sigma} = \oxf{q}^\dag \comp \amalg_1 \comp 
      \oxf{\sigma'} = (\amalg_1^\dag \comp \oxf{q})^\dag \comp 
      \oxf{\sigma'} = \ridm{\amalg_1^\dag \comp \oxf{q}}^\dag 
      \comp \oxf{\sigma'} \\
      & = \ridm{\amalg_1^\dag \comp \oxf{q}} \comp \oxf{\sigma'} =
      \oxf{\sigma'} \comp \ridm{\amalg_1^\dag \comp \oxf{q} \comp 
      \oxf{\sigma'}} = \oxf{\sigma'} \comp \ridm{\amalg_1^\dag \comp \oxf{q}
      \comp \oxf{\sigma'}} ~\comp \ridm{\oxf{q} \comp \oxf{\sigma'}}
    \end{align*}
    where we exploit that $\oxf{c_1} \comp \oxf{\sigma} = \oxf{\sigma'}$ as 
    well as the fact that $\amalg_1^\dag \comp \oxf{q} = \ridm{\amalg_1^\dag 
    \comp \oxf{q}}$ (by Lemma~\ref{lem:basic_ext}) and the fact that 
    restriction idempotents are their own partial inverses. But then 
    $\ridm{\oxf{q} 
    \comp \oxf{\sigma'}}$ must be total, in turn meaning that
    $\oxf{q} \comp \oxf{\sigma'}$ must be total. But then by
    Lemma~\ref{lem:adequacy_pred}, there must exist $b'$ such that
    $\bigstepp{\sigma'}{q}{b'}$ by some derivation $\mathcal{D}_3$, with
    $\oxf{q} \comp \oxf{\sigma'} = \oxf{b'} \comp \oxf{\sigma'}$ by
    Lemma~\ref{lem:preservation_pred}. Again, we have two cases depending on
    $b'$. If $b'=\stt$, we derive 
    $\bigstepc{\sigma}{\sif{p}{c_1}{c_2}{q}}{\sigma'}$ by
    \begin{equation*}
      \dfrac{\overset{\mathcal{D}_1}{\bigstepp{\sigma}{p}{\stt}} \qquad\enspace
      \overset{\mathcal{D}_2}{\bigstepc{\sigma}{c_1}{\sigma'}} \qquad\enspace
      \overset{\mathcal{D}_3}{\bigstepp{\sigma'}{q}{\stt}}}{\bigstepc{\sigma}{\sif
      {p}{c_1}{c_2}{q}}{\sigma'}}
    \end{equation*}
    On the other hand, when $b' = \sff$, we have
    \begin{align*}
      \oxf{c} \comp \oxf{\sigma} & = \oxf{\sigma'} \comp \ridm{\amalg_1^\dag
      \comp \oxf{q} \comp \oxf{\sigma'}} = \oxf{\sigma'} \comp
      \ridm{\amalg_1^\dag \comp \oxf{\sff} \comp \oxf{\sigma'}} = \oxf{\sigma'}
      \comp \ridm{\amalg_1^\dag \comp \amalg_2 \comp \oxf{\sigma'}} \\ &
      = \oxf{\sigma'} \comp \ridm{0_{\Sigma,\Sigma} \comp \oxf{\sigma'}}
      = \oxf{\sigma'} \comp \ridm{0_{I,\Sigma}}
      = \oxf{\sigma'} \comp 0_{I,I}
      = 0_{I,\Sigma}
    \end{align*}
    using $\oxf{q} \comp \oxf{\sigma'} = \oxf{\sff} \comp \oxf{\sigma'} = 
    \amalg_2 \comp \oxf{\sigma'}$ and $\amalg_1^\dag \comp \amalg_2 = 
    0_{\Sigma,\Sigma}$ by Lemma~\ref{lem:basic_disj}. But then
    $\ridm{\oxf{c} \comp \oxf{\sigma}} = 0_{I,I}$, contradicting 
    $\ridm{\oxf{c} \comp \oxf{\sigma}} = \id_I$ since $0_{I,I} \neq \id_I$ by 
    definition of a model. Thus there exists $\sigma'$ such that 
    $\bigstepc{\sigma}{\sif{p}{c_1}{c_2}{q}}{\sigma'}$.
    
    To show the case when $b = \sff$, we proceed as before. We then have
    \begin{align*}
      \oxf{c} \comp \oxf{\sigma} & = \oxf{q}^\dag \comp (\oxf{c_1} \oplus
      \oxf{c_2}) \comp \oxf{p} \comp \oxf{\sigma} = \oxf{q}^\dag \comp
      (\oxf{c_1} \oplus \oxf{c_2}) \comp \oxf{\sff} \comp \oxf{\sigma} \\
      & = \oxf{q}^\dag \comp (\oxf{c_1} \oplus \oxf{c_2}) \comp \amalg_2 \comp
      \oxf{\sigma} = \oxf{q}^\dag \comp \amalg_2 \comp \oxf{c_2} \comp 
      \oxf{\sigma}
    \end{align*}
    using the fact that $\oxf{p} \comp \oxf{\sigma} = \oxf{\sff} \comp 
    \oxf{\sigma} = \amalg_2 \comp \oxf{\sigma}$ and naturality of $\amalg_2$.
    Thus $\oxf{c_2} \comp \oxf{\sigma}$ must be total by 
    Lemma~\ref{lem:basic_ridm}, which means that by induction there exists 
    $\sigma'$ such that $\bigstepc{\sigma}{c_2}{\sigma'}$ by a derivation 
    $\mathcal{D}_2$, and by Lemma~\ref{lem:preservation_cmd}, $\oxf{c_2} \comp
    \oxf{\sigma} = \oxf{\sigma'}$. Continuing as before, we obtain now that
    \begin{align*}
      \oxf{c} \comp \oxf{\sigma} & = \oxf{q}^\dag \comp \amalg_2 \comp 
      \oxf{c_2} \comp \oxf{\sigma} = \oxf{q}^\dag \comp \amalg_2 \comp 
      \oxf{\sigma'} = (\amalg_2^\dag \comp \oxf{q})^\dag \comp 
      \oxf{\sigma'} = \ridm{\amalg_2^\dag \comp \oxf{q}}^\dag 
      \comp \oxf{\sigma'} \\
      & = \ridm{\amalg_2^\dag \comp \oxf{q}} \comp \oxf{\sigma'} =
      \oxf{\sigma'} \comp \ridm{\amalg_2^\dag \comp \oxf{q} \comp 
      \oxf{\sigma'}} = \oxf{\sigma'} \comp \ridm{\amalg_2^\dag \comp \oxf{q}
      \comp \oxf{\sigma'}} ~\comp \ridm{\oxf{q} \comp \oxf{\sigma'}}
    \end{align*}
    and so $\oxf{q} \comp \oxf{\sigma'}$ must be total in this case as well (by 
    arguments analogous to the corresponding case for $b=\stt$), so
    by Lemma~\ref{lem:adequacy_pred} there must exist $b'$ such that 
    $\bigstepp{\sigma'}{q}{b'}$ by some derivation $\mathcal{D}_3$, and $\oxf{q} 
    \comp \oxf{\sigma'} = \oxf{b'} \comp \oxf{\sigma'}$ by
    Lemma~\ref{lem:preservation_pred}. Again, we do a case analysis depending 
    on the value of $b'$.
    
    If $b' = \stt$, we have
    \begin{align*}
      \oxf{c} \comp \oxf{\sigma} & = \oxf{\sigma'} \comp \ridm{\amalg_2^\dag
      \comp \oxf{q} \comp \oxf{\sigma'}} = \oxf{\sigma'} \comp
      \ridm{\amalg_2^\dag \comp \oxf{\stt} \comp \oxf{\sigma'}} = \oxf{\sigma'}
      \comp \ridm{\amalg_2^\dag \comp \amalg_1 \comp \oxf{\sigma'}} \\ &
      = \oxf{\sigma'} \comp \ridm{0_{\Sigma,\Sigma} \comp \oxf{\sigma'}}
      = \oxf{\sigma'} \comp \ridm{0_{I,\Sigma}}
      = \oxf{\sigma'} \comp 0_{I,I}
      = 0_{I,\Sigma}
    \end{align*}
    by arguments analogous to the corresponding case where $b=\stt$.
    This contradicts the totality of $\oxf{c} \comp \oxf{\sigma}$ by 
    $\ridm{\oxf{c} \comp \oxf{\sigma}} = \ridm{0_{I,\Sigma}} = 0_{I,I} \neq 
    \id_I$, and we get by contradiction that there exists $\sigma'$ such 
    that $\bigstepc{\sigma}{\sif{p}{c_1}{c_2}{q}}{\sigma'}$.
    
    If $b' = \sff$, we derive $\bigstepc{\sigma}{\sif{p}{c_1}{c_2}{q}}{\sigma'}$ 
    by
    \begin{equation*}
      \dfrac{\overset{\mathcal{D}_1}{\bigstepp{\sigma}{p}{\sff}} \qquad\enspace
      \overset{\mathcal{D}_2}{\bigstepc{\sigma}{c_2}{\sigma'}} \qquad\enspace
      \overset{\mathcal{D}_3}{\bigstepp{\sigma'}{q}{\sff}}}{\bigstepc{\sigma}{\sif
      {p}{c_1}{c_2}{q}}{\sigma'}} \enspace.
    \end{equation*}
    
    \item Case $c = \sfrom{p}{c_1}{q}$. \\[0.5\baselineskip] In this case, 
    $\oxf{c} \comp \oxf{\sigma} = \oxf{\sfrom{p}{c_1}{q}} \comp \oxf{\sigma} = 
    \Tr^\Sigma_{\Sigma,\Sigma}(\beta[p,c_1,q]) \comp \oxf{\sigma}$. Since this 
    is total and $\oxf{\sigma} : I \to \Sigma$, it follows by 
    Lemma~\ref{lem:trace_disjoint} that either $$\oxf{\sfrom{p}{c_1}{q}} \comp 
    \oxf{\sigma} = \beta[p,c_1,q]_{11} \comp \oxf{\sigma}$$ or there exists $n 
    \in \omega$ such that $$\oxf{\sfrom{p}{c_1}{q}} \comp 
    \oxf{\sigma} = \beta[p,c_1,q]_{21} \comp \beta[p,c_1,q]_{22}^n \comp 
    \beta[p,c_1,q]_{12} \comp \oxf{\sigma}.$$
    
    If $\oxf{\sfrom{p}{c_1}{q}} \comp \oxf{\sigma} = \beta[p,c_1,q]_{11} \comp
    \oxf{\sigma}$, we have
    \begin{align*}
      \oxf{\sfrom{p}{c_1}{q}} \comp \oxf{\sigma} & = \beta[p,c_1,q]_{11} \comp
      \oxf{\sigma} = \amalg_1^\dag \comp (\id_\Sigma \oplus \oxf{c_1}) \comp
      \oxf{q} \comp \oxf{p}^\dag \comp \amalg_1 \comp \oxf{\sigma} \\
      & = \amalg_1^\dag \comp (\id_\Sigma \oplus \oxf{c_1}) \comp
      \oxf{q} \comp \ridm{\amalg_1^\dag \comp \oxf{p}} \comp \oxf{\sigma} \\
      & = \amalg_1^\dag \comp (\id_\Sigma \oplus \oxf{c_1}) \comp
      \oxf{q} \comp \oxf{\sigma} \comp \ridm{\amalg_1^\dag \comp \oxf{p} \comp 
      \oxf{\sigma}} \\
      & = \amalg_1^\dag \comp (\id_\Sigma \oplus \oxf{c_1}) \comp
      \oxf{q} \comp \oxf{\sigma} \comp \ridm{\amalg_1^\dag \comp \oxf{p} \comp 
      \oxf{\sigma}} ~\comp \ridm{\oxf{q} \comp \oxf{\sigma}} ~\comp 
      \ridm{\oxf{p} \comp \oxf{\sigma}}
    \end{align*}
    using the trick from previously that $\oxf{p}^\dag \comp \amalg_1 =
    (\oxf{p}^\dag \comp \amalg_1)^{\dag\dag} = (\amalg_1^\dag \comp
    \oxf{p})^{\dag} = \ridm{\amalg_1^\dag \comp \oxf{p}}^\dag =
    \ridm{\amalg_1^\dag \comp \oxf{p}}$, $\oxf{q} \comp \oxf{\sigma} = \oxf{q}
    \comp \oxf{\sigma} \comp \ridm{\oxf{q} \comp \oxf{\sigma}}$, the fourth
    axiom of restriction idempotents and commutativity of restriction
    idempotents to obtain this. It follows by the totality of $\oxf{c} \comp 
    \oxf{\sigma}$ that 
    $\ridm{\oxf{p} \comp \oxf{\sigma}}$ and $\ridm{\oxf{q} \comp \oxf{\sigma}}$
    must be total, so $\oxf{p} \comp \oxf{\sigma}$ and $\oxf{q} \comp 
    \oxf{\sigma}$ must be total as well. It then follows by 
    Lemma~\ref{lem:adequacy_pred} that there exist $b_1$ and $b_2$ such that 
    $\bigstepp{\sigma}{p}{b_1}$ and $\bigstepp{\sigma}{q}{b_2}$ by derivations 
    $\mathcal{D}_1$ respectively $\mathcal{D}_2$. But then it follows by 
    Lemma~\ref{lem:beta_preds} that $b_1 = b_2 = \stt$, as we would otherwise 
    have $\oxf{c} \comp \oxf{\sigma} = \beta[p,c_1,q]_{11} \comp
    \oxf{\sigma} = 0_{I,\Sigma}$, contradicting totality. Thus we may derive
    $\bigstepc{\sigma}{\sfrom{p}{c_1}{q}}{\sigma}$ by
    \begin{equation*}
      \dfrac{\overset{\mathcal{D}_1}{\bigstepp{\sigma}{p}{\stt}} \qquad\enspace
      \overset{\mathcal{D}_2}{\bigstepp{\sigma}{q}{\stt}}}
      {\bigstepc{\sigma}{\sfrom{p}{c_1}{q}}{\sigma}}.
    \end{equation*}
    On the other hand, suppose that there exists $n \in \omega$ such that
    $$\oxf{c} \comp \oxf{\sigma} = \oxf{\sfrom{p}{c_1}{q}} \comp \oxf{\sigma} =
    \beta[p,c_1,q]_{21} \comp \beta[p,c_1,q]_{22}^n \comp \beta[p,c_1,q]_{12}
    \comp \oxf{\sigma}.$$ Since this is total, by Lemma~\ref{lem:basic_ridm}
    $\beta[p,c_1,q]_{12} \comp \oxf{\sigma}$ is total as well, and
    \begin{align*}
      \beta[p,c_1,q]_{12} \comp \oxf{\sigma} & = 
      \amalg_2^\dag \comp (\id_\Sigma \oplus \oxf{c_1}) \comp \oxf{q} \comp
      \oxf{p}^\dag \comp \amalg_1 \comp \oxf{\sigma}
      = \amalg_2^\dag \comp (\id_\Sigma \oplus \oxf{c_1}) \comp \oxf{q} \comp
      \ridm{\amalg_1^\dag \comp \oxf{p}} \comp \oxf{\sigma} \\
      & = \amalg_2^\dag \comp (\id_\Sigma \oplus \oxf{c_1}) \comp \oxf{q} \comp
      \oxf{\sigma} \comp \ridm{\amalg_1^\dag \comp \oxf{p} \comp \oxf{\sigma}} 
      \\
      & = \amalg_2^\dag \comp (\id_\Sigma \oplus \oxf{c_1}) \comp \oxf{q} \comp
      \oxf{\sigma} \comp \ridm{\amalg_1^\dag \comp \oxf{p} \comp \oxf{\sigma}}
      \comp \ridm{\oxf{p} \comp \oxf{\sigma}} ~\comp \ridm{\oxf{q} \comp 
      \oxf{\sigma}} 
    \end{align*}
    which follows by arguments analogous to the previous case. But then
    $\oxf{p} \comp \oxf{\sigma}$ and $\oxf{q} \comp \oxf{\sigma}$ must be total
    as well, so by Lemma~\ref{lem:adequacy_pred} there exist $b_1$ and $b_2$
    such that $\bigstepp{\sigma}{p}{b_1}$ and $\bigstepp{\sigma}{q}{b_2}$ by
    derivations $\mathcal{D}_1$ respectively $\mathcal{D}_2$. Since
    $\beta[p,c_1,q]_{12} \comp \oxf{\sigma}$ is total, it further follows by
    Lemma~\ref{lem:beta_preds} that $b_1 = \stt$ and $b_2 = \sff$, as we would
    otherwise have $\beta[p,c_1,q]_{12} \comp \oxf{\sigma} = 0_{I,\Sigma}$.
    Further, $\beta[p,c_1,q]_{12} \comp \oxf{\sigma} = \oxf{c_1} \comp
    \oxf{\sigma}$, and since this is total, by induction there exists $\sigma'$
    such that $\bigstepc{\sigma}{c_1}{\sigma'}$ by some derivation
    $\mathcal{D}_3$, with $\oxf{c_1} \comp \oxf{\sigma} = \oxf{\sigma'}$ by
    Lemma~\ref{lem:preservation_cmd}.
    
    To summarize, we have now that
    \begin{align*}
      \oxf{\sfrom{p}{c_1}{q}} \comp \oxf{\sigma} & =
      \beta[p,c_1,q]_{21} \comp \beta[p,c_1,q]_{22}^n \comp \beta[p,c_1,q]_{12}
      \comp \oxf{\sigma} \\&=
      \beta[p,c_1,q]_{21} \comp \beta[p,c_1,q]_{22}^n \comp \oxf{c_1}
      \comp \oxf{\sigma} \\&=
      \beta[p,c_1,q]_{21} \comp \beta[p,c_1,q]_{22}^n \oxf{\sigma'}
    \end{align*}
    is total, and we have derivations $\mathcal{D}_1$ of 
    $\bigstepp{\sigma}{p}{\stt}$, $\mathcal{D}_2$ of $\bigstepp{\sigma}{q}{\sff}$,
    and $\mathcal{D}_3$ of $\bigstepc{\sigma}{c_1}{\sigma'}$. To finish the 
    proof, we show by induction on $n$ that if $\beta[p,c_1,q]_{21} \comp 
    \beta[p,c_1,q]_{22}^n \oxf{\sigma'}$ is total for any state $\sigma'$, 
    there exists a state $\sigma''$ such that 
    $\bigstepc{\sigma'}{\sloop{p}{c_1}{q}}{\sigma''}$ by some derivation 
    $\mathcal{D}_4$. \vspace{0.5\baselineskip}
    
    \begin{itemize}
      \item In the base case $n = 0$, so $\beta[p,c_1,q]_{21} \comp 
      \beta[p,c_1,q]_{22}^n 
      \oxf{\sigma'} = \beta[p,c_1,q]_{21} \comp \oxf{\sigma'}$. By proof 
      analogous to previous cases, we have that $\oxf{p} \comp \oxf{\sigma'}$
      and $\oxf{q} \comp \oxf{\sigma'}$ must be total, so by 
      Lemma~\ref{lem:adequacy_pred} there exist $b'_1$ and $b'_2$ such that 
      $\bigstepp{\sigma'}{p}{b'_1}$ and $\bigstepp{\sigma'}{q}{b'_1}$ by 
      derivations $\mathcal{D}'_1$ respectively $\mathcal{D}'_2$. Further, by 
      totality, it follows by Lemma~\ref{lem:beta_preds} that we must have 
      $b_1' = \sff$, $b_2' = \stt$. Thus we may produce our derivation 
      $\mathcal{D}_4$ of $\bigstepc{\sigma'}{\sloop{p}{c_1}{q}}{\sigma'}$ by
      \begin{equation*}
        \dfrac{
          \overset{\mathcal{D}_1'}{\bigstepp{\sigma'}{p}{\sff}} \quad\enspace
          \overset{\mathcal{D}_2'}{\bigstepp{\sigma'}{q}{\stt}}
        }{
          \bigstepc{\sigma'}{\sloop{p}{c}{q}}{\sigma'}
        }
      \end{equation*}
      \item In the inductive case, we have that
      \begin{align*}
        \beta[p,c_1,q]_{21} \comp \beta[p,c_1,q]_{22}^{n+1} \oxf{\sigma'} =
        \beta[p,c_1,q]_{21} \comp \beta[p,c_1,q]_{22}^n \comp 
        \beta[p,c_1,q]_{22} \comp \oxf{\sigma'}
      \end{align*}
      and since this is assumed to be total, it follows by 
      Lemma~\ref{lem:basic_ridm} that $\beta[p,c_1,q]_{22} \comp \oxf{\sigma'}$
      is total as well. Again, by argument analogous to previous cases, this
      implies that $\oxf{p} \comp \oxf{\sigma'}$ and $\oxf{q} \comp 
      \oxf{\sigma'}$ are both total, so by Lemma~\ref{lem:adequacy_pred} there 
      exist $b_1'$ and $b_2'$ such that $\bigstepp{\sigma'}{p}{b_1'}$ and
      $\bigstepp{\sigma'}{q}{b_2'}$ by derivations $\mathcal{D}'_1$ respectively
      $\mathcal{D}'_2$. Likewise, it then follows by Lemma~\ref{lem:beta_preds}
      that since this is total, we must have $b_1' = b_2' = \sff$, and so
      $\beta[p,c_1,q]_{22} \comp \oxf{\sigma'} = \oxf{c_1} \comp \oxf{\sigma'}$,
      again by Lemma~\ref{lem:beta_preds}. Since $\oxf{c_1} \comp \oxf{\sigma'}$
      is total, by the outer induction hypothesis there exists a derivation 
      $\mathcal{D}_3'$ of $\bigstepc{\sigma'}{c_1}{\sigma''}$ for some 
      $\sigma''$, and so $\oxf{c_1} \comp \oxf{\sigma'} = \oxf{\sigma''}$ by 
      Lemma~\ref{lem:preservation_cmd}. But then
      \begin{align*}
        \beta[p,c_1,q]_{21} \comp \beta[p,c_1,q]_{22}^{n+1} \oxf{\sigma'} & =
        \beta[p,c_1,q]_{21} \comp \beta[p,c_1,q]_{22}^n \comp 
        \beta[p,c_1,q]_{22} \comp \oxf{\sigma'} \\ 
        & = \beta[p,c_1,q]_{21} \comp \beta[p,c_1,q]_{22}^n \comp \oxf{c_1} 
        \comp \oxf{\sigma'} \\
        & = \beta[p,c_1,q]_{21} \comp \beta[p,c_1,q]_{22}^n \comp \oxf{\sigma''}
      \end{align*}
      since this is total, by (inner) induction there exists a derivation 
      $\mathcal{D}_4'$ of $\bigstepc{\sigma''}{\sloop{p}{c_1}{q}}{\sigma'''}$.
      Thus, we may produce our derivation $\mathcal{D}_4$ as
      \begin{equation*}
        \dfrac{
          \overset{\mathcal{D}'_1}{\bigstepp{\sigma'}{p}{\sff}} \quad\enspace
          \overset{\mathcal{D}'_2}{\bigstepp{\sigma'}{q}{\sff}} \quad\enspace
          \overset{\mathcal{D}'_3}{\bigstepc{\sigma'}{c_1}{\sigma''}} 
          \quad\enspace
          \overset{\mathcal{D}'_4}{\bigstepc{\sigma''} 
          {\sloop{p}{c}{q}}{\sigma'''}}
        }{
          \bigstepc{\sigma'}{\sloop{p}{c_1}{q}}{\sigma'''}
        }
      \end{equation*}
      concluding the internal lemma.
    \end{itemize} 
    \vspace{0.5\baselineskip}
    
    \noindent
    Since this is the case, we may finally show 
    $\bigstepc{\sigma}{\sfrom{p}{q}{c}}{\sigma'}$ by
    \begin{equation*}
      \dfrac{
        \overset{\mathcal{D}_1}{\bigstepp{\sigma}{p}{\stt}} \quad\enspace
        \overset{\mathcal{D}_2}{\bigstepp{\sigma}{q}{\sff}} \quad\enspace
        \overset{\mathcal{D}_3}{\bigstepc{\sigma}{c}{\sigma'}} \quad\enspace
        \overset{\mathcal{D}_4}{\bigstepc{\sigma'}{\sloop{p}{c}{q}}{\sigma''}}
      }{
        \bigstepc{\sigma}{\sfrom{p}{c}{q}}{\sigma''}
      },
    \end{equation*}
    which concludes the proof. \qedhere
  \end{itemize}
\end{proof}

\section{Full abstraction} 
\label{sec:full_abstraction}
Where computational soundness and adequacy state an agreement in the notions of
convergence between the operational and denotational semantics, full
abstraction deals with their respective notions of equivalence (see, \eg,
\cite{Stoughton1988,Plotkin1999}). Unlike the case for computational soundness
and computational adequacy, where defining a proper notion of convergence
required more work on the categorical side, the tables have turned when it
comes to program equivalence. In the categorical semantics, program equivalence
is clear -- equality of interpretations. Operationally, however, there is
nothing immediately corresponding to equality of behaviour at runtime.

To produce this, we consider how two programs may behave when executed from the
same start state. If they always produce the same result, we say that they are operationally equivalent. Formally, we define this as follows:

\begin{defi}
  Say that programs $p_1$ and $p_2$ are \emph{operationally equivalent}, 
  denoted $p_1 \approx p_2$, if for all states $\sigma$, 
  $\bigstepc{\sigma}{p_1}{\sigma'}$ if and only if 
  $\bigstepc{\sigma}{p_2}{\sigma'}$.
\end{defi}

A model is said to be \emph{equationally fully abstract} (see, \eg,
\cite{Stoughton1988}) if these two notions of program equivalence are in
agreement. In the present section, we will show a sufficient condition for
equational full abstraction of models of structured reversible flowchart
languages. This condition will be that the given model additionally has the
properties of being $I$-\emph{well pointed} and \emph{bijective on states}.

\begin{defi}
  Say that a model of a structured reversible flowchart
  language is $I$-\emph{well pointed} if, for all parallel morphisms $f, g : A 
  \to B$,  $f = g$ precisely when $f \comp p = g \comp p$ for all $p : I \to A$.
\end{defi}

\begin{defi}
  Say that a model \C{} of a structured reversible flowchart
  language $\mathcal{L}$ is \emph{bijective on states} if there is a 
  bijective correspondence between states of $\mathcal{L}$ and \emph{total} 
  morphisms $I \to \Sigma$ of \C.
\end{defi}

If a computationally sound and adequate model of a structured reversible
flowchart language is bijective on states, we can show a stronger version of
Lemma~\ref{lem:preservation_cmd}.

\begin{lem}\label{lem:adequacy_inv}
  If \C{} be a sound and adequate model of a structured reversible flowchart
  language which is bijective on states. Then $\bigstepc{\sigma}{c}{\sigma'}$ 
  iff $\oxf{c} \comp \oxf{\sigma} = \oxf{\sigma'}$ and this is total.
\end{lem}
\begin{proof}
  By Lemma \ref{lem:preservation_cmd} and Theorem \ref{thm:soundness}, we only 
  need to show that $\oxf{c} \comp \oxf{\sigma} = \oxf{\sigma'}$ implies 
  $\bigstepc{\sigma}{c}{\sigma'}$. Assume that $\oxf{c} \comp \oxf{\sigma} = 
  \oxf{\sigma'}$ and that this is total. By Theorem~\ref{thm:adequacy}, there 
  exists $\sigma''$ such that $\bigstepc{\sigma}{c}{\sigma''}$, so by 
  Lemma~\ref{lem:preservation_cmd}, $\oxf{c} \comp \oxf{\sigma} = 
  \oxf{\sigma''}$. Thus $\oxf{\sigma'} = \oxf{c} \comp \oxf{\sigma} = 
  \oxf{\sigma''}$, so $\sigma' = \sigma''$ by bijectivity on states.
\end{proof}

With this, we can show equational full abstraction.

\begin{thm}[Equational full abstraction]
  Let \C{} be a sound and adequate model of a structured reversible flowchart
  language that is furthermore $I$-well pointed and bijective on states. Then
  \C{} is equationally fully abstract, \ie, $p_1 \approx p_2$ if and only if 
  $\oxf{p_1} = \oxf{p_2}$.
\end{thm}
\begin{proof}
  Suppose $p_1 \approx p_2$, \ie, for all $\sigma$,
  $\bigstepc{\sigma}{p_1}{\sigma'}$ if and only if
  $\bigstepc{\sigma}{p_2}{\sigma'}$. Let $s : I \to \Sigma$ be a morphism. Since
  \C{} is a model of a structured reversible flowchart language, it follows 
  that either $\ridm{s} = \id_I$ or $\ridm{s} = 0_{I,I}$ where $\id_I \neq 
  0_{I,I}$.
  
  If $\ridm{s} = 0_{I,I}$, we have $\oxf{p_1} \comp s = \oxf{p_2} \comp s = 
  0_{I,\Sigma}$ by unicity of the zero map. 
  
  On the other hand, if $\ridm{s} = \id_I$, by bijectivity on states there
  exists $\sigma_0$ such that $s = \oxf{\sigma_0}$. Consider now $\oxf{p_1} 
  \comp s = \oxf{p_1} \comp \oxf{\sigma_0}$. If this is total, by 
  Theorem~\ref{thm:adequacy} there exists $\sigma_0'$ such that 
  $\bigstepc{\sigma_0}{p_1}{\sigma_0'}$, and by $p_1 \approx p_2$,
  $\bigstepc{\sigma_0}{p_2}{\sigma_0'}$ as well. But then, applying 
  Lemma~\ref{lem:preservation_cmd} on both yields that 
  $$
  \oxf{p_1} \comp s = 
  \oxf{p_1} \comp \oxf{\sigma_0} = \oxf{\sigma_0'} = 
  \oxf{p_2} \comp \oxf{\sigma_0} =
  \oxf{p_2} \comp s \enspace.
  $$
  If, on the other hand, $\oxf{p_1} \comp s$ is not total, by the
  contrapositive to Theorem~\ref{thm:soundness}, there exists \emph{no}
  $\sigma_0'$ such that $\bigstepc{\sigma_0}{p_1}{\sigma_0'}$, so by $p_1
  \approx p_2$ there exists no $\sigma_0''$ such that
  $\bigstepc{\sigma_0}{p_1}{\sigma_0''}$. But then, by the contrapositive 
  of Theorem~\ref{thm:adequacy} and the fact that restriction idempotents on 
  $I$ are either $\id_I$ or $0_{I,I}$, it follows that 
  $$\oxf{p_1} \comp s = \oxf{p_1} \comp \oxf{\sigma_0} = 0_{I,\Sigma} = 
  \oxf{p_2} \comp \oxf{\sigma_0} = \oxf{p_2} \comp s.$$
  Since $s$ was chosen arbitrarily and $\oxf{p_1} \comp s = \oxf{p_2} \comp s$ 
  in all cases, it follows by $I$-well pointedness that $\oxf{p_1} = \oxf{p_2}$.
  
  In the other direction, suppose $\oxf{p_1} = \oxf{p_2}$, let $\sigma_0$ 
  be a state, and suppose that there exists $\sigma_0'$ such that 
  $\bigstepc{\sigma_0}{p_1}{\sigma_0'}$. By Lemma~\ref{lem:preservation_cmd}, 
  $\oxf{p_1} \comp \oxf{\sigma_0} = \oxf{\sigma_0'}$, and by 
  Theorem~\ref{thm:soundness} this is total. But then, by $\oxf{p_1} = 
  \oxf{p_2}$, $\oxf{p_2} \comp \oxf{\sigma_0} = \oxf{p_1} \comp \oxf{\sigma_0} 
  = \oxf{\sigma_0'}$, so by Lemma~\ref{lem:adequacy_inv}, 
  $\bigstepc{\sigma_0}{p_2}{\sigma_0'}$. The other direction follows similarly.
\end{proof}


\section{Applications} 
\label{sec:applications}
In this section, we briefly cover some applications of the developed theory: We show how the usual program invertion rules can be verified using the semantics;
introduce a small reversible flowchart language, and use the results from the
previous sections to give it semantics; and discuss how decisions may be used
as a programming technique to naturally represent predicates in a reversible
functional language.

\subsection{Verifying program inversion} 
\label{sub:extracting_a_program_inverter}
A desirable syntactic property for reversible programming languages is to be
closed under program inversion, in the sense that for each program $p$, there
is another program $\I[p]$ such that $\oxf{\I[p]} = \oxf{p}^\dag$. Janus,
\texttt{R-WHILE}, and \texttt{R-CORE}~\cite{YoGl07a,GlYo16,GlYo17} are all
examples of reversible programming languages with this property. This is
typically witnessed by a \emph{program inverter}~\cite{AbrGlu00:URA,AbrGlu02:URAJ}, that is, a procedure
mapping the program text of a program to the program text of its inverse
program\footnote{While semantic inverses \emph{are} unique, their program texts
generally are not. As such, a programming language may have many different
sound and complete program inverters, though they will all be equivalent up to program semantics.}.

Program inversion of irreversible languages is inherently difficult because of their backwards nondeterminism and the elimination of nondeterminism requires advanced machinery (\eg, LR-based parsing methods~\cite{GlueckKawabe:05:LRJinv}). Constructing program inversion for reversible languages is typically straightforward because of their backward determinism, but exploting this object language property, \eg, for proving program inverters correct, is severly hindered by the fact that conventional metalanguages do not naturally capture these object-language properties. On the other hand, the categorical foundation considered here leads to a very concise theorem regarding program inversion for all reversible flowchart languages.

Suppose that we are given a language where elementary operations are closed
under program inversion (\ie, where each elementary operation $b$ has an
inverse $\I[b]$ such that $\oxf{\I[b]} = \oxf{b}^\dag$). We can use the
semantics to verify the usual program inversion rules for $\sskip$, sequencing,
reversible conditionals and loops as follows, by structural induction on $c$
with the hypothesis that $\oxf{\I[c]} = \oxf{c}^\dag$. For $\sskip$, we have
\begin{equation*}
  \oxf{\sskip}^\dag = \id_\Sigma^\dag = \id_\Sigma = \oxf{\sskip}
\end{equation*}
thus justifying the inversion rule $\I[\sskip] = \sskip$. Likewise for sequences,
\begin{equation*}
  \oxf{\sseq{c_1}{c_2}}^\dag = (\oxf{c_2} \comp \oxf{c_1})^\dag = 
  \oxf{c_1}^\dag \comp \oxf{c_2}^\dag = \oxf{\I[c_1]} \comp \oxf{\I[c_2]} =
  \oxf{\sseq{\I[c_2]}{\I[c_1]}}
\end{equation*}
which verifies the inversion rule $$\I[\sseq{c_1}{c_2}] = \sseq{\I[c_2]}{\I[c_1]}\enspace.$$

Our approach becomes more interesting when we come to conditionals. Given some
conditional statement $\sif{p}{c_1}{c_2}{q}$, we notice that
\begin{align*}
  \oxf{\sif{p}{c_1}{c_2}{q}}^\dag 
  & = (\oxf{q}^\dag \comp (\oxf{c_1} \oplus \oxf{c_2}) \comp \oxf{p})^\dag \\
  & = \oxf{p}^\dag \comp (\oxf{c_1}^\dag \oplus \oxf{c_2}^\dag) \comp
  \oxf{q}^{\dag\dag} \\
  & = \oxf{p}^\dag \comp (\oxf{c_1}^\dag \oplus \oxf{c_2}^\dag) \comp
  \oxf{q} \\
  & = \oxf{p}^\dag \comp (\oxf{\I[c_1]} \oplus \oxf{\I[c_2]}) \comp
  \oxf{q} \\
  & = \oxf{\sif{q}{\I[c_1]}{\I[c_2]}{p}}
\end{align*}
which verifies the correctness of usual the inversion rule (see, \eg, \cite{GlYo16,GlYo17})
\begin{equation*}
  \I[\sif{p}{c_1}{c_2}{q}] = \sif{q}{\I[c_1]}{\I[c_2]}{p}.
\end{equation*}
Finally, for reversible loops, we have
\begin{align*}
  \oxf{\sfrom{p}{c}{q}}^\dag 
  & = \Tr_{\Sigma,\Sigma}^\Sigma((\id_\Sigma \oplus \oxf{c}) \comp \oxf{q} \comp
  \oxf{p}^\dag)^\dag \\
  & = \Tr_{\Sigma,\Sigma}^\Sigma(((\id_\Sigma \oplus \oxf{c}) \comp \oxf{q} \comp
  \oxf{p}^\dag)^\dag) \\
  & = \Tr_{\Sigma,\Sigma}^\Sigma(\oxf{p} \comp \oxf{q}^\dag \comp (\id_\Sigma 
  \oplus \oxf{c}^\dag)) \\
  & = \Tr_{\Sigma,\Sigma}^\Sigma((\id_\Sigma \oplus \oxf{c}^\dag) \comp \oxf{p}
  \comp \oxf{q}^\dag) \\
  & = \Tr_{\Sigma,\Sigma}^\Sigma((\id_\Sigma \oplus \oxf{\I[c]}) \comp \oxf{p}
  \comp \oxf{q}^\dag) \\
  & = \oxf{\sfrom{q}{\I[c]}{p}}
\end{align*}
where the fact that it is a \dag-trace allows us to move the dagger inside the
trace, and dinaturality of the trace in the second component allows us to move
$\id_\Sigma \oplus \oxf{c}^\dag$ from the very right to the very left. This
brief argument verifies the correctness of the inversion rule (see 
\cite{GlYo17})
\begin{equation*}
  \I[\sfrom{p}{c}{q}] = \sfrom{q}{\I[c]}{p}
\end{equation*}
We summarize this in the following theorem:

\begin{thm}
If a reversible structured flowchart language is syntactically closed under 
inversion of elementary operations, it is also closed under inversion of 
reversible sequencing, conditionals, and loops.
\end{thm}

\subsection{Example: A reversible flowchart language} 
\label{sub:an_example_reversible_flowchart_language}
Consider the following family of (neither particularly useful nor particularly
useless) reversible flowchart languages for reversible computing with integer
data, $\mathtt{RINT}_k$. $\mathtt{RINT}_k$ has precisely $k$ variables
available for storage, denoted $\mathtt{x_1}$ through $\mathtt{x_k}$ (of which
$\mathtt{x_1}$ is designated by convention as the input/output variable), and 
its only atomic operations are addition and subtraction of variables, as well 
as addition with a constant. Variables are used as elementary predicates,
with zero designating truth and non-zero values all designating falsehood.
For control structures we have reversible conditionals and loops,
and sequencing as usual. This gives the syntax:
\begin{align*}
  p & ::= \stt \mid \sff \mid \mathtt{x_i} \mid \sand{p}{p} \mid \snot{p} & \text{(Tests)} \\
  c & ::= \sseq{c}{c} \mid \mathtt{x_i} \mathrel{+\!\!=} \mathtt{x_j} \mid 
  \mathtt{x_i} \mathrel{-\!\!=} \mathtt{x_j} \mid 
  \mathtt{x_i} 
  \mathrel{+\!\!=} \overline{n} & \\ & \mid \sif{p}{c}{c}{p} \\
  & \mid \sfrom{p}{c}{p} & 
  \text{(Commands)}
\end{align*}
Here, $\overline{n}$ is the syntactic representation of an integer $n$. In the
cases for addition and subtraction, we impose the additional syntactic
constraints that $1 \le i \le k$, $1 \le j \le k$, and $i \neq j$, the latter
to guarantee reversibility. Subtraction by a constant is not included as it may
be derived straightforwardly from addition with a constant. A program in
$\mathtt{RINT}_k$ is then simply a command.

We may now give semantics to this language in our framework. For a concrete
model, we select the category \PInj{} of sets and partial injections, which is
a join inverse category with a join-preserving disjointness tensor (given on
objects by the disjoint union of sets), so it is extensive in the sense of
Definition~\ref{def:inv_ext} by Theorem~\ref{thm:join_inv_ext}. By our
developments previously in this section, to give a full semantics to
$\mathtt{RINT}_k$ in \PInj{}, it suffices to provide an object (\ie, a set) of
stores $\Sigma$, denotations of our three classes of elementary operations
(addition by a variable, addition by a constant, and subtraction by a variable)
as morphisms (\ie, partial injective functions) $\Sigma \to \Sigma$, and
denotations of our class of elementary predicates (here, testing whether a
variable is zero or not) as decisions $\Sigma \to \Sigma \oplus \Sigma$. These
are all shown in Figure~\ref{fig:semantics_rintk}.
\begin{figure}[t]
\centering
\small

\begin{align*}
\Sigma & = \mathbb{Z}^k \\
\oxf{\mathtt{x_i}}(a_1, \dots, a_k) & = \left\{ \begin{array}{ll}
\amalg_1(a_1, \dots, a_k) & \quad \text{if } a_i = 0 \\
\amalg_2(a_1, \dots, a_k) & \quad \text{otherwise}
\end{array} \right. \\
\oxf{\mathtt{x_i} \mathrel{+\!\!=} \mathtt{x_j}}(a_1, \dots, a_k) & =
(a_1, \dots, a_{i-1}, a_i + a_j, \dots, a_k) \\
\oxf{\mathtt{x_i} \mathrel{+\!\!=} \overline{n}}(a_1, \dots, a_k) & =
(a_1, \dots, a_{i-1}, a_i + n, \dots, a_k) \\
\oxf{\mathtt{x_i} \mathrel{-\!\!=} \mathtt{x_j}}(a_1, \dots, a_k) & =
(a_1, \dots, a_{i-1}, a_i - a_j, \dots, a_k)
\end{align*}
\vspace{-1em}
\caption{The object of stores and semantics of elementary operations and
predicates of $\mathtt{RINT}_k$ in \PInj.}
\label{fig:semantics_rintk}
\end{figure}
It is uncomplicated to show
that all of these are partial injective functions, and that the denotation of
each predicate $\oxf{\mathtt{x_i}}$ is a decision, so that this is, in fact, a
model of $\mathtt{RINT}_k$ in \PInj{}.

We can now reap the benefits in the form of a reversibility theorem for free.
Operationally, semantic reversibility of a program is taken to mean that it is
\emph{locally} forward and backward deterministic; roughly, that any execution
of any subpart of that program is performed in a way that is both forward and
backward deterministic (see, \eg, \cite{Bennett1973}). Denotationally (see \cite[Sec.~1.1]{Kaarsgaard2018}), the
semantic reversibility of a program becomes the property that all of its
meaningful subprograms (including itself) have denotations as partial
isomorphisms. Since every meaningful $\mathtt{RINT}_k$ program fragment takes its semantics in an inverse category, reversibility follows directly.
\begin{thm}[Reversibility]
  Every $\mathtt{RINT}_k$ program is semantically reversible.
\end{thm}
Further, since we can straightforwardly show that $\oxf{\mathtt{x_i}
\mathrel{+\!\!=} \mathtt{x_j}}^\dag = \oxf{\mathtt{x_i} \mathrel{-\!\!=}
\mathtt{x_j}}$ and $\oxf{\mathtt{x_i} \mathrel{+\!\!=} \overline{n}}^\dag =
\oxf{\mathtt{x_i} \mathrel{+\!\!=} \overline{-n}~\!}$, we can use the technique
from Sec.~\ref{sub:extracting_a_program_inverter} to obtain a sound and
complete program inverter.
\begin{thm}[Program inversion]
  $\mathtt{RINT}_k$ has a (sound and complete) program inverter. In
  particular, for every $\mathtt{RINT}_k$ program $p$ there exists a program
  $\I[p]$ such that $\oxf{\I[p]} = \oxf{p}^\dag$.
\end{thm}

\subsection{Decisions as a programming technique} 
\label{sub:decisions_as_a_programming_technique}
Decisions offer a solution to the awkwardness in representing predicates
reversibly, as the usual representations of predicates on $X$ as maps \<X \to
Bool\> fail to be reversible in all but the most trivial cases. On the
programming side, the reversible \emph{duplication/equality
operator} $\lfloor \cdot \rfloor$ (see \cite{GlKa03, Yokoyama2012}), defined on lists as 
\begin{align*}
  \lfloor \langle x \rangle \rfloor & = \langle x, x \rangle \\
  \lfloor \langle x, y \rangle \rfloor & = \left\{ \begin{array}{ll}
    \langle x \rangle & \text{if $x=y$} \\
    \langle x, y \rangle & \text{if $x \neq y$}
  \end{array} \right.\enspace,
\end{align*}
can be seen as a distant ancestor to this idea of predicates as decisions, in
that it provides an ad-hoc solution to the problem of checking whether two
values are equal in a reversible manner.

Decisions offer a more systematic approach: They suggest that one ought to
define Boolean values in reversible functional programming not in the usual
way, but rather by means of the polymorphic datatype
\begin{haskell*}
\hskwd{data} PBool \alpha =  True \alpha \mid  False \alpha
\end{haskell*}
storing not only the \emph{result}, but also \emph{what} was tested to begin
with. With this definition, negation on these polymorphic Booleans (\<pnot\>)
may be defined straightforwardly as shown in Figure~\ref{fig:even}. In turn,
this allows for more complex predicates to be expressed in a largely familiar
way. For example, the decision for testing whether a natural number is even
(\<peven\>) is also shown in Figure~\ref{fig:even}, with \<fmap\> given in the
straightforward way on polymorphic Booleans, \ie
\begin{haskell*}
  fmap &::& (\alpha \leftrightarrow \beta) \to (PBool \alpha \leftrightarrow 
  PBool \beta) \\
  fmap f (True x) &=& True (f x) \\
  fmap f (False x) &=& False (f x) \enspace.
\end{haskell*}
For comparison with the \<peven\> function shown in Figure~\ref{fig:even}, the
corresponding irreversible predicate is typically defined as follows, with
\<not\> the usual negation of Booleans
\begin{haskell*}
  even &::& Nat \to Bool \\
  even 0 &=& True \\
  even (n+1) &=& not (even n) \enspace.
\end{haskell*}
As such, the reversible implementation as a decision is nearly identical, the
only difference being the use of \<fmap\> in the definition of
\<peven\> to recover the input value once the branch has been decided.
\begin{figure}
\setlength{\hsmargin}{0pt}
\setlength{\columnsep}{0ex}
\begin{multicols}{2}
\begin{haskell*}
pnot &::& PBool \alpha \leftrightarrow PBool \alpha \\
pnot (True x) &=& False x \\
pnot (False x) &=& True x
\end{haskell*}
\begin{haskell*}
peven &::& Nat \leftrightarrow PBool Nat \\
peven 0 &=& True 0 \\
peven (n+1) &=& fmap (+1) (pnot (peven n))
\end{haskell*}
\end{multicols}
\vspace{-2.5em}
\caption{The definition of the $\mathit{even}$-predicate as a decision on
natural numbers.}
\label{fig:even}
\end{figure}



\section{Concluding remarks} 
\label{sec:conclusion}
In the present paper, we have built on the work on extensive restriction
categories~\cite{Cockett2002,Cockett2003,Cockett2007} to derive a related concept of extensivity for inverse categories.
We have used this concept to give a novel reversible representation of
predicates and their corresponding assertions in (specifically extensive) join
inverse categories with a disjointness tensor, and in turn used these to model
the fundamental control structures of reversible loops and conditionals in
structured reversible flowchart languages. We have shown that these categorical
semantics are computationally sound and adequate with respect to the
operational semantics, and given a sufficient condition for equational full
abstraction.

Further, this approach also allowed us to derive a program inversion theorem for
structured reversible flowchart languages, and we illustrated our approach by
developing a family of structured reversible flowchart languages and using our
framework to give it denotational semantics, with theorems regarding
reversibility and program inversion for free.

The idea to represent predicates by decisions was inspired by the
\emph{instruments} associated with predicates in Effectus
theory~\cite{Jacobs2015}. Given that \emph{side
effect free} instruments $\iota$ in Effectus theory satisfy a similar rule as decisions in extensive restriction categories, namely $\codiag \comp \iota =
\id$,
and that Boolean effecti are extensive, it
could be interesting to explore the connections between extensive restriction
categories and Boolean effecti, especially as regards their internal logic.

Finally, on the programming language side, it could be interesting to further 
explore how decisions can be used 
in reversible 
programming, \eg, to do the heavy lifting involved in pattern matching and branch joining.
As our focus has been
on the representation of predicates, 
our approach may be easily adapted to other reversible flowchart structures,
\eg, Janus-style loops~\cite{YoGl07a}.

\bibliographystyle{entcs}
\bibliography{mfps,ref-used}

\end{document}